\documentclass[a4paper]{article}     %{lipics-v2019}

\usepackage{comment}
\usepackage{xspace} 
\usepackage{xfrac} %
\usepackage{mathtools}
\usepackage{fullpage}
\usepackage{gastex}
\usepackage{subfig}
\usepackage{color}
\usepackage{enumitem}

\def\o{\texttt{1}}
\def\t{\texttt{2}}
\def\th{\texttt{3}}
\definecolor{gray50}{gray}{0.5}
\def\gray50#1{\textcolor{gray50}{#1}}

 \usepackage{amssymb}
 \usepackage{amsthm}
\usepackage{graphicx}
\usepackage{tikz}

\newcommand{\nat}{\mathbb N}
\newcommand{\tuple}[1]{\langle #1 \rangle}
\def\abs#1{\ensuremath{\lvert #1\rvert}} 

\let\epsilon\varepsilon
\let\emptyset\varnothing
\newcommand{\straa}{s} 
 
\newcommand{\val}{\mathit{val}}

\def\mynote#1{}
\def\mynewnote#1{}
\newcommand*{\quot}[2]{\sfrac{#1}{#2}}

\newcommand{\A}{\mathcal{A}}
\newcommand{\F}{\mathcal{F}}
\newcommand{\G}{\mathcal{G}}
\renewcommand{\H}{\mathcal{H}}

\newcommand{\M}{\mathcal{M}}
\renewcommand{\P}{\mathcal{P}}

\newcommand{\U}{\mathcal{U}}
\newcommand{\V}{\mathcal{V}}
\newcommand{\W}{\mathcal{W}}

\newcommand{\outcome}{\mathrm{Out}}

\newcommand*{\act}{\mathrm{act}}

\newcommand*{\sync}{\mathrm{sync}}
\newcommand*{\View}{\mathrm{View}}
\newcommand{\prio}{\mathrm{prio}}

\renewcommand{\a}{a}
\renewcommand{\b}{b}

\newcommand{\bi}[2]{{}^{#1}_{#2}}
\newcommand{\tbi}[2]{\bi{#1}{#2}}   % text
\newcommand{\fbi}[2]{\mbox{{\large $\textstyle \bi{#1}{#2}$}}}   % figure
\newcommand{\cfbi}[2]{\fbi{#1}{#2}}  % compact
\newcommand{\reject}{q_{\mathrm{rej}}}

\def\presuper#1#2%
  {\mathop{}%
   \mathopen{\vphantom{#2}}^{#1}%
   \kern-\scriptspace%
   #2}

\newcommand*{\dfa}{\textsc{dfa}\xspace}

\newcommand*{\init}{\varepsilon}

\newtheorem{theorem}{Theorem}

\newtheorem{lemma}{Lemma}
\newtheorem{corollary}[lemma]{Corollary}

\newtheorem{remark}{Remark}

\newcommand{\clift}{\nabla}

\newcommand{\lift}{\mathrm{lift}}
\newcommand{\up}[1]{#1^{\uparrow}}

\newcommand{\Com}{\mathrm{Com}}
\newcommand{\rej}{\mathrm{rej}}

\newenvironment{longversion}{}{}
\newenvironment{shortversion}{}{}

\title{Synthesising Full-Information Protocols}
\author{Dietmar Berwanger \and Laurent Doyen \and Thomas Soullard}
\date{LMF, ENS Paris-Saclay \& CNRS}

\begin{document}
%\linenumbers
%\excludecomment{proof}
% \excludecomment{longversion}
\excludecomment{shortversion}

\maketitle

\begin{abstract}
  We lay out a model of games with imperfect information that features
  explicit communication actions, by which the 
  entire observation history of a player is revealed to another player.
  Such full-information
  protocols are common in asynchronous distributed systems; here, we
  consider a synchronous setting with a single active player who may
  communicate with multiple passive observers in an indeterminate
  environment. We present a procedure for solving the basic
  strategy-synthesis problem under regular winning conditions.  

  We present our solution in an abstract framework of games with 
  imperfect information and we split the proof in two conceptual 
  parts: $(i)$ a generic reduction schema from imperfect-information
  to perfect-information games, and $(ii)$ a specific construction 
  for full-information protocols that satisfies the requirement of the
  reduction schema.

  Furthermore we show that the number of passive observers induces
  a strict hierarchy, both in terms of expressiveness and complexity:
  with $n$ observers, a full-information protocol can express indistinguishability relations
  (defining imperfect information for the player in the protocol) 
  that are not expressible with $n-1$ observers, and 
  the strategy-synthesis problem is $(n+1)$-EXPTIME-complete. 
\end{abstract}

\section{Introduction}

One fundamental paradigm for the analysis of complex systems is that of reactive processes
proposed by Harel and Pnueli~\cite{HarelPnu85}.
A reactive process is one that interacts perpetually with its environment:
at every stage of the execution, it observes an input signal and then responds with a control action
towards the purpose of enforcing that the global system runs successfully with respect to a specified objective.
In contrast to programs that evaluate a function on a given input and then terminate,
reactive processes are intended to run forever.
The ongoing interaction is modelled naturally as a
game played over infinitely many stages between a strategic player representing the process, 
which seeks to satisfy the objective, and a non-strategic opponent, Nature,
which chooses the moves of the environment. 
The task of designing a reactive process that enforces a specified objective translates into the problem of constructing a winning strategy in such an infinite game \cite{BL69,Thomas95}.
That is, a function that maps the information acquired by the process player
to actions, such that every possible run of the global system satisfies the specification, regardless of the moves of Nature.

Strategies are based on the information available to the player.
In the particular situation, where the sequence of inputs signals together with the output actions
determine the run of the global system completely, the game is of perfect information.
However, reactive systems often involve events that are not directly observable to the process,
such that one input-output sequence viewed by the process 
may correspond to multiple possible global runs.
Based on its local, partial view, the process thus needs to choose its action in a way to account
for all contingencies of the global run.
We are therefore in the
setting of infinite games with imperfect information.

The original framework of reactive-system synthesis,
detailed by Pnueli and Rosner~\cite{PnueliRos89}, concerns single-process architectures described
by a sequential finite-state machine, which represents the global system together
together with the objective of the process. 
Every process action is associated with a set of state transitions that it enables,
and each transition yields
an observation from a finite alphabet.
The objective is described by a colouring of states,
a run is winning if the sequence of visited colors infinitely often maps to a specified set of colours.
In this framework, the synthesis problem is to decide whether there exists a
process strategy that is winning with respect to the objective and, if possible,
to construct a finite-state machine that implements it. 

Indeed, the synthesis problem can be solved effectively for single-process architectures
in the basic finite-state framework~\cite{Chu62,BL69,Rabin72}. 
Under perfect information, the task reduces to solving parity games
between two strictly conflicting players,
a problem that has been
well studied with good algorithmic results~\cite{Thomas95,CaludeJKLS22}.
In the setting of imperfect information, the corresponding games can be solved
via a power-set construction that goes back to Reif \cite{Rei84}.
For any observation history, consider the set of possible play histories,
and map it to the set of their end nodes. 
This map respects update operations associated to receiving a new observation.
Its image describes a game of perfect information that is equivalent to the game with imperfect information at the outset in a strong sense:
every strategy in the image corresponds to one in the original game and vice versa,
such that the two strategies have the same outcomes in terms of observation sequences.
In this way, one obtains a game that is exponentially larger, but with perfect information,
and the strategies transfer back and forth between the games preserving their outcomes,
and in particular their winning status.

A main challenge in the analysis of complex system is, however, that they are distributed.
In practice, the global system involves multiple processes,
each receiving its own observations and executing actions based on its local view.
Here, the synthesis problem concerns coordination strategies for a coalition of players with
a common objective: a solution consists in a strategy profile, a list of strategies one for
each player, which, if played simultaneously, enforces that the global run satisfies the
objective.
Typically one process has no direct access to the observation received by other processes,
we are thus in the setting of infinite coordination games with imperfect information. 

Unfortunately, the distributed variant of the synthesis problem is
algorithmically unsolvable, in the general context of infinite games with imperfect information~\cite{PR90,KupfermanVar01,FinkbeinerS05}.
Already for two processes that receive separate input sequences from the environment,
it is undecidable whether a coordination strategy exists to enforce a common goal specified by a
finite-state automaton~\cite{Schewe2014}.
As the environment of an individual process now includes the other processes of the
coalition, we are no longer in an antagonistic setting.
To coordinate successfully, however, the strategy of each player may
need to keep track of the information held by the other players,
and this is algorithmically hard to manage over an infinite duration~\cite{BerwangerKP11}.

In contrast to single-processor architectures, where actions are chosen with the purpose of controlling
the global system, the design of multi-processor architectures has an important focus on the communication
between processes.
Rather than producing outputs or actions relevant for an external observer,
a process may just convey information, derived from the observation of its own
input sequence, to another process; this could greatly help solving the control problem for the global system. 
In the basic model, such a message-passing event can be modelled as a side effect of an action chosen by the sender
that triggers a particular (global) transition, which in turn emits a particular observation to the receiving process.
This corresponds to modelling a communication channel of fixed bandwidth, bounded by the number of observations.
As the amount of information acquired by a process increases along a run whereas the channel can only convey messages
from a fixed finite range in one round, communication in this model involves a strategic choice.
Therefore, the task of designing a suitable communication strategy 
brings us back to the synthesis problem for games
between multiple decision-makers with imperfect information, known to be undecidable.

In an attempt to push the undecidability frontier for the synthesis of
distributed reactive systems, we explore a model where communication between
processes is not restricted beforehand.

We propose a model of \emph{full-information protocols (FIP)} that draws its name and the basic idea from the a well-known
concept in distributed computing: whenever a communication event between two 
processes occurs, all the information that the sender holds is conveyed to the receiver~\cite{PeaseSL80,DworkMos90,WooLam94}.
The occurrence of such an event is not necessarily controllable by the processes.
For instance, Nature may choose to keep a communication link between two processes down
over an arbitrary amount of time so that no information is transmitted.
However, when the link is re-established,
one process receives instantly the entire sequence of inputs observed by the process at the other end.
In particular, the model makes no assumption on the bandwidth of communication channels.
Another crucial feature is that communication is passive: processes do not have the choice
to reveal only a part of their information.
The meaningful choices thus concern either control actions or the triggering of communication events --
albeit with no control on the contents of the message.
Essentially, the model captures a setting of maximal information that can be
conveyed in a system where the availability of communication channels is subject
to interactive control.
Whenever a synthesis task can be solved with a communication strategy on an architecture
with arbitrarily high bandwidth, it is solvable in full-information protocols.

The principle of conveying maximal information with every communication event
also plays a key role in the model of asynchronous systems 
interpreted over Mazurkiewicz traces.
The corresponding notion of causal memory built into the fundamental
model of Zielonka automata \cite{Zie87}, proved instrumental for 
solving the synthesis problem for several classes of 
architectures~\cite{GLZ04,MTY05,GGMW13} and for Petri games~\cite{FO17}. 

We focus on the synchronous setting and
model distributed systems with FIP semantics as a repeated game played over infinitely
many stages between several players on the one side and Nature on the other side.
The players have a common objective described by a colouring of a finite-state machine. 
Every play drives the machine sequentially, by triggering  transitions that arise as an outcome of the stage
game; we call such transitions a move.
In every stage, each player chooses an action, and the profile of chosen actions determines a nonempty subset of enabled moves.
Among these, Nature chooses one. Every move is associated to a profile of observations, one for each player.
Firstly, every observation received by a player carries a local input symbol that she receives directly. 
Additionally, there is a special attribute that designates
the list of players to which she can communicate in the current round.
In consequence, she also observes the entire view of each player in this list.
Thus, the information held by a player~$i$ is encoded by its view which consists of the sequence of her own inputs and additionally,
of the sequence of inputs received by any player~$j$ up to any earlier round in which $i$ could communicate with $j$, further
the input sequence of any player $k$ with which $j$ could communicate, and so forth.
The information structure of a player in the repeated game is a tree obtained by following the possible updates of her view.
A strategy is a function that maps any view to an action.
Thus a strategy profile determines as an outcome a set of plays, that is, infinite move sequences
that correspond to infinite runs in the finite-state machine from the outset.
The objective of the players is given by an acceptance condition of this machine defined in terms of colours.
The distributed strategy profile is winning if all plays in its outcome are winning.

Our formalisation of FIP games subsumes the synchronous models of infinite games with imperfect information played on finite graphs from the
literature~\cite{AzharPetRei01,RamadgeWon87,MohalikWal03,DR11}.
Accordingly, there is no hope for solving the synthesis problem for games that involve more than one decision maker,
in the general case. We therefore restrict our attention to the case of one decision maker, corresponding to one process that chooses
actions. Additionally, there can be any number of passive players, which we call observers:
they just acquire information -- either by observing their own inputs or
by receiving the views of other player through communication events.
Observers play a crucial role, as their current view may be communicated to the decision maker
at different rounds of the infinite play, conveying an unbounded amount of information in a single stage.
The challenge for a synthesis procedure is to process this information.  
In this paper, we show that the synthesis problem is effectively solvable for FIP games with one active player
and arbitrarily many observers for a winning condition described by a finite-state automaton on infinite words.

One obstacle, even for the case of a single active player, is that information trees can be of unbounded branching.
Imagine, for instance, that the active player receives just a non-informative input symbol in each of the first 100 rounds of a game,
whereas an observer can receive any sequence of bits until the two are allowed to communicate in round $101$.
Then, the possible views of the active player will be arranged on a simple path of length $100$ and then suddenly
branch to $2^{100}$ successor views, one for each possible bit sequence received by the observer, which will henceforth be included into her
view. If the scenario continues in the same way, until the next communication event occurs in round $300$,
there will be a branching of degree $2^{200}$, and so on.
However, the automata-theoretic approach to synthesis
(\cite{Rabin72,GurevichHar82,ArnoldWal03})
relies of tree models of bounded branching degree, so we cannot expect 
classical techniques to apply directly for solving the synthesis problem.  

The basis of our approach lies in a particular notion of game equivalence which
is supported by a homomorphism that maps the original game structure to a finite image and 
satisfies a key property: its kernel commutes with the indistinguishability relation of
the active player of the game in the outset. Indeed, the composition of the kernel with the indistinguishability
relation yields an equivalence of finite index ; the quotient of the original game structure by this equivalence yields a finite
game that is bisimilar to one at the outset. As a consequence,
the winning strategies can be transferred back and forth to the original one via the homomorphism.  

The main technical contribution consists in constructing a homomorphism with the required property.
Intuitively, this is done by a bold generalisation of the powerset construction of Reif.
Besides recording the set of end states of possible histories associated to an information state,
our construction keeps track recursively of the possible records along increasing chains of coalitions
starting with the active player and up to the grand coalition.
In this way, the information sets of the original game are mapped to a bisimilar copy that is finite, but
where each node is annotated with a record of $n$-fold exponential size, where $n$ is the number of observers.
Accordingly, our solution procedure is of nonelementary complexity.   

Nevertheless, we show that the non-elementary complexity of the synthesis procedure for FIP is unavoidable.
Indeed, the acceptance problem for a Turing machine that use $n$-fold exponential space
in the length of its input reduces to the synthesis problem for a FIP game with $n$ observers.     

Nonelementary complexity is not unusual in the case of games with imperfect information involving several players.
Indeed, this bound is characteristic for games with hierarchical information in the synchronous
setting~\cite{PR90,AzharPetRei01,KupfermanVar01,FinkbeinerS05,BMvdB18},
%,BBBFS22
or for acyclic architectures in the asynchronous setting \cite{GGMW13}.
In all these cases, the synthesis problem is solvable in exponential time in the single-process case and
the complexity grows as a tower of exponentials of height $n$ with the number $n$ of active players.
Thus, the nonelementary lower bound comes as a surprise in the setting of FIP with a single decision maker. 

\section{Basic Notions}

%Automata + output (Mealy) + sequence of outputs. 2DFA. 
For a function $f: X \to Y$ and a domain subset $Z \subseteq X$, 
we denote by $f(Z)= \{f(z) \mid z \in Z \}$ the set of images of elements in $Z$.

We use finite automata as a model of acceptor of finite words,
and Mealy automata as a model of transducer. 
They share a common underlying structure of the form $\tuple{Q, \Gamma, q_\init, \delta}$,
called a \emph{semi-automaton}, described by a finite set~$Q$ of states, 
a finite input alphabet~$\Gamma$, 
a designated initial state $q_\init \in Q$, 
and a transition function $\delta: Q \times \Gamma \to Q$.
%We define the size $\abs{\A}$ of $\A$ to be the number of its transitions, that is $\abs{Q} \cdot \abs{\Gamma}$.
To extend the transition function from letters to words, we define the function $\delta: Q \times \Gamma^* \to Q$ 
by setting, for every state~$q \in Q$, by $\delta(q, \epsilon) = q$ for the empty word~$\epsilon$, 
and, recursively $\delta(q, \tau c) = \delta(\delta( q, \tau), c )$, for any word $\tau c$ 
obtained by concatenation of a word $\tau \in \Gamma^*$ and a letter $c \in \Gamma$. 
The \emph{synchronous product} of two semi-automata $\tuple{Q, \Gamma, q_\init, \delta}$
and $\tuple{P, \Gamma, p_\init, \delta'}$ is the semi-automaton 
$\tuple{Q \times P, \Gamma, (p_\init,q_\init), \Delta}$ with transition function $\Delta((q,p), c) = 
(\delta(q,c),\delta'(p,c))$ for all $q \in Q$, $p \in P$, and $c \in \Gamma$.

A \emph{deterministic finite automaton} ($\dfa$)  $\A = (Q, \Gamma, q_\init, \delta, F)$ 
expands a semi-automaton with a set~$F \subseteq Q$ of accepting states. 
A finite input word $\tau \in \Gamma^*$ is \emph{accepted} by~$\A$ %from a state $q$ 
if $\delta(q_{\init}, \tau) \in F$. 
A \emph{Mealy automaton} is described by a 
a tuple  $(Q, \Gamma, \Sigma, q_\init, \delta, \lambda)$ 
where $(Q, \Gamma, q_\init, \delta)$ is a semi-automaton, $\Sigma$ 
is a finite output alphabet, and 
$\lambda: Q \times \Gamma \to \Sigma$ is an output function.
The Mealy automaton defines a function $\lambda: \Gamma^+ \to \Sigma$
obtained by setting $\lambda(\epsilon) = \epsilon$ and $\lambda(\tau c) = \lambda(\delta(q_\init, \tau), c)$
for all words $\tau \in \Gamma^*$ and letters $c \in \Gamma$.
We say that a function on $\Gamma^*$ is \emph{regular} if there exists a Mealy automaton
that defines it.
Given an input word $\tau = c_1 c_2 \dots c_n \in \Gamma^*$,
let $\hat{\lambda}(\tau) = \lambda(c_1) \lambda(c_1 c_2) \dots \lambda(c_1 c_2 \dots c_n)$
be the output sequence consisting of the output of all prefixes of $\tau$.
We extend $\hat{\lambda}$ to infinite words $\pi = c_1 c_2 \dots \in \Gamma^{\omega}$
by setting $\hat{\lambda}(c_1 c_2 \dots) = \lambda(c_1) \lambda(c_1 c_2) \dots$ as expected.

\subsection{Repeated games with imperfect information}

Our purpose is to model reactive systems driven by occurrences of
discrete state transitions, which we call \emph{moves}. 
Towards this, we use \emph{abstract repeated games} played in infinitely many stages between a
fixed set~$I = \{1, \dots, |I|\}$ of players and Nature.
In every stage, a move is produced as an outcome of
a one-shot \emph{base game} played as follows:
each player~$i \in I$ chooses an \emph{action}~$a^i$ from her
given action set $A^i$; the chosen profile $a = (a_i)_{i \in I}$ 
constrains the set of possible outcomes to the subset of moves
supported by~$a$, from which Nature chooses one.
The outcoming move is recorded
in the play history, then the base game is repeated.
The outcome of the multistage game, called a \emph{play},
is thus an infinite sequence $\pi = c_1 c_2 \dots $ of moves.
A~\emph{history} (of length $\ell$) is a finite prefix 
$\tau = c_1 c_2 \dots c_\ell$ of a play; 
the empty history $\init$ has length zero. 
We denote by $\pi(\ell) = c_1 c_2 \ldots c_\ell$
the prefix of length $\ell$ of a play $\pi$, with $\pi(0) = \epsilon$.

\smallskip\noindent{\em Winning condition.}
The objective of a player is specified by a winning condition, a
set $W \subseteq \Gamma^{\omega}$ of plays declared to be winning.
Of special interest is the class of $\omega$-regular languages that 
extends regular languages to infinite words, and provides a robust 
specification language to express commonly used specifications~\cite{Thomas97}.

It is convenient to specify winning conditions in two parts: $(1)$ a logical specification 
$L \subseteq C^{\omega}$ over an alphabet $C$ of colors, which 
is independent of the game and its move alphabet, and $(2)$~a 
regular coloring function $\lambda: \Gamma^+ \to C$ that induces 
the winning condition $W =  \{\pi \in \Gamma^{\omega} \mid \hat{\lambda}(\pi) \in L\}$.
In this setting, the condition $L$ can be fixed and defines the type of game 
while the function $\lambda$ can be 
specified by a Mealy machine that is part of the game instance (e.g., as given 
in the input of the synthesis algorithm). 
For example parity games, which are a canonical way of representing games 
with $\omega$-regular winning conditions~\cite{Thomas97}, correspond to $C = \nat$ and 
$L = \{n_1 n_2 \ldots \in \nat^{\omega} \mid \liminf_{i \to \infty} n_i \text{ is even} \}$.
Reachability games correspond to $C = \{0,1\}$ and 
$L = \{n_1 n_2 \ldots \in \nat^{\omega} \mid \inf_{i} n_i = 0 \}$.

\smallskip\noindent{\em Imperfect information.}
To pursue their objective, players choose actions
based on the information available to them.  
The information of a player~$i \in I$ is modeled by a  
partition $U^i$ of the set
of histories; 
the parts of $U^i$ are called \emph{information sets} (of the player). 
The intended meaning is that if the actual history belongs to an 
information set, then the player considers every history in the set
possible.
The particular case where all information sets in the partition 
are singletons characterises the setting of \emph{perfect information}. 

Our model is \emph{synchronous}, which means, intuitively, 
that the players always know how many stages have been played. 
This amounts to asserting that all histories in an information 
set have the same length; 
in particular the empty history forms a singleton information set.
Further, we assume that the player has \emph{perfect recall} --- 
he never forgets what he knew previously 
and which actions he took. 
Formally, if an information set contains nontrivial histories $\tau c$ and $\tau' c'$, 
then the predecessor histories $\tau$ and $\tau'$ belong to the same information set 
and the moves~$c$ and $c'$ are supported by the same action.

An alternative representation of an information partition $U$ is given by the equivalence
relation $\sim \, \in \Gamma^* \times \, \Gamma^*$ such that $\tau \sim \tau'$
if $\tau, \tau' \in u$ for some $u \in U$. Such an equivalence is called an \emph{indistinguishability 
relation}~\cite{BD23} as it relates
the pairs of histories that the player cannot distinguish.
Formally, an indistinguishability relation $\sim \, \in \Gamma^* \times \, \Gamma^*$ is
an equivalence relation satisfying the following conditions,
for all $\tau, \tau' \in \Gamma^*$ and $c, c' \in \Gamma$:
\begin{itemize}
\item if $\tau \sim \tau'$, then $\abs{\tau} = \abs{\tau'}$ (indistinguishable histories have the same length), %and
\item if $\tau c \sim \tau' c'$, then $\tau \sim \tau'$ (the relation is prefix-closed),
\item if $\tau c \sim \tau' c'$, then $\act(c) = \act(c')$ (the action is visible).
\end{itemize}

For a history $\tau \in \Gamma^*$, we denote by 
$[\tau]_{\sim} = \{\tau' \in \Gamma^* \mid \tau' \sim \tau\}$ 
the information set containing~$\tau$.
%, and given a set $Z \subseteq \Gamma^*$ of histories, let 
%$[Z]_{\sim} = \bigcup_{\tau \in Z}[\tau]_{\sim} = 
%\{\tau' \in \Gamma^* \mid \exists \tau \in Z: \tau' \sim \tau\}$ 
%be the \emph{closure} of $Z$. 
%
%Given the history $\tau$ of moves played in the previous stages, player~$1$ 
%cannot distinguish $\tau$ from the histories $\tau' \in [\tau]_{\sim}$
%in the information set of $\tau$. 
%
Intuitively, the first condition above states that%on the equal length of indistinguishable histories 
the player knows how many rounds have been played. 
The condition of prefix-closure formalises perfect recall,
% --- he never forgets what he was able to distinguish previously --- 
and visibility of actions means that he can distinguish his own actions. 
%Indistinguishability relations can be specified using two-tape \dfa~\cite{BD20}. 
%\mynote{L: improve here.}

\smallskip\noindent{\em Restrictions: one active player, visible winning condition.}
For our analysis of the synthesis problem, we restrict to the particular case where only
one player can make relevant choices, namely Player~$0$.
Concretely, we assume that the action set of every other player $i \in I \setminus \{0\}$
is trivial $|A^i| = 1$.
Whenever we refer to a set of actions or an indistinguishability relation without specifying to which 
player it pertains, we mean Player~$0$. 
Additionally, we require the function $\lambda$ defining the color of a history
to be information-consistent, that is, constant over every information set: 
$\lambda(\tau) = \lambda(\tau')$ for all indistinguishable histories $\tau \sim \tau'$. 
We say that the induced winning condition is \emph{visible}.
%\def\mynewnote#1{L: add comment that visible is loss of generality?}
%As a consequence, visible winning conditions are unions of information sets.

\smallskip\noindent{\em Strategies.}
The following definitions concern the single active player in a repeated game. 
A \emph{decision function} is a map $f: \Gamma^* \to A$ from histories to actions. 
We say that a play $c_1 c_2 \dots$ \emph{follows}~$f$ if
$\act( c_t ) = f(c_1 \dots c_{t-1})$, for every stage $t > 0$ (and similarly for a history).
We denote by $\outcome(f)$ the set of all plays that follow $f$.
%Further, we say that an information set $u \in U$ is \emph{reachable} if there exists a history in $u$ that follows~$f$.
%A decision function~$f$ is \emph{information consistent} 
%if it is constant on all reachable information sets.

A \emph{strategy} is a decision function that is information consistent.
Given a winning condition $W \subseteq \Gamma^{\omega}$, the strategy $\straa$ 
is \emph{winning} if all plays that follow $\straa$ belong to~$W$,
that is $\outcome(\straa) \subseteq W$. When the winning condition is induced
by a logical specification $L \subseteq C^{\omega}$ (and a regular function 
$\lambda: \Gamma^+ \to C$ that is clear from the context), 
we also say that $\straa$ is \emph{winning} for $L$.

\smallskip\noindent{\em Game description.}
%The indistinguishability relation $\sim$ is specified by a two-tape \dfa
%over input alphabet $\Gamma \times \Gamma$.
Given an action set $A$, a move set $\Gamma$, and a function $\act: \Gamma \to A$, 
a game with imperfect information consists of a tuple 
$\G = \tuple{A,\Gamma,\act,\sim, \lambda}$ and a winning condition $L \subseteq C^{\omega}$,
where $\sim$ is an indistinguishability 
relation and $\lambda$ is a coloring function. 
%$\M$ is a Mealy machine with input alphabet $\Gamma$ and output alphabet $C$ that defines the function~$\lambda$. 
In the special case of perfect-information games, characterised by 
the indistinguishability relation $\sim$ being the identity (or equivalently 
by the informations sets $[\tau]_{\sim} = \{\tau\}$ being singletons 
for all $\tau \in \Gamma^*$), we omit the relation $\sim$ in the tuple~$\G$. 

\smallskip\noindent{\em Synthesis problem.}
For a fixed winning condition $L \subseteq C^{\omega}$,
the synthesis problem asks, given a game $\G$ with imperfect information,
whether there exists a winning strategy for $L$ in $\G$.
%If the answer is $\textsf{Yes}$, 
%\mynote{L: define the problem we consider ? (synthesis,game solving)}

\section{Full-Information Protocols}
In the standard model of partial-observation games~\cite{Rei84}, the indistinguishability
relation $\sim$ is induced by a regular observation function $\beta: \Gamma^* \to \Sigma$
(where $\Sigma$ is a finite set of observations),
such that $\tau \sim \tau'$ if $\hat{\beta}(\tau) = \hat{\beta}(\tau')$.
Intuitively, the player receives at every nonempty history~$\tau c$ 
the observation symbol $\beta(\tau c)$, and by the assumption of perfect recall,
remembers the sequence $\hat{\beta}(\tau)$ of previous observations.
An equivalent characterisation is
$\tau c \sim \tau'c'$ iff 
$\tau \sim \tau'$ and $\beta(\tau c) = \beta(\tau' c')$. % (with $\epsilon \sim \epsilon$). 
As a consequence, for any information set $u$, there are at most
$\abs{\Sigma}$ information sets $u'$ such that $\tau c \in u'$
for some $\tau \in u$ and $c \in \Gamma$, that is, the information
tree has bounded branching.

In a full-information protocol, the active player, namely Player~$0$, 
is accompanied by $n$ passive players, which we call observers. 
Each player $i = 0, \dots, n$ receives an observation symbol at every round, given 
by a regular observation function $\beta_i: \Gamma^* \to \Sigma$.
However, only player~$0$ is able to make strategic choices;
the other players $1, \dots, n$ have singleton actions sets. However, they may 
communicate with other observers or with the main player.
Let $I = \{0,1,\dots,n\}$ be the set of all players. 
Communication is specified by relations $R_{\sigma} \subseteq I \times I$
indexed by observations $\sigma \in \Sigma$: when Player~$i$ receives observation $\sigma$,
he also receives the entire view of all players $j \in R_{\sigma}(i) = \{j \mid (i,j) \in R_{\sigma}\}$,
which consists of all observations of players in $R_{\sigma}(i)$ as well
as (recursively) the view of players in $R_{\sigma}(i)$. 
Intuitively, a link $(i,j) \in R_{\sigma}$ specifies a one-way communication
with receiver $i$ and sender $j$ upon observation of~$\sigma$ (Player~$i$ peeks at Player~$j$). We refer to such links
as \emph{direct} links. If at some history, there is a direct link from
Player~$i$ to Player~$j$, and also a direct link from Player~$j$ to Player~$k$, then a communication
is established from Player~$i$ to Player~$k$, even if the protocol does not
specify the link from $i$ to $k$ directly. We refer to such links as \emph{indirect} links.

We represent the information available to the player and observers along a history 
$\tau = c_1 c_2 \dots c_{\ell}$ by a graph $\View(\tau) = (V,E)$,
called the \emph{view graph}, where:
\begin{itemize}
\item $V = I \times \{0,1,\dots, \ell\}$ is the set of nodes, and a node $(i,t) \in V$
represents the viewpoint of Player~$i$ after $t$ rounds;

\item $E \subseteq V \times V$ is the set of edges, where an edge $\tuple{(i,t),(j,u)}$
intuitively means that after $t$ rounds, Player~$i$ has access to the view of Player~$j$ 
at round $u$; the set $E$ contains the edges $\tuple{(i,t),(i,t-1)}$
for all $i \in I$ and $1 < t \leq \ell$, which correspond to looking into the past,
and the edges $\tuple{(i,t),(j,t)}$
for all $i,j \in I$ and $1 \leq t \leq \ell$
such that $j \in R_{\sigma}(i)$ where $\sigma = \beta_i(c_1 c_2 \dots c_{t})$, 
which correspond to communicating the view of Player~$j$ to Player~$i$ (via a direct link).
\end{itemize}

%  c_t[j]  =  \beta_j(c_1 c_2 \dots c_{t}) 
%  c_t[j]  =  \beta_j(c_1 c_2 \dots c_{t}) 

Two histories $\tau,\tau' \in \Gamma^*$
are indistinguishable for Player~$i$,
denoted $\tau \sim_i \tau'$, if $\abs{\tau} = \abs{\tau'}$ and 
$\beta_j(\tau(t)) = \beta_j(\tau'(t))$ for
all nodes $(j,t)$ reachable from $(i,\abs{\tau})$ in the view graph $\View(\tau)$.
Note that the definition implies that if $\tau \sim_i \tau'$, then 
the reachable nodes from $(i,\abs{\tau})$ in $\View(\tau)$ and 
in $\View(\tau')$ coincide.
We say that the histories $\tau, \tau'$ are indistinguishable
for a coalition $J \subseteq I$, denoted $\tau \sim_J \tau'$, 
if they are indistinguishable for all players of the coalition,
that is, $\tau \sim_i \tau'$ for all $i \in J$.

\begin{shortversion}
An example of a view graph is given in the extended version (appendix).
\end{shortversion}

\begin{longversion}
\figurename~\ref{fig:view-graphs} shows a view graph for a FIP with four
players (the main player and three observers). The figure shows the edges
corresponding to communications, but we omit the edges corresponding
to looking into the past. Given the view graph 
of $\tau = c_1 c_2 c_3 c_4 c_5 c_6 c_7 \dots$ in \figurename~\ref{fig:view-graph},
the view of Player~$0$ after $c_6$ is illustrated in \figurename~\ref{fig:view-graph6},
and after $c_7$ in \figurename~\ref{fig:view-graph7}. 

%\mynote{T: Here we need an example of view graph as a figure}

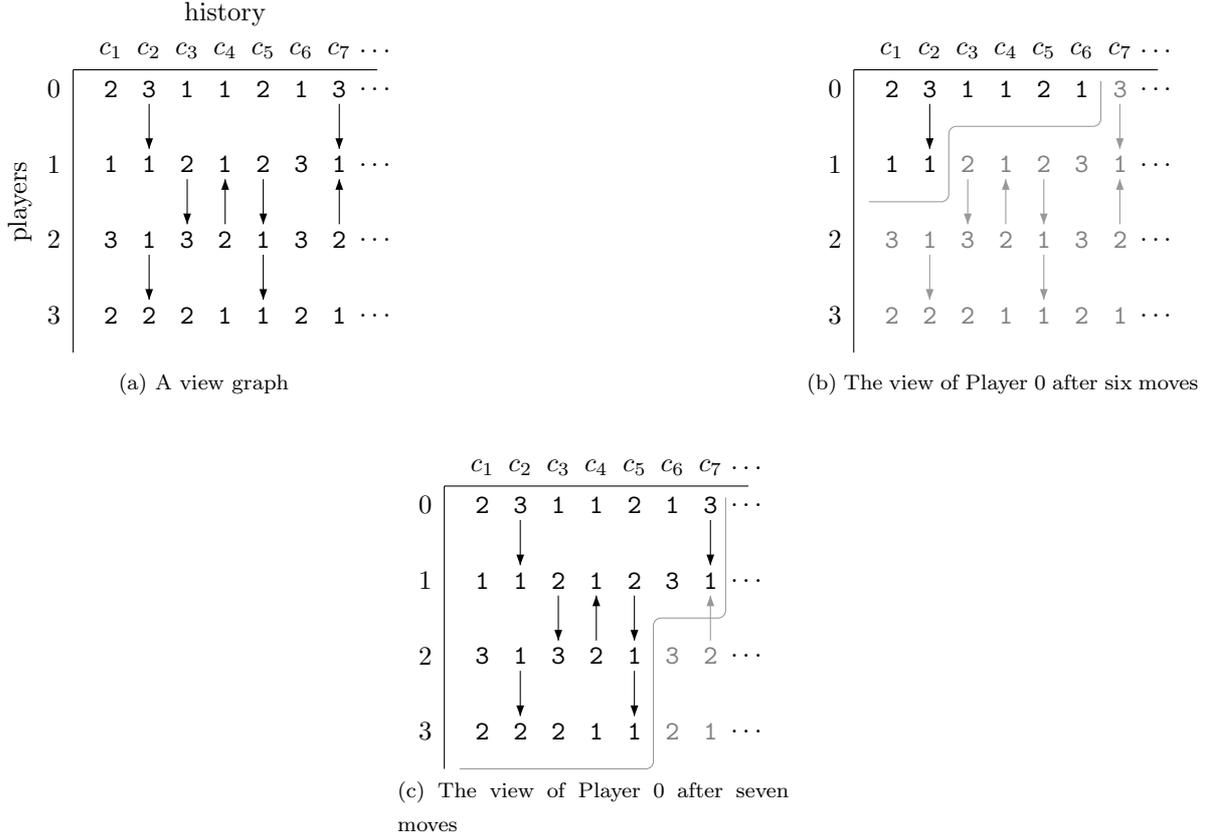
\begin{figure}[!tb]%
\begin{center}
\hrule
%\hspace{10mm}
\subfloat[A view graph {\large \strut}]{
   %\renewcommand{\sb}[1]{\scalebox{0.75}[1]{#1}}

%{\scriptsize 
\begin{picture}(53,45)(0,5)
%\put(0,0){\framebox(53,50){}}

      \gasset{Nw=3,Nh=4,Nmr=2, rdist=1, loopdiam=5}      % , ELdist=0

      \node[Nmarks=n, Nframe=n](p0)(3,25){\rotatebox{90}{players}}
      \node[Nmarks=n, Nframe=n](p0)(30,50){history}
      \node[Nmarks=n, Nframe=n](p0)(15,45){$c_1$}
      \node[Nmarks=n, Nframe=n](p0)(20,45){$c_2$}
      \node[Nmarks=n, Nframe=n](p0)(25,45){$c_3$}
      \node[Nmarks=n, Nframe=n](p0)(30,45){$c_4$}
      \node[Nmarks=n, Nframe=n](p0)(35,45){$c_5$}
      \node[Nmarks=n, Nframe=n](p0)(40,45){$c_6$}
      \node[Nmarks=n, Nframe=n](p0)(45,45){$c_7$}
      \node[Nmarks=n, Nframe=n](p0)(50,45){$\dots$}

      \node[Nmarks=n, Nframe=n](p0)(7.5,40){$0$}
      \node[Nmarks=n, Nframe=n](p1)(7.5,30){$1$}
      \node[Nmarks=n, Nframe=n](q2)(7.5,20){$2$}
      \node[Nmarks=n, Nframe=n](q3)(7.5,10){$3$}

      \drawline[AHnb=0,arcradius=1](10,5)(10,42.5)
      \drawline[AHnb=0,arcradius=1](50,42.5)(10,42.5)
      
      \node[Nmarks=n, Nframe=n](z1)(15,40){\t}
      \node[Nmarks=n, Nframe=n](o1)(15,30){\o}
      \node[Nmarks=n, Nframe=n](t1)(15,20){\th}
      \node[Nmarks=n, Nframe=n](th1)(15,10){\t}

      \node[Nmarks=n, Nframe=n](z2)(20,40){\th}
      \node[Nmarks=n, Nframe=n](o2)(20,30){\o}
      \node[Nmarks=n, Nframe=n](t2)(20,20){\o}
      \node[Nmarks=n, Nframe=n](th2)(20,10){\t}

      \node[Nmarks=n, Nframe=n](z3)(25,40){\o}
      \node[Nmarks=n, Nframe=n](o3)(25,30){\t}
      \node[Nmarks=n, Nframe=n](t3)(25,20){\th}
      \node[Nmarks=n, Nframe=n](th3)(25,10){\t}

      \node[Nmarks=n, Nframe=n](z4)(30,40){\o}
      \node[Nmarks=n, Nframe=n](o4)(30,30){\o}
      \node[Nmarks=n, Nframe=n](t4)(30,20){\t}
      \node[Nmarks=n, Nframe=n](th4)(30,10){\o}

      \node[Nmarks=n, Nframe=n](z5)(35,40){\t}
      \node[Nmarks=n, Nframe=n](o5)(35,30){\t}
      \node[Nmarks=n, Nframe=n](t5)(35,20){\o}
      \node[Nmarks=n, Nframe=n](th5)(35,10){\o}

      \node[Nmarks=n, Nframe=n](z6)(40,40){\o}
      \node[Nmarks=n, Nframe=n](o6)(40,30){\th}
      \node[Nmarks=n, Nframe=n](t6)(40,20){\th}
      \node[Nmarks=n, Nframe=n](th6)(40,10){\t}

      \node[Nmarks=n, Nframe=n](z7)(45,40){\th}
      \node[Nmarks=n, Nframe=n](o7)(45,30){\o}
      \node[Nmarks=n, Nframe=n](t7)(45,20){\t}
      \node[Nmarks=n, Nframe=n](th7)(45,10){\o}

      \node[Nmarks=n, Nframe=n](z8)(50,40){$\dots$}
      \node[Nmarks=n, Nframe=n](o8)(50,30){$\dots$}
      \node[Nmarks=n, Nframe=n](t8)(50,20){$\dots$}
      \node[Nmarks=n, Nframe=n](th8)(50,10){$\dots$}

      \drawedge[ELpos=50, ELside=l, curvedepth=0](z2,o2){}
      \drawedge[ELpos=50, ELside=l, curvedepth=0](t2,th2){}
      %\drawedge[ELpos=50, ELside=l, curvedepth=0](z2,z1){}

      \drawedge[ELpos=50, ELside=l, curvedepth=0](o3,t3){}

      \drawedge[ELpos=50, ELside=l, curvedepth=0](t4,o4){}

      \drawedge[ELpos=50, ELside=l, curvedepth=0](o5,t5){}
      \drawedge[ELpos=50, ELside=l, curvedepth=0](t5,th5){}

      \drawedge[ELpos=50, ELside=l, curvedepth=0](z7,o7){}
      \drawedge[ELpos=50, ELside=l, curvedepth=0](t7,o7){}

%\node[Nmarks=n](n1)(30,60){$\ver_1$}
%\node[Nmarks=n, Nmr=0](n11)(15,45){$\ver_2$}
%\nodelabel[ExtNL=y, NLangle=270, NLdist=1](n111){{\small val=1}}

%\drawedge[ELpos=50, ELside=r, curvedepth=0, AHLength = 2, AHangle=25, AHlength = 1.81](n1,n11){$-1$}
%\drawedge[ELpos=50, ELside=r, curvedepth=0, AHnb=0, linewidth=.6](n1,n11dummy){}

%\drawedge[ELpos=50, ELside=l, curvedepth=-12](t3,n1){}
%\drawbpedge(t3,70,22,n1,110,22){}
%\drawline[AHnb=1,arcradius=1](113,17.5)(113,29)(5,29)(5,17.5)

%\drawloop[ELside=l,loopCW=y, loopdiam=6](n4){$1$}

%\drawloop[ELside=l,loopCW=y](nk){$0,1$}

%\drawedge[dash={1}0](n3bis,nkbis){$0,1$}

\end{picture}
%}
   \label{fig:view-graph}  
}
\hfill
\subfloat[The view of Player~$0$ after six moves{\large \strut}]{
   %\renewcommand{\sb}[1]{\scalebox{0.75}[1]{#1}}

%{\scriptsize 
\begin{picture}(48,45)(5,5)
%\put(5,0){\framebox(48,50){}}

      \gasset{Nw=3,Nh=4,Nmr=2, rdist=1, loopdiam=5}      % , ELdist=0

      %\node[Nmarks=n, Nframe=n](p0)(3,25){\rotatebox{90}{players}}
      %\node[Nmarks=n, Nframe=n](p0)(30,50){history}
      \node[Nmarks=n, Nframe=n](p0)(15,45){$c_1$}
      \node[Nmarks=n, Nframe=n](p0)(20,45){$c_2$}
      \node[Nmarks=n, Nframe=n](p0)(25,45){$c_3$}
      \node[Nmarks=n, Nframe=n](p0)(30,45){$c_4$}
      \node[Nmarks=n, Nframe=n](p0)(35,45){$c_5$}
      \node[Nmarks=n, Nframe=n](p0)(40,45){$c_6$}
      \node[Nmarks=n, Nframe=n](p0)(45,45){$c_7$}
      \node[Nmarks=n, Nframe=n](p0)(50,45){$\dots$}

      \node[Nmarks=n, Nframe=n](p0)(7.5,40){$0$}
      \node[Nmarks=n, Nframe=n](p1)(7.5,30){$1$}
      \node[Nmarks=n, Nframe=n](q2)(7.5,20){$2$}
      \node[Nmarks=n, Nframe=n](q3)(7.5,10){$3$}

      \drawline[AHnb=0,arcradius=1](10,5)(10,42.5)
      \drawline[AHnb=0,arcradius=1](50,42.5)(10,42.5)
      
      \node[Nmarks=n, Nframe=n](z1)(15,40){\t}
      \node[Nmarks=n, Nframe=n](o1)(15,30){\o}
      \node[Nmarks=n, Nframe=n](t1)(15,20){\gray50{\th}}
      \node[Nmarks=n, Nframe=n](th1)(15,10){\gray50{\t}}

      \node[Nmarks=n, Nframe=n](z2)(20,40){\th}
      \node[Nmarks=n, Nframe=n](o2)(20,30){\o}
      \node[Nmarks=n, Nframe=n](t2)(20,20){\gray50{\o}}
      \node[Nmarks=n, Nframe=n](th2)(20,10){\gray50{\t}}

      \node[Nmarks=n, Nframe=n](z3)(25,40){\o}
      \node[Nmarks=n, Nframe=n](o3)(25,30){\gray50{\t}}
      \node[Nmarks=n, Nframe=n](t3)(25,20){\gray50{\th}}
      \node[Nmarks=n, Nframe=n](th3)(25,10){\gray50{\t}}

      \node[Nmarks=n, Nframe=n](z4)(30,40){\o}
      \node[Nmarks=n, Nframe=n](o4)(30,30){\gray50{\o}}
      \node[Nmarks=n, Nframe=n](t4)(30,20){\gray50{\t}}
      \node[Nmarks=n, Nframe=n](th4)(30,10){\gray50{\o}}

      \node[Nmarks=n, Nframe=n](z5)(35,40){\t}
      \node[Nmarks=n, Nframe=n](o5)(35,30){\gray50{\t}}
      \node[Nmarks=n, Nframe=n](t5)(35,20){\gray50{\o}}
      \node[Nmarks=n, Nframe=n](th5)(35,10){\gray50{\o}}

      \node[Nmarks=n, Nframe=n](z6)(40,40){\o}
      \node[Nmarks=n, Nframe=n](o6)(40,30){\gray50{\th}}
      \node[Nmarks=n, Nframe=n](t6)(40,20){\gray50{\th}}
      \node[Nmarks=n, Nframe=n](th6)(40,10){\gray50{\t}}

      \node[Nmarks=n, Nframe=n](z7)(45,40){\gray50{\th}}
      \node[Nmarks=n, Nframe=n](o7)(45,30){\gray50{\o}}
      \node[Nmarks=n, Nframe=n](t7)(45,20){\gray50{\t}}
      \node[Nmarks=n, Nframe=n](th7)(45,10){\gray50{\o}}

      \node[Nmarks=n, Nframe=n](z8)(50,40){$\dots$}
      \node[Nmarks=n, Nframe=n](o8)(50,30){$\dots$}
      \node[Nmarks=n, Nframe=n](t8)(50,20){$\dots$}
      \node[Nmarks=n, Nframe=n](th8)(50,10){$\dots$}

      \drawline[AHnb=0,arcradius=1,linegray=.6](42.5,41)(42.5,35)(22.5,35)(22.5,25)(12,25)

      \drawedge[ELpos=50, ELside=l, curvedepth=0](z2,o2){}
      \drawedge[ELpos=50, ELside=l, curvedepth=0,linegray=.6](t2,th2){}
      %\drawedge[ELpos=50, ELside=l, curvedepth=0](z2,z1){}

      \drawedge[ELpos=50, ELside=l, curvedepth=0,linegray=.6](o3,t3){}

      \drawedge[ELpos=50, ELside=l, curvedepth=0,linegray=.6](t4,o4){}

      \drawedge[ELpos=50, ELside=l, curvedepth=0,linegray=.6](o5,t5){}
      \drawedge[ELpos=50, ELside=l, curvedepth=0,linegray=.6](t5,th5){}

      \drawedge[ELpos=50, ELside=l, curvedepth=0,linegray=.6](z7,o7){}
      \drawedge[ELpos=50, ELside=l, curvedepth=0, linegray=.6](t7,o7){}

%\node[Nmarks=n](n1)(30,60){$\ver_1$}
%\node[Nmarks=n, Nmr=0](n11)(15,45){$\ver_2$}
%\nodelabel[ExtNL=y, NLangle=270, NLdist=1](n111){{\small val=1}}

%\drawedge[ELpos=50, ELside=r, curvedepth=0, AHLength = 2, AHangle=25, AHlength = 1.81](n1,n11){$-1$}
%\drawedge[ELpos=50, ELside=r, curvedepth=0, AHnb=0, linewidth=.6](n1,n11dummy){}

%\drawedge[ELpos=50, ELside=l, curvedepth=-12](t3,n1){}
%\drawbpedge(t3,70,22,n1,110,22){}
%\drawline[AHnb=1,arcradius=1](113,17.5)(113,29)(5,29)(5,17.5)

%\drawloop[ELside=l,loopCW=y, loopdiam=6](n4){$1$}

%\drawloop[ELside=l,loopCW=y](nk){$0,1$}

%\drawedge[dash={1}0](n3bis,nkbis){$0,1$}

\end{picture}
%}
    \label{fig:view-graph6}
}
\hfill
\subfloat[The view of Player~$0$ after seven moves{\large \strut}]{
   %\renewcommand{\sb}[1]{\scalebox{0.75}[1]{#1}}

%{\scriptsize 
\begin{picture}(48,45)(5,5)
%\put(5,0){\framebox(48,50){}}

      \gasset{Nw=3,Nh=4,Nmr=2, rdist=1, loopdiam=5}      % , ELdist=0

      %\node[Nmarks=n, Nframe=n](p0)(3,25){\rotatebox{90}{players}}
      %\node[Nmarks=n, Nframe=n](p0)(30,50){history}
      \node[Nmarks=n, Nframe=n](p0)(15,45){$c_1$}
      \node[Nmarks=n, Nframe=n](p0)(20,45){$c_2$}
      \node[Nmarks=n, Nframe=n](p0)(25,45){$c_3$}
      \node[Nmarks=n, Nframe=n](p0)(30,45){$c_4$}
      \node[Nmarks=n, Nframe=n](p0)(35,45){$c_5$}
      \node[Nmarks=n, Nframe=n](p0)(40,45){$c_6$}
      \node[Nmarks=n, Nframe=n](p0)(45,45){$c_7$}
      \node[Nmarks=n, Nframe=n](p0)(50,45){$\dots$}

      \node[Nmarks=n, Nframe=n](p0)(7.5,40){$0$}
      \node[Nmarks=n, Nframe=n](p1)(7.5,30){$1$}
      \node[Nmarks=n, Nframe=n](q2)(7.5,20){$2$}
      \node[Nmarks=n, Nframe=n](q3)(7.5,10){$3$}

      \drawline[AHnb=0,arcradius=1](10,5)(10,42.5)
      \drawline[AHnb=0,arcradius=1](50,42.5)(10,42.5)
      
      \node[Nmarks=n, Nframe=n](z1)(15,40){\t}
      \node[Nmarks=n, Nframe=n](o1)(15,30){\o}
      \node[Nmarks=n, Nframe=n](t1)(15,20){\th}
      \node[Nmarks=n, Nframe=n](th1)(15,10){\t}

      \node[Nmarks=n, Nframe=n](z2)(20,40){\th}
      \node[Nmarks=n, Nframe=n](o2)(20,30){\o}
      \node[Nmarks=n, Nframe=n](t2)(20,20){\o}
      \node[Nmarks=n, Nframe=n](th2)(20,10){\t}

      \node[Nmarks=n, Nframe=n](z3)(25,40){\o}
      \node[Nmarks=n, Nframe=n](o3)(25,30){\t}
      \node[Nmarks=n, Nframe=n](t3)(25,20){\th}
      \node[Nmarks=n, Nframe=n](th3)(25,10){\t}

      \node[Nmarks=n, Nframe=n](z4)(30,40){\o}
      \node[Nmarks=n, Nframe=n](o4)(30,30){\o}
      \node[Nmarks=n, Nframe=n](t4)(30,20){\t}
      \node[Nmarks=n, Nframe=n](th4)(30,10){\o}

      \node[Nmarks=n, Nframe=n](z5)(35,40){\t}
      \node[Nmarks=n, Nframe=n](o5)(35,30){\t}
      \node[Nmarks=n, Nframe=n](t5)(35,20){\o}
      \node[Nmarks=n, Nframe=n](th5)(35,10){\o}

      \node[Nmarks=n, Nframe=n](z6)(40,40){\o}
      \node[Nmarks=n, Nframe=n](o6)(40,30){\th}
      \node[Nmarks=n, Nframe=n](t6)(40,20){\gray50{\th}}
      \node[Nmarks=n, Nframe=n](th6)(40,10){\gray50{\t}}

      \node[Nmarks=n, Nframe=n](z7)(45,40){\th}
      \node[Nmarks=n, Nframe=n](o7)(45,30){\o}
      \node[Nmarks=n, Nframe=n](t7)(45,20){\gray50{\t}}
      \node[Nmarks=n, Nframe=n](th7)(45,10){\gray50{\o}}

      \node[Nmarks=n, Nframe=n](z8)(50,40){$\dots$}
      \node[Nmarks=n, Nframe=n](o8)(50,30){$\dots$}
      \node[Nmarks=n, Nframe=n](t8)(50,20){$\dots$}
      \node[Nmarks=n, Nframe=n](th8)(50,10){$\dots$}

      \drawline[AHnb=0,arcradius=1,linegray=.6](47,41)(47,25)(37.5,25)(37.5,5)(12,5)

      \drawedge[ELpos=50, ELside=l, curvedepth=0](z2,o2){}
      \drawedge[ELpos=50, ELside=l, curvedepth=0](t2,th2){}
      %\drawedge[ELpos=50, ELside=l, curvedepth=0](z2,z1){}

      \drawedge[ELpos=50, ELside=l, curvedepth=0](o3,t3){}

      \drawedge[ELpos=50, ELside=l, curvedepth=0](t4,o4){}

      \drawedge[ELpos=50, ELside=l, curvedepth=0](o5,t5){}
      \drawedge[ELpos=50, ELside=l, curvedepth=0](t5,th5){}

      \drawedge[ELpos=50, ELside=l, curvedepth=0](z7,o7){}
      \drawedge[ELpos=50, ELside=l, curvedepth=0, linegray=.6](t7,o7){}

%\node[Nmarks=n](n1)(30,60){$\ver_1$}
%\node[Nmarks=n, Nmr=0](n11)(15,45){$\ver_2$}
%\nodelabel[ExtNL=y, NLangle=270, NLdist=1](n111){{\small val=1}}

%\drawedge[ELpos=50, ELside=r, curvedepth=0, AHLength = 2, AHangle=25, AHlength = 1.81](n1,n11){$-1$}
%\drawedge[ELpos=50, ELside=r, curvedepth=0, AHnb=0, linewidth=.6](n1,n11dummy){}

%\drawedge[ELpos=50, ELside=l, curvedepth=-12](t3,n1){}
%\drawbpedge(t3,70,22,n1,110,22){}
%\drawline[AHnb=1,arcradius=1](113,17.5)(113,29)(5,29)(5,17.5)

%\drawloop[ELside=l,loopCW=y, loopdiam=6](n4){$1$}

%\drawloop[ELside=l,loopCW=y](nk){$0,1$}

%\drawedge[dash={1}0](n3bis,nkbis){$0,1$}

\end{picture}
%}
    \label{fig:view-graph7}
}
%\hspace{10mm}
\hrule
%\smallskip
\caption{View graphs (we omit all edges pointing backwards, correspond to looking into the past). \label{fig:view-graphs}}%
\end{center}
\end{figure}
\end{longversion}

\medskip
A \emph{full-information protocol} (FIP) $F = \tuple{I, (\M_i)_{i\in I},(R_{\sigma})_{\sigma \in \Sigma}}$
with $n$ observers over move alphabet $\Gamma$ and observation alphabet $\Sigma$
consists of a set $I = \{0,1,\dots,n\}$ of players, Mealy machines $\M_i$
defining the observation functions $\beta_i: \Gamma^* \to \Sigma$ of each player $i \in I$,
and the relations $R_{\sigma} \subseteq I \times I$ defining the communication links
between the players on observations $\sigma \in \Sigma$.
By extension, a full-information protocol is a game 
$\tuple{A, \Gamma, \act, \sim, \M}$ where the indistinguishability relation $\sim$ is 
$\sim_0$ defined by $F$.
Moreover, we require that two moves with different actions have different observation,
if $\act(c) \neq \act(c')$, then $\beta_0(\tau c) \neq \beta_0(\tau' c')$ for 
all histories $\tau, \tau' \in \Gamma^*$ and moves $c,c' \in \Gamma$,
ensuring that the action is visible to the player.
It is then easy to see that $\sim$ is indeed an indistinguishability relation.

Note that FIP games with one player and no observer ($I= \{0\}$) correspond
to the special case of partial-observation games~\cite{Rei84} where  
the indistinguishability relation is represented by a single (regular) observation
function.
%\mynote{L: insert characterizschation of $\sim_i$ without view graphs ?}

%\def\mynewnote#1{L: add comments on FIP model, relate to Byzantine, faults, etc.}

\section{Graph Games and Morphisms}\label{sec:graph-games}
The key tool to strategy synthesis for infinite games is the
automata-theoretic procedure founded on the works of B\"uchi and Landweber~\cite{BL69},
and of Rabin~\cite{Rabin69}.
Setting out from an automaton that recognises the set of strategies in a game and a second one
that recognises the winning condition, 
the procedure constructs a new automaton that recognises the set of winning strategies.
The emptiness test for the constructed automaton is decidable,
answering the question of whether winning strategies exist.
Moreover, by Rabin's Basis Theorem~\cite{Rabin72},
every nonempty automaton accepts a regular tree, 
which corresponds to the unfolding of a finite graph --
this allows to effectively construct a winning strategy defined by a Mealy machine.

An essential feature of the automata-theoretic approach is that strategies
are presented as trees with bounded, finite branching.
In our setting, however, 
the information trees which support strategies might have 
unbounded degree. Indeed, it was shown in \cite{BD23},
that a regular indistinguishability relation defines an information tree with finite branching
if, and only if, there exists an equivalent observation function.
Since FIP protocols are more expressive than observation functions, as we show
in Section~\ref{sec:expressiveness},
this means that we cannot rely on tree automata to recognise the set of strategies of a FIP game in general.

To overcome this obstacle, we propose a construction that transforms
any FIP game into a game with perfect information, by preserving
the existence of winning strategies in the following sense:
$(1)$~whenever a winning strategy exists in the original game, there exists one in the transformed game; $(2)$~given a regular winning strategy for transformed game,
we can effectively construct a winning strategy for the original game.

To prepare the ground, we first discuss some general transformation of games with imperfect
information into games of perfect information that preserves
the existence of winning strategies, and present a sufficient 
condition for the transformed game of perfect information to
be regular and thus solvable.
In Section~\ref{sec:solving-FIP}, we describe a particular transformation
for solving the FIP synthesis problem.

\subsection{Game graphs}

It will be convenient to consider repeated games played on a graph,
which is a model equivalent to abstract repeated games~\cite{Thomas95}.
We briefly recall the definition of game graphs as repeated games.
%To describe infinite multi-stage games in general, we use a relational vocabulary,
%where moves correspond to unnamed edges labelled with actions.

%\mynote{node vs. vertex}
Let~$A$ be a set of actions and $C$ be a set of colors.
A \emph{(game) graph} is a structure $\V = (V, v_\init, (E_a)_{a \in A}, \lambda)$
on a set $V$ of nodes called the domain with a designated initial node $v_\init \in V$, 
a binary edge relation $E_a \subseteq V \times V$ for every action~$a \in A$, 
and a node-labeling function $\lambda: V \to C$. We require that for
every node $v \in V$ and action $a \in A$, the set $E_a(v) = \{w \mid (v,w) \in E_a\}$
of successors of $v$ by $a$ is nonempty.

Intuitively, a game on $\V$ is played in rounds as follows. 
Each round starts in a node, the first round starts in the initial node $v_{\init}$. 
In each round, given the node $v$ in which the round starts, the player chooses 
an action $a \in A$, then the environment chooses a node $w$ such that $(v,w) \in E_a$.
The next round starts in the node $w$. 

As a repeated game, the game on $\V$ is the perfect-information game 
$\G_{\V} = \tuple{A,\Gamma,\act, \M}$
with the set of moves $\Gamma = A \times V$,
the function $\act$ defined by $\act(a,v) = a$ for all $(a,v) \in \Gamma$,
and the Mealy machine $\M$ that maps a history $\tau = (a_1,v_1) \ldots (a_n,v_n)$
to the color $\lambda(v_n)$ if $\tau$ forms a path
from $v_\init$ in the graph $\V$, and to the color $\bot$ otherwise.
Given a winning condition $L \subseteq C^{\omega}$ for the game on $\V$,
a play is declared winning in $\G_{\V}$ if it is mapped by $\M$ to 
a sequence in~$L$ or to a sequence containing $\bot$, which
forces the environment to respect the edge relations of the game graph.

The domain of $\V$ may be finite or infinite.
However, the graphs we consider have at most countable set of nodes, they are finitely branching, 
that is, every node has finitely many successors,
and the range of the coloring function~$\lambda$ is finite.
%Finally, we assume that transition graphs are connected, that is, 
%for every node $v$ there is a path, from the initial node, that ends at $v$. 
%A tree is a (connected) graph where, for every node~$v$, there is a 
%unique path that ends at $v$.
%\mynote{L: do we assume these, or do we only construct graphs that satisfy them?}

In the sequel we construct game graphs as the quotient of Mealy automata, defined
as follows. 
%\smallskip\noindent{\em Quotient.}
Let $\A = (Q, \Gamma, \Sigma, q_\init, \delta, \lambda)$ be a Mealy automaton.
Given an action map $\act: \Gamma \to A$ and an equivalence relation 
$R\subseteq Q \times Q$ that respects $\lambda$ (i.e., if $(q,q') \in R$, 
then $\lambda(q) = \lambda(q')$), the \emph{quotient} by $R$ of $\A$ 
is the graph $\quot{\A}{R} = (V, v_\init, (E_a)_{a \in A}, \lambda)$
where $V$ is the set of all equivalence classes in $Q$, the initial node
is $v_\init = [q_\init]_R$, and for all $a \in A$ the edge relation $E_a$ connects two
equivalence classes $(u,u')$ whenever there exists a state $q \in u$
with a successor $\delta(q,c) \in u'$ for some $c \in \Gamma$ such that $\act(c) = a$.
By an abuse of notation, we define $\lambda([q]_R) = \lambda(q)$ for all $q \in Q$.

\subsection{Information tree}

\begin{comment}
Our goal is to reduce large games, represented on an infinite domain, to smaller ones,
on a finite domain, for which the synthesis problem is solvable, and then transfer
winning strategies back to the original game.
The reduction should be conservative in the sense that, whenever a winning strategy
exists in the large game, there is one in the small game.

In our representation, every game with a non-trivial indistinguishability relation
is infinite.
To be effective, a reduction should hence yield a game with perfect information.

One way of transforming larger structures into smaller ones is by
collapsing elements that we do not wish to distinguish.
\end{comment}
Our first step is to represent games with imperfect information 
as a game of perfect information played on the information tree, 
which is an infinite graph whose nodes are information sets. 
%\mynote{nodes vs. vertices}

The \emph{information tree} for a game $\G = \tuple{A,\Gamma,\act,\sim,\lambda}$ 
is the transition graph
$\U(\G)$ on the domain $U = \{ [\tau]_{\sim} \mid \tau \in \Gamma^* \} $, 
with initial node $u_\init = [\init]_{\sim}$, 
edge sets
$E_a^{\U} := \bigl\{ \bigl( [\tau]_{\sim}, [\tau c]_{\sim} \bigr) \mid \act(c) = a \bigr\}$,
for each action $a \in A$, 
and with coloring
$\lambda^{\U}( [\tau]_{\sim} ) = \lambda( \tau )$ for all $\tau \in \Gamma^*$
(which is well-defined since $\lambda$ is information-consistent).
Note that $\U(\G)$ is a tree due to the perfect-recall property of~$R$.
Therefore, every node $u \in U$ identifies a unique path from $u_\init$
to $u$ and we view strategies in the game payed on $\U(\G)$
as functions $\straa: U \to A$
rather than $\straa: U^{+} \to A$.

The information tree can also be obtained as the quotient by $\sim$ of
the (infinite-state) automaton with state space $Q = \Gamma^*$
and transition function defined by $\delta(\tau, c) = \tau c$
for all $\tau \in \Gamma^*$ and $c \in \Gamma$.

Although structurally different, a game with imperfect information and the perfect-information
game played on its information tree are the same game, in the following sense.

\begin{lemma}\label{lem:info-quotient}
  For every game $\G$ with indistinguishability relation~$\sim$, 
  there is a bijection that maps every strategy $\straa$ in the information tree $\U(\G)$
  to an (information-consistent) strategy $\tilde{\straa}$ of the original game $\G$,
  such that $\tilde{\straa}(\tau) = \straa([\tau]_{\sim})$, for all histories $\tau \in \Gamma^*$.
  Moreover, the outcomes of corresponding strategies agree on the
  colors: $\hat{\lambda}(\outcome(\tilde{\straa}) = \hat{\lambda}^{\U}(\outcome(\straa)) $. 
\end{lemma}

\begin{longversion}
\begin{proof}
  By the correspondence between $\straa$ and $\tilde{\straa}$, 
  for each play $c_1 c_2 \dots \in \outcome(\tilde{\straa})$,
  the sequence $a_1 [c_1]_{\sim}\, a_2 [c_1 c_2]_{\sim} \dots$,
  where $a_i = \act(c_i)$ for all $i \geq 1$ is a play in the game on $\U(\G)$ 
  that follows~$\straa$ and forms a path from $[\epsilon]_{\sim}$, thus
  the colors $\lambda^{\U}([c_1 \dots c_i]_{\sim}) = \lambda(c_i)$
  are equal in every stage~$i \geq 1$.

  To show, conversely, that $\hat{\lambda}^{\U}(\outcome(\straa)) \subseteq 
  \lambda(\outcome (\tilde{\straa})$,
  consider a play~$\pi = a_1 v_1 a_2 v_2 \dots$ in $\outcome(\straa)$, 
  and for every $i \geq 0$ let~$\tau_i \in v_{i}$ be a history of the class $v_i$,
  and note that $\tau_i$ follows $\tilde{\straa}$ in $\G$ and $\lambda(\tau_i) = \lambda^{\U}(\pi(i))$.
  The set of all histories $\tau_i$ forms an infinite subtree of $\Gamma^*$, which has degree
  bounded by $\abs{\Gamma}$, and thus contains an infinite path $\theta$ 
  by K{\"o}nig's lemma~\cite{Konig36}.
  It follows that $\hat{\lambda}^{\U}(\pi) = \hat{\lambda}(\theta)$ and 
  $\theta \in \outcome (\tilde{\straa})$, which establishes the desired inclusion.
\end{proof}
\end{longversion}

\subsection{Bisimulation}
%\balance
To construct transformations that allow taking strategies back and forth between 
games systematically, we use the classical notion of bisimulation.
A bisimulation between two graphs $\G$ and $\H$ with the usual vocabulary,
is a relation $Z \subseteq V^\G \times V^\H$ such that,
every related pair of nodes $(u, v) \in Z$ agree on the color
$\lambda^\G( u ) = \lambda^\H( v )$ and

\begin{description}[leftmargin=3.5em,style=nextline]  %[\labelwidth=5em]
\item[(Zig)]
  for each action~$a \in A$ and every edge $(u, u') \in E_a^\G$,
  there exists an edge $(v, v') \in E_a^\H$ such that $(u', v') \in Z$, and
\item[(Zag)]
  for every action~$a \in A$ and every edge $(v, v') \in E_a^\H$,
  there exists an edge $(u, u') \in E_a^\G$ such that $(u', v') \in Z$.
\end{description}

One can verify easily that the union of two bisimulations is again a 
bisimulation, and thus the coarsest bisimulation between two transition graphs
can be obtained by taking the union of of all bisimulations between them:
two nodes $u \in V^\G$ and $v \in V^\H$ are \emph{bisimilar}, denoted by $u \approx v$,
if there exists a bisimulation that contains $(u, v)$.
By extension, we say that two graphs $\G$ and $\H$ are bisimilar
if their initial nodes are bisimilar, $v_\init^\G \approx v_\init^\H$. 
As a basic notion of dynamic equivalence,
bisimulation has been studied widely and in different variants~\cite{BK08,San11}.
For games with perfect information on graphs, it is folklore that winning strategies
are preserved across bisimilar representations. 

\begin{lemma}\label{lem:bisimulation-preserves-strat}
Given two bisimilar game graphs $\G$ and $\H$ and a logical specification $L \subseteq C^{\omega}$, 
the following equivalence holds:
there exists a winning strategy for $L$ in $\G$ if and only if 
there exists a winning strategy for $L$ in $\H$.
\end{lemma}

Moreover, if there exists a functional bisimulation $k: V^\G \to V^\H$ 
containing the initial nodes of the two graphs $\G$ and $\H$, $k(v_\init^\G) = v_\init^\H$,
then there exists a pruning of the unfolding of $\G$ that is isomorphic
to the unfolding of $\H$.

\subsection{Rectangular morphisms}\label{sec:covering}
%\section{Rectangular morphisms}\label{SectionRectMor}

Throughout this section, we fix a game $\G = \tuple{A,\Gamma,\act,\sim, \lambda}$ 
and a winning condition $L \subseteq C^{\omega}$. 

Intuitively, we aim at constructing a finite-state abstraction of the 
information tree, as a graph bisimilar to the information tree. 
The greatest difficulty is that the information tree may be of unbounded branching,
while a finite-state abstraction must have bounded branching by definition. 
The key is to be able to describe the navigation through the information tree in the
universe of histories: given an information set $u$ identified by an history $\tau \in u$
(that is $u = [\tau]_{\sim}$), the successors of $u$ can be identified by
the histories $\tau' c$ obtained by taking a companion $\tau' \sim \tau$
and then a successor $\tau' c$ by appending a move $c \in \Gamma$.

Intuitively, our finite-state abstraction is induced by a finite-valued function $h:\Gamma^* \to P$
such that we can deduce, given the value of $h(\tau)$, the following elements useful
to navigate through the information tree:
\begin{enumerate}
\item the set of values $\{h(\tau') \mid \tau' \sim \tau\}$,  
\item the set of values $\{h(\tau c) \mid c \in \Gamma\}$, and
\item the value of $\lambda(\tau)$.
\end{enumerate}

Note that these values should be deducible without knowing the value of $\tau$,
so that we can faithfully navigate in the universe $h(\Gamma^*) \subseteq P$
(in the sequel we assume w.l.o.g that $P = h(\Gamma^*)$). 
This is possible when the function $h$ is a \emph{rectangular morphism}
for $\G$, that is satisfying the following properties, 
for all histories $\tau, \tau' \in \Gamma^*$ such that $h( \tau ) = h( \tau' )$:
%\begin{linenomath*}
  \begin{align*}
    \tag{rectangularity}  & h([\tau]_{\sim}) = h([\tau']_{\sim}), \\
    \tag{morphism} & h( \tau c) = h( \tau' c ), \text{ for all } c \in \Gamma,\\
    \tag{refinement} & \lambda( \tau ) = \lambda( \tau' ). \\
  \end{align*}
%%\end{linenomath*}
Note that a finite-valued morphism on $\Gamma^*$ is a regular function.

%\begin{remark}
A variant of the refinement property is to require that $h$ is a refinement of 
the automaton $\M$ defining $\lambda$, that is, 
if $h(\tau) = h(\tau')$, then $\delta(q_\init, \tau) = \delta(q_\init, \tau')$
where $q_\init$ is the initial state of $\M$. 
This variant implies the original refinement property.
If $h$ is a morphism and $\M$ is the minimal automaton defining $\lambda$, 
then the two properties (refinement and the variant) are equivalent. 
In the sequel, we extend $\lambda$ to the set $P$ and define $\lambda(p) = \lambda(\tau)$
for all $\tau$ such that $h(\tau) = p$, which is well defined by the refinement property of $h$.

We show that the solution of the synthesis problem for games with imperfect information boils down
to the construction of a function $h$ satisfying the
four conditions of being $(1)$~rectangular, $(2)$~a morphism, $(3)$~a refinement, and $(4)$ finite-state. It is easy to define functions satisfying any three of the four conditions, namely:

\begin{itemize}
\item all but rectangular ($2,3,4$): $h_1(\tau) = \delta(q_\init, \tau)$; 
\item all but a morphism ($1,3,4$): $h_2(\tau) = \tuple{\delta(q_\init, \tau), \{\delta(q_\init, \tau') \mid \tau' \sim \tau\} }$;
\item all but a refinement ($1,2,4$): $h_3$ is constant.   %, $h_4(\tau) = \text{cte}$
\item all but finite-state ($1,2,3$): $h_4(\tau) = \tau$. 
\end{itemize}

The proof that $h_i$ satisfies conditions $\{1,2,3,4\} \setminus \{i\}$ is 
straightforward and left to the reader. We note that the rectangularity of $h_2$
is a corollary of the fact that for all functions $f$ on $\Gamma^*$, 
the function $h$ defined by $h(\tau) = \tuple{f(\tau), f([\tau]_{\sim})}$
for all $\tau \in \Gamma^*$ is rectangular (which is also straightforward 
to prove).

Partial-observation games (i.e., FIP games with no observer) admit a rectangular
morphism of the form $h_2$ where $\delta$ is the transition function of the
synchronous product of the Mealy automaton defining the coloring $\lambda$
and the Mealy automaton defining the observation function $\beta_0$. 
The function $h_2$ is then morphic~\cite{Rei84,CDHR07}: given $q = \delta(q_\init, \tau)$ 
and $u = \{\delta(q_\init, \tau') \mid \tau' \sim \tau\}$, thus $h_2(\tau) = (q,u)$,
and given a move $c$, we can define $q' = \delta(q, c)$ and 
$u' = \{\delta(q,c') \mid \exists c' \in \Gamma: \beta_0(q,c') = \beta_0(q,c) \}$,
and show that $h_2(\tau c) = (q',u')$.

In the rest of this section, we fix a rectangular morphism~$h$ for $\G$.
The role of rectangularity appears in the two crucial lemmas below, which lead
to the construction of a finite-state abstraction of the information tree.

\mynewnote{L: convey that the result is more surprising/challenging than it looks.}

\begin{lemma}\label{lem:RH-equivalence}
The relation~$R^{\H} = \bigl\{ \bigl( h(\tau), h(\tau') \bigr) \mid \tau \sim \tau' \bigr\}$
is an equivalence.
\end{lemma}

\begin{proof}
  It is immediate that $R^{\H}$ is reflexive (as $h$ is surjective) and symmetric. 

  To show that $R^{\H}$ is transitive, consider
  $\tau_1 \sim \tau_2$ and $\tau_2' \sim \tau_3$
  such that $h(\tau_2) = h(\tau_2')$,
  and show that there exists $\tau_1' \sim \tau_3'$ such that
  $h(\tau'_1) = h(\tau_1)$ and $h(\tau'_3) = h(\tau_3)$.
  By rectangularity, since $h(\tau_2) = h(\tau_2')$ we have 
  $h([\tau_1]_{\sim}) = h([\tau_3]_{\sim})$ (call that set $Y$) 
  and since in particular $h(\tau_1) \in Y$ and $h(\tau_3) \in Y$, 
  there exists $\tau_3' \in [\tau_1]_{\sim}$ such that $h(\tau_3') = h(\tau_3)$.
  We can take $\tau_1' = \tau_1$ and the result follows.
\end{proof}

The tight link between the equivalence classes of $\sim$ and of $R^{\H}$ 
is described in Lemma~\ref{lem:h-preserves-eq-class}.

\begin{lemma}\label{lem:h-preserves-eq-class}
$h([\tau]_{\sim}) = [h(\tau)]_{R^\H}$ for all $\tau \in \Gamma^*$.
\end{lemma}

\begin{proof}
  That $h([\tau]_{\sim}) \subseteq [h(\tau)]_{R^\H}$ follows by definition of $R^{\H}$.
  For the converse inclusion, let $p \in [h(\tau)]_{R^\H}$ and show that 
  there exists $\tau' \in [\tau]_{\sim}$ with $h(\tau') = p$. By definition
  of $R^\H$ there exist $\tau_1 \sim \tau_1'$ such that $h(\tau_1) = h(\tau)$ 
  and $h(\tau_1') = p$. It follows by rectangularity that $h([\tau_1]_{\sim}) = h([\tau]_{\sim})$
  and thus there exists $\tau' \in [\tau]_{\sim}$ with $h(\tau') = h(\tau_1') = p$
  as required.
\end{proof}

Consider the semi-automaton $\H = \tuple{P, p_{\init},\delta,\lambda}$ 
where $p_{\init} = h(\epsilon)$ and $\delta(p,c) = p'$ if there exists $\tau \in \Gamma^*$
such that $h(\tau) = p$ and $h(\tau c) = p'$, which is well defined by the morphism
property of $h$. 
The following result is an immediate consequence of Lemma~\ref{lem:h-preserves-eq-class}.

\begin{theorem}\label{thm:bisimilar}
  Let~$\G$ be a game with indistinguishability relation~$R$. If $h$ is a rectangular homomorphism
  on $\G$, the information tree $\U(\G)$ is bisimilar to the quotient of the semi-automaton $\H$  
  by the equivalence $R^\H$ constructed via~$h$.
\end{theorem}

Specifically, the function~$k$ induced by~$h$ that maps each information set $u = [\tau]_{\sim}$ 
to the set $k(u) = [h(\tau)]_{R^\H}$ is a p-morphism. 

Using Theorem~\ref{thm:bisimilar}, the solution of the synthesis problem
for a winning condition $L$
boils down to showing the existence of a rectangular morphism $h$, 
constructing the automaton $\H$, and solving the perfect-information
game played on $\H$ for $L$.
Note that this reduction holds for arbitrary winning conditions~$L$,
but is of practical interest only if the synthesis problem for 
perfect-information games is decidable, which is the case of $\omega$-regular
winning conditions~\cite{automata}.

Although Theorem~\ref{thm:bisimilar} does not show that having 
a rectangular morphism is required to solve the synthesis problem 
for games with imperfect information (nor to guarantee the existence of a finite 
bisimulation quotient), this approach appears to be sufficient
to show the decidability of partial-observation games and FIP games. 
Moreover, as discussed above, the requirement of rectangularity is very natural (if not necessary)
as we want to ``simulate'' the navigation through the information tree.

\section{Solving FIP Games}\label{sec:solving-FIP}

We solve FIP games by constructing a rectangular morphism
and reducing to a game of perfect information using Theorem~\ref{thm:bisimilar}.

\subsection{Pre-processing}\label{sec:pre-processing}
To simplify the presentation, we show that every FIP can be transformed
into an equivalent one where the observation functions are trivial.
Intuitively, the move alphabet in the transformed FIP is $\Gamma' = \Sigma^{n+1}$
where $\Sigma$ is the observation alphabet of the original FIP,
and $n$ is the number of observers. As the moves can now be any profile
of observations from the original FIP, we use the winning condition 
to ensure that if the sequence of observations in the transformed
FIP is not possible in the original FIP, then the player wins.
We consider the product of all Mealy machines for the 
observation functions in the original FIP, which defines a function
$\beta': \Gamma^* \to \Sigma^{n+1}$ (where $\Gamma$ is the set of moves
in the original FIP) such that $\beta'(\tau) = \tuple{\beta_0(\tau),
\beta_1(\tau),\dots,\beta_n(\tau)}$ for all $\tau \in \Gamma^*$ 
and we consider the language 
$L = \{\beta'(\tau) \mid \tau \in \Gamma^* \} \subseteq (\Gamma')^*$,
which is a regular language (a \dfa recognising $L$ can be obtained by
a standard subset construction on the Mealy machine 
defining $\beta'$). The winning condition in the transformed FIP 
accepts all plays in $(\Gamma')^*$ that have a winning pre-image by $\beta'$ in 
the original FIP, as well as all plays  in $(\Gamma')^*$ that have a (finite) 
prefix outside $L$.
The transformed FIP is equivalent to the original one in the sense that 
there exists a winning strategy in the transformed FIP if and only if there exists a 
winning strategy in the original FIP.

\mynote{L: skip the details of the transformation, and the proof? }
\mynote{L: ok, but should we mention the proof is very similar to proof of
Lemma~\ref{lem:info-quotient}? and uses K\"{o}nig's Lemma?}

From now on, we assume without loss of generality
that the move alphabet is $\Gamma = \Sigma^{n+1}$ and the observation
function for player $i \in I$ is defined by $\beta_i(\tau c) = c[i]$,
the component of $c$ corresponding to player~$i$,
for all $\tau \in \Gamma^*$ and $c \in \Gamma$.
For $J\subseteq I$, we denote by $c[J] = (c[j])_{j\in J}$ the observations 
of the coalition~$J$.

We present an alternative characterisation
of the indistinguishability relations $\sim_i$, without resorting to
view graphs. 
Given a move $c \in \Gamma$ and a player~$i \in I$, we define the set $\sync_i(c)$ of players 
communicating (directly, or via other players) with Player~$i$ on move $c$ as follows.
The set $\{i\} \times R_{c[i]}(i)$ 
contains the (direct) communication links in which Player~$i$ is a receiver.
Let $T = \bigcup_{i \in I} \{i\} \times R_{c[i]}(i)$, and the reflexive transitive
closure $T^*$ contains all (direct or indirect) communication links between the players,
so we define $\sync_i(c) = T^*(i)$.
For $J \subseteq I$, let $\sync_J(c) = \bigcup_{i \in J} \sync_i(c)$.
%For singleton $J = \{i\}$, we write $\sync_i(c)$ instead of $\sync_{\{i\}}(c)$.

Note that $J \subseteq \sync_J(c)$ for all coalitions~$J$ 
and moves $c \in \Gamma$ (coalitions always communicate with themselves), 
and for the coalition $K = \sync_J(c)$, there is no communication
with other players, $\sync_K(c) = K$. In fact $\sync_K(c) = \sync_J(c)$
for all $J \subseteq K \subseteq \sync_J(c)$.
Finally, note that $\sync_{J \cup K}(c) = \sync_J(c) \cup \sync_K(c)$
and thus $\sync$ is monotone with respect to coalitions:
if $J \subseteq K$, then $\sync_J(c) \subseteq \sync_K(c)$.

\begin{lemma}\label{lem:fip-properties}
For all coalitions $J \subseteq I$, for all histories $\tau \in \Gamma^*$ 
and moves $c \in \Gamma$, the following properties hold:

\begin{enumerate} 
\item for all $K \subseteq I$ such that $J \subseteq K \subseteq \sync_J(c)$, we have $[\tau c]_{\sim_K} = [\tau c]_{\sim_J}$, 
%(i.e., $\tau c \sim_K \tau' d \Leftrightarrow \tau c \sim_J \tau' d$)
\label{lem:fip-properties:P1}

\item for $K = \sync_J(c)$, we have $[\tau c]_{\sim_K} = \{ \tau'd \mid \tau' \sim_K \tau \land d[K] = c[K] \}$.
\label{lem:fip-properties:P2} 
\end{enumerate} 
\end{lemma}

Indeed, Lemma~\ref{lem:fip-properties} characterises
the indistinguishability relations~$\mathop{\sim_J}$ of coalitions~$J$, which can be
be defined inductively by $[\tau c]_{\sim_J} = [\tau c]_{\sim_K} = 
\{ \tau'd \mid \tau' \in [\tau]_{\sim_K} \land d[K] = c[K] \}$ where $K = \sync_J(c)$.
Note that, for the grand coalition~$I$, the relation $\sim_I$ is the identity.

\subsection{A rectangular morphism for FIP games}

We fix a FIP game $\G = \tuple{A,\Sigma^{n+1},\act,\sim,\lambda}$,
where $\sim$ is defined by Lemma~\ref{lem:fip-properties}, 
and $\lambda$ is defined by a given Mealy machine 
$\M = \tuple{Q, \Gamma, C, q_{\varepsilon}, \delta_{\M}, \lambda}$.

We define a function $h$ for our game $\G$ and then show that it is a rectangular morphism
for $\G$. %\mynote{L: should we call it $h_{\text{\textsf{{\tiny FIP}}}}}$ ?}
The configuration $h(\tau)$ at a history $\tau$ is a vector indexed by all
coalitions $J$ containing the main Player, that is, $0\in J$.
The component $h(\tau)[J]$ corresponding to a coalition $J \neq I$
is a knowledge set for $J$, which consists of a set of configuration components
corresponding to all coalitions $K$ greater than $J$, where the configurations
are calculated at histories $\tau' \sim_J \tau$ indistinguishable from $\tau$
for coalition $J$. For $J = I$ the grand coalition, the configuration
$h(\tau)[I]$ stores the state reached in $\M$ upon reading $\tau$.
It will be convenient to define $h(\tau)[I]$ as a set (a singleton) for uniform
treatment as a knowledge set. 

For all $\tau \in \Gamma^*$, define
$h(\tau)[I] = \{\delta_{\M}(q_{\init}, \tau)\}$,
and for all $J\subsetneq I, \text{ define } h(\tau)[J] = \{(h(\tau')[K])_{K \in \up{J}} \mid \tau' \sim_J \tau\}$
% \end{align*}

% \begin{align*}
% 	&h(\tau)[I] = \{\delta_{\M}(q_{\init}, \tau)\},\\ 
% 	&\text{and for all } J\subsetneq I, \text{ define } h(\tau)[J] = \{(h(\tau')[K])_{K \in \up{J}} \mid \tau' \sim_J \tau\}.
% \end{align*}
where $\up{J} = \{ K \subseteq I \mid J \subsetneq K\}$, which defines the function
$h: \Gamma^* \to P$ with $P = \prod_{\{0\} \subseteq J \subseteq I} \Psi_J$ where
$\Psi_I = 2^Q$ and, inductively, $\Psi_J = 2^{\prod_{K \in \up{J}} \Psi_K}$ for all
$J \subseteq I$.

It follows immediately from this definition that the component $h(\tau)[J]$ 
corresponding to a coalition $J$ is information consistent for $\sim_J$.

\begin{lemma}\label{lem:basic-property-h}
For all histories $\tau,\tau'\in \Gamma^*$ and all coalitions $\{0\} \subseteq J\subseteq I$,
if $\tau \sim_J \tau'$, then $h(\tau)[J] = h(\tau')[J]$.
\end{lemma}

We show that the function $h$ defined above is a rectangular morphism for $\G$.
That $h$ is a refinement is immediate since $h(\tau)[I] = \{\delta_{\M}(q_{\init}, \tau)\}$,
and that $h$ is rectangular is relatively straightforward. The proof that $h$ is a morphism
is more involved.

\begin{lemma}\label{lem:refinement-rectangular-h}
The constructed morphism $h$ is rectangular.
\end{lemma}

\begin{longversion}
\begin{proof}
%That $h$ is a refinement of $\M$ is immediate since $h(\tau)[I] = \{\delta(q_{\init}, \tau)\}$.
To show that $h$ is rectangular, let $h(\tau) = h(\xi)$ and $\tau' \sim \tau$. 
We construct $\xi' \sim \xi$ such that $h(\xi') = h(\tau')$. 
Since $\tau' \sim \tau$ (and $\sim$ is $\sim_0$), the tuple 
$(h(\tau')[K])_{K \in \up{\{0\}}}$ belongs to $h(\tau)$,
and since $h(\tau) = h(\xi)$, there exists $\xi' \sim \xi$ such that 
the tuple $h(\xi')[K] = h(\tau')[K]$ for all $K \in \up{\{0\}}$.
By Lemma~\ref{lem:basic-property-h} (with $J = \{0\}$), 
we also have $h(\xi')[\{0\}] = h(\tau')[\{0\}]$, and thus $h(\xi') = h(\tau')$,
which concludes the proof.
\end{proof}
\end{longversion}

To show that $h$ is a morphism, we construct in the rest of this section an update function 
$\Delta: P \times \Gamma \to P$ and show that $\Delta(h(\tau),c) = h(\tau c)$ for 
all $\tau \in \Gamma^*$ and $c \in \Gamma$.

%\subsubsection{The $\lift$ operator}
	
To define the update function $\Delta$, we need an auxiliary operator 
to deal with the effect of communication on the knowledge sets.
	%Plan de bataille pour la soiree. Il me reste une heure:
	%3)modifier l'entree en matiere, donc avant ce plan pour avoir la def de up,
%	
When a coalition $J$ synchronises with a set $S$ of observers, 
the knowledge of the coalition $J \cup S$ is transferred to the
coalition $J$. The transfer is not a simple copy, as the 
knowledge sets of different coalitions are not of the same type.
In particular, the knowledge of $J$ about a (larger) coalition $K$ 
is transferred from the knowledge of $J \cup S$ about the coalition $K \cup S$.
We present a lifting operator that transforms the knowledge set of $J \cup S$ 
into a knowledge set for $J$. The definition is inductive, assuming 
that the operator is defined for all coalitions larger than $J$.
We define the function $\lift^S_J: \Psi_{J \cup S} \to \Psi_{J}$,
for all $\psi \in \Psi_{J \cup S}$, as follows:

\begin{itemize} 
\item if $J = J\cup S$, then $\lift$ is the identity: $\lift^S_J(\psi) = \psi$;

\item otherwise (i.e., $J \subsetneq J\cup S$), we proceed recursively:

%\begin{linenomath*}
$$\lift^S_J(\psi) = \{\clift^S_J(\varphi,\psi) \mid \varphi \in \psi\},$$
%\end{linenomath*}

where $\clift^S_J(\varphi,\psi)$ is a tuple of knowledge sets, 
one for each coalition $K \in \up{J}$ larger than $J$:

\begin{comment}
\begin{linenomath*}
$$\begin{array}{l}
  \clift^S_J(\varphi,\psi)[K] = \phantom{spaaaaaaaaaaaaaaaaaaaaaaaaace}
  \end{array}
$$
$$
\phantom{spaaaace}
\begin{cases}
\lift^S_K(\psi)              & \text{if } K \cup S = J \cup S,\\
\lift^S_K(\varphi[K \cup S]) & \text{otherwise (i.e., } K \cup S \supsetneq J \cup S\text{)}.
\end{cases}
$$
\end{linenomath*}
\end{comment}

 $$\clift^S_J(\varphi,\psi)[K] = 
         \begin{cases}
         \lift^S_K(\psi)              & \text{if } K \cup S = J \cup S,\\
 	\lift^S_K(\varphi[K \cup S]) & \text{otherwise (i.e., } K \cup S \supsetneq J \cup S\text{)}.
 	\end{cases}$$

\end{itemize} 
%\mynote{L: find a good name for function $\clift$. T: In a macro, still to do.}

Note that $\varphi \in \psi \in \Psi_{J \cup S}$ and therefore the knowledge
set $\varphi[K \cup S]$ is well defined only for $K \cup S \supsetneq J \cup S$
and the knowledge set for $K \cup S = J \cup S$ is given by $\psi$ itself.

We illustrate this definition with an example of configuration 
in \figurename~\ref{fig:configuration}. With three players $I = \{0,1,2\}$,
there are four coalitions containing player~$0$: the singleton $\{0\}$,
the grand coalition $\{0,1,2\}$, and the two coalitions $\{0,1\}$ and $\{0,2\}$.
The coalition $\{0,1\}$ (and therefore also the grand coalition $\{0,1,2\}$) 
knows the current state, namely $q_3$. However, the coalition $\{0,2\}$ know only that the current state is either $q_2$ or $q_3$, and Player~$0$ sees three possibilities:
the current state is $q_1$ and coalition $\{0,2\}$ knows it, 
or the current state is $q_2$ and neither coalition $\{0,1\}$ nor coalition $\{0,2\}$ knows it, 
or the current state is $q_3$, and coalition $\{0,1\}$ knows it.
Note that if it is a possibility for Player~$0$ that a 
coalition $J \supsetneq \{0\}$ sees $k$ possibilities (for example $k=2$ with $J = \{0,1\}$
in the left branch in \figurename~\ref{fig:configuration}), then those $k$ possibilities
should appear in the configuration for Player~$0$ (as the left and middle branch in 
our example).

The lifting of this configuration after Player~$0$ communicates with Player~$1$
is shown in \figurename~\ref{fig:lift-configuration1}. Intuitively, the effect 
of the lifting can be understood as replacing every coalition $J$ containing $0$
by $J \cup \{1\}$. For example, what Player~$0$ knows about coalition 
$\{0,2\}$ after the communication, is what  coalition $\{0,1\} = \{0\} \cup \{1\}$ knows 
about the grand coalition $\{0,1,2\} = \{0,2\} \cup \{1\}$. It turns out in this
case that all coalitions know the current state is $q_3$. 
\figurename~\ref{fig:lift-configuration2} show the lifting if Player~$0$ 
communicates with Player~$2$ instead.

\begin{remark}\label{rmk:lift}
It follows immediately from the definition that $\lift^S_J = \lift^{J \cup S}_J$,
for all coalitions $J,S$. It is then easy to show that, equivalently, 
if $J \cup S = J \cup T$, then $\lift^S_J = \lift^{T}_J$. 
We use this property in the latter form.
\end{remark}

%\footnote{Note that $J\cup S = S$ if $S = \sync_J(c)$.}

%Note that, given a move $c \in \Gamma$, the coalition $J$ synchronizes
%with $S = \sync_J(c)$ and we have $J \subseteq S$. Hence 
%$\lift^S_J = \lift^{J \cup S}_J$ by the previous remark.

The $\lift$ function is compositional:
lifting for a coalition~$J$ that synchronises with $S \cup T$
can be obtained by first lifting for the coalition $J\cup S$ 
synchronizing with~$T$, and then lifting for the coalition~$J$ 
synchronizing with~$S$.

\begin{figure*}[t]
\begin{center}
\hrule
%\renewcommand{\sb}[1]{\scalebox{0.75}[1]{#1}}

%{\scriptsize 
\begin{picture}(152,45)(0,0)
%\put(0,0){\framebox(152,45){}}

%\gasset{Nw=9,Nh=9,Nmr=4.5,rdist=1, loopdiam=6}
%\gasset{Nw=10,Nh=6,Nmr=3,rdist=1, loopdiam=5}
\gasset{Nw=6,Nh=5,Nmr=4, rdist=1, loopdiam=5, ELdist=0, AHnb=0}
%\gasset{Nw=5,Nh=5,Nmr=2.5,rdist=1, loopdiam=6, linewidth=0.12}

\node[Nmarks=n, Nframe=y, Nmr=0, Nw=150](root)(76,40){}
\node[Nmarks=n, Nframe=n, Nmr=0](n1)(56,40){{\scriptsize $0$}}

\node[Nmarks=n, Nframe=n, Nmr=0, Nw=14](s1)(9,25){{\scriptsize $01$}}
\node[Nmarks=n, Nframe=y, Nmr=0, Nw=8](s2)(21,25){{\scriptsize $02$}}
\node[Nmarks=n, Nframe=n, Nmr=0, Nw=8](s3)(29,25){{\scriptsize $012$}}
\node[Nmarks=n, Nframe=y, Nmr=0, Nw=31](s)(17,25){}    %Nw=32
%\drawline(20,22.5)(20,27.5)
%\drawline(30,22.5)(30,27.5)

\node[Nmarks=n, Nframe=y, Nmr=0](I1)(5,15){{\scriptsize $012$}}
\node[Nmarks=n, Nframe=y, Nmr=0](I2)(13,15){{\scriptsize $012$}}
\node[Nmarks=n, Nframe=y, Nmr=0](I3)(21,15){{\scriptsize $012$}}
\node[Nmarks=n, Nframe=n](l4)(29,15){$q_1$}

\node[Nmarks=n, Nframe=n](l1)(5,5){$q_1$}
\node[Nmarks=n, Nframe=n](l2)(13,5){$q_2$}
\node[Nmarks=n, Nframe=n](l3)(21,5){$q_1$}

\drawedge[ELpos=50, ELside=l, syo=-2, eyo=2.5, curvedepth=0](n1,s){}

\drawedge[ELpos=50, ELside=l, curvedepth=0](s1,I1){}
\drawedge[ELpos=50, ELside=l, curvedepth=0](s1,I2){}
\drawedge[ELpos=50, ELside=l, curvedepth=0](s2,I3){}
\drawedge[ELpos=50, ELside=l, curvedepth=0](s3,l4){}

\drawedge[ELpos=50, ELside=l, curvedepth=0](I1,l1){}
\drawedge[ELpos=50, ELside=l, curvedepth=0](I2,l2){}
\drawedge[ELpos=50, ELside=l, curvedepth=0](I3,l3){}

% % % % % % % % % % % % % % % % % % % % % % % % % % % % % % % % % % % % % % % % % % 
% % % 
% % % % % % % % % % % % % % % % % % % % % % % % % % % % % % % % % % % % % % % % % % 

\node[Nmarks=n, Nframe=n, Nmr=0, Nw=20](s1)(45,25){{\scriptsize $01$}}
\node[Nmarks=n, Nframe=y, Nmr=0, Nw=16](s2)(60,25){{\scriptsize $02$}}
\node[Nmarks=n, Nframe=n, Nmr=0, Nw=10](s3)(72,25){{\scriptsize $012$}}
\node[Nmarks=n, Nframe=y, Nmr=0, Nw=39](s)(56,25){}
%\drawline(20,22.5)(20,27.5)
%\drawline(30,22.5)(30,27.5)

\node[Nmarks=n, Nframe=y, Nmr=0](I1)(40,15){{\scriptsize $012$}}
\node[Nmarks=n, Nframe=y, Nmr=0](I2)(48,15){{\scriptsize $012$}}
\node[Nmarks=n, Nframe=y, Nmr=0](I3)(56,15){{\scriptsize $012$}}
\node[Nmarks=n, Nframe=y, Nmr=0](I4)(64,15){{\scriptsize $012$}}
\node[Nmarks=n, Nframe=n](l5)(72,15){$q_2$}

\node[Nmarks=n, Nframe=n](l1)(40,5){$q_1$}
\node[Nmarks=n, Nframe=n](l2)(48,5){$q_2$}
\node[Nmarks=n, Nframe=n](l3)(56,5){$q_2$}
\node[Nmarks=n, Nframe=n](l4)(64,5){$q_3$}

\drawedge[ELpos=50, ELside=l, syo=-2, eyo=2.5, curvedepth=0](n1,s){}

\drawedge[ELpos=50, ELside=l, curvedepth=0](s1,I1){}
\drawedge[ELpos=50, ELside=l, curvedepth=0](s1,I2){}
\drawedge[ELpos=50, ELside=l, curvedepth=0](s2,I3){}
\drawedge[ELpos=50, ELside=l, curvedepth=0](s2,I4){}
\drawedge[ELpos=50, ELside=l, curvedepth=0](s3,l5){}

\drawedge[ELpos=50, ELside=l, curvedepth=0](I1,l1){}
\drawedge[ELpos=50, ELside=l, curvedepth=0](I2,l2){}
\drawedge[ELpos=50, ELside=l, curvedepth=0](I3,l3){}
\drawedge[ELpos=50, ELside=l, curvedepth=0](I4,l4){}

% % % % % % % % % % % % % % % % % % % % % % % % % % % % % % % % % % % % % % % % % % 
% % % 
% % % % % % % % % % % % % % % % % % % % % % % % % % % % % % % % % % % % % % % % % % 

\node[Nmarks=n, Nframe=n, Nmr=0, Nw=10](s1)(83,25){{\scriptsize $01$}}
\node[Nmarks=n, Nframe=y, Nmr=0, Nw=16](s2)(95,25){{\scriptsize $02$}}
\node[Nmarks=n, Nframe=n, Nmr=0, Nw=10](s3)(107,25){{\scriptsize $012$}}
\node[Nmarks=n, Nframe=y, Nmr=0, Nw=31](s)(95,25){}
%\drawline(20,22.5)(20,27.5)
%\drawline(30,22.5)(30,27.5)

\node[Nmarks=n, Nframe=y, Nmr=0](I1)(83,15){{\scriptsize $012$}}
\node[Nmarks=n, Nframe=y, Nmr=0](I2)(91,15){{\scriptsize $012$}}
\node[Nmarks=n, Nframe=y, Nmr=0](I3)(99,15){{\scriptsize $012$}}
\node[Nmarks=n, Nframe=n](l4)(107,15){$q_3$}

\node[Nmarks=n, Nframe=n](l1)(83,5){$q_3$}
\node[Nmarks=n, Nframe=n](l2)(91,5){$q_2$}
\node[Nmarks=n, Nframe=n](l3)(99,5){$q_3$}

\drawedge[ELpos=50, ELside=l, syo=-2, eyo=2.5, curvedepth=0](n1,s){}

\drawedge[ELpos=50, ELside=l, curvedepth=0](s1,I1){}
\drawedge[ELpos=50, ELside=l, curvedepth=0](s2,I2){}
\drawedge[ELpos=50, ELside=l, curvedepth=0](s2,I3){}
\drawedge[ELpos=50, ELside=l, curvedepth=0](s3,l4){}

\drawedge[ELpos=50, ELside=l, curvedepth=0](I1,l1){}
\drawedge[ELpos=50, ELside=l, curvedepth=0](I2,l2){}
\drawedge[ELpos=50, ELside=l, curvedepth=0](I3,l3){}

% % % % % % % % % % % % % % % % % % % % % % % % % % % % % % % % % % % % % % % % % % 
% % % 
% % % % % % % % % % % % % % % % % % % % % % % % % % % % % % % % % % % % % % % % % % 

\node[Nmarks=n, Nframe=y, Nmr=0, Nw=10](s1)(123,40){{\scriptsize $01$}}
\node[Nmarks=n, Nframe=y, Nmr=0, Nw=14](s2)(135,40){{\scriptsize $02$}}
\node[Nmarks=n, Nframe=n, Nmr=0, Nw=10](s3)(147,40){{\scriptsize $012$}}
\node[Nmarks=n, Nframe=n, Nmr=0, Nw=31](s)(135,40){}
%\drawline(20,22.5)(20,27.5)
%\drawline(30,22.5)(30,27.5)

\node[Nmarks=n, Nframe=y, Nmr=0](I1)(123,25){{\scriptsize $012$}}
\node[Nmarks=n, Nframe=y, Nmr=0](I2)(131,25){{\scriptsize $012$}}
\node[Nmarks=n, Nframe=y, Nmr=0](I3)(139,25){{\scriptsize $012$}}
\node[Nmarks=n, Nframe=n](l4)(147,25){$q_3$}

\node[Nmarks=n, Nframe=n](l1)(123,15){$q_3$}
\node[Nmarks=n, Nframe=n](l2)(131,15){$q_2$}
\node[Nmarks=n, Nframe=n](l3)(139,15){$q_3$}

%\drawedge[ELpos=50, ELside=l, syo=-2.5, eyo=2.5, curvedepth=0](n1,s){}

\drawedge[ELpos=50, ELside=l, curvedepth=0](s1,I1){}
\drawedge[ELpos=50, ELside=l, curvedepth=0](s2,I2){}
\drawedge[ELpos=50, ELside=l, curvedepth=0](s2,I3){}
\drawedge[ELpos=50, ELside=l, curvedepth=0](s3,l4){}

\drawedge[ELpos=50, ELside=l, curvedepth=0](I1,l1){}
\drawedge[ELpos=50, ELside=l, curvedepth=0](I2,l2){}
\drawedge[ELpos=50, ELside=l, curvedepth=0](I3,l3){}

%\node[Nmarks=n](n1)(30,60){$\ver_1$}
%\node[Nmarks=n, Nmr=0](n11)(15,45){$\ver_2$}
%\nodelabel[ExtNL=y, NLangle=270, NLdist=1](n111){{\small val=1}}

%\drawedge[ELpos=50, ELside=r, curvedepth=0, AHLength = 2, AHangle=25, AHlength = 1.81](n1,n11){$-1$}
%\drawedge[ELpos=50, ELside=r, curvedepth=0, AHnb=0, linewidth=.6](n1,n11dummy){}

%\drawedge[ELpos=50, ELside=l, curvedepth=-12](t3,n1){}
%\drawbpedge(t3,70,22,n1,110,22){}
%\drawline[AHnb=1,arcradius=1](113,17.5)(113,29)(5,29)(5,17.5)

%\drawloop[ELside=l,loopCW=y, loopdiam=6](n4){$1$}

%\drawloop[ELside=l,loopCW=y](nk){$0,1$}

%\drawedge[dash={1}0](n3bis,nkbis){$0,1$}

\end{picture}
%}
\hrule
\caption{A configuration.\label{fig:configuration}}
\end{center}
\end{figure*}
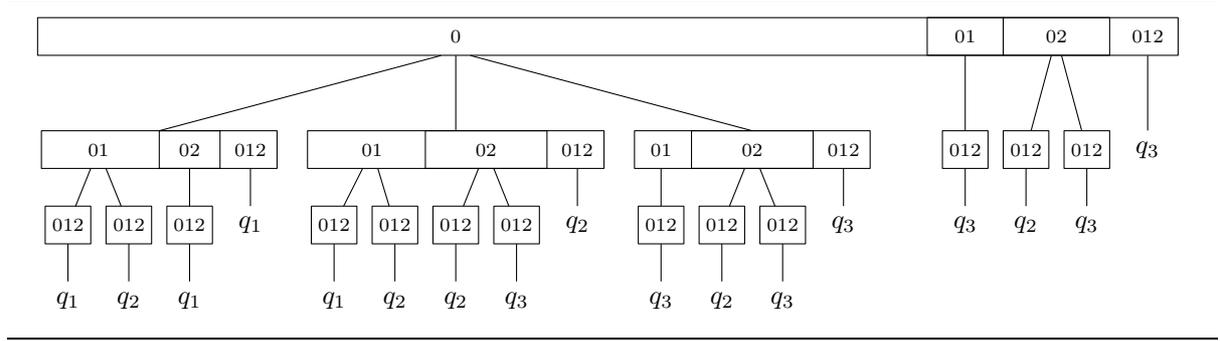

\begin{figure}[t]
\begin{center}
\hrule
%\renewcommand{\sb}[1]{\scalebox{0.75}[1]{#1}}

%{\scriptsize 
\begin{picture}(26,45)(0,0)
%\put(0,0){\framebox(26,45){}}

%\gasset{Nw=9,Nh=9,Nmr=4.5,rdist=1, loopdiam=6}
%\gasset{Nw=10,Nh=6,Nmr=3,rdist=1, loopdiam=5}
\gasset{Nw=6,Nh=5,Nmr=4, rdist=1, loopdiam=5, ELdist=0, AHnb=0}
%\gasset{Nw=5,Nh=5,Nmr=2.5,rdist=1, loopdiam=6, linewidth=0.12}

%\node[Nmarks=n, Nframe=y, Nmr=0, Nw=150](root)(76,40){}
\node[Nmarks=n, Nframe=y, Nmr=0, Nw=24](n1)(13,40){{\scriptsize $0$}}

\node[Nmarks=n, Nframe=n, Nmr=0, Nw=10](s1)(5,25){{\scriptsize $01$}}
\node[Nmarks=n, Nframe=y, Nmr=0, Nw=8](s2)(13,25){{\scriptsize $02$}}
\node[Nmarks=n, Nframe=n, Nmr=0, Nw=10](s3)(21,25){{\scriptsize $012$}}
\node[Nmarks=n, Nframe=y, Nmr=0, Nw=24](s)(13,25){}
%\drawline(20,22.5)(20,27.5)
%\drawline(30,22.5)(30,27.5)

\node[Nmarks=n, Nframe=y, Nmr=0](I1)(5,15){{\scriptsize $012$}}
\node[Nmarks=n, Nframe=y, Nmr=0](I2)(13,15){{\scriptsize $012$}}
\node[Nmarks=n, Nframe=n](l3)(21,15){$q_3$}

\node[Nmarks=n, Nframe=n](l1)(5,5){$q_3$}
\node[Nmarks=n, Nframe=n](l2)(13,5){$q_3$}

\drawedge[ELpos=50, ELside=l, syo=-2, eyo=2.5, curvedepth=0](n1,s){}

\drawedge[ELpos=50, ELside=l, curvedepth=0](s1,I1){}
\drawedge[ELpos=50, ELside=l, curvedepth=0](s2,I2){}
\drawedge[ELpos=50, ELside=l, curvedepth=0](s3,l3){}

\drawedge[ELpos=50, ELside=l, curvedepth=0](I1,l1){}
\drawedge[ELpos=50, ELside=l, curvedepth=0](I2,l2){}

% % % % % % % % % % % % % % % % % % % % % % % % % % % % % % % % % % % % % % % % % % 
% % % 
% % % % % % % % % % % % % % % % % % % % % % % % % % % % % % % % % % % % % % % % % % 

%\node[Nmarks=n](n1)(30,60){$\ver_1$}
%\node[Nmarks=n, Nmr=0](n11)(15,45){$\ver_2$}
%\nodelabel[ExtNL=y, NLangle=270, NLdist=1](n111){{\small val=1}}

%\drawedge[ELpos=50, ELside=r, curvedepth=0, AHLength = 2, AHangle=25, AHlength = 1.81](n1,n11){$-1$}
%\drawedge[ELpos=50, ELside=r, curvedepth=0, AHnb=0, linewidth=.6](n1,n11dummy){}

%\drawedge[ELpos=50, ELside=l, curvedepth=-12](t3,n1){}
%\drawbpedge(t3,70,22,n1,110,22){}
%\drawline[AHnb=1,arcradius=1](113,17.5)(113,29)(5,29)(5,17.5)

%\drawloop[ELside=l,loopCW=y, loopdiam=6](n4){$1$}

%\drawloop[ELside=l,loopCW=y](nk){$0,1$}

%\drawedge[dash={1}0](n3bis,nkbis){$0,1$}

\end{picture}
%}
\hrule
\caption{Lifting of the configuration of \figurename~\ref{fig:configuration} after Player~$0$ synchronises with Player~$1$.\label{fig:lift-configuration1}}
\end{center}
\end{figure}

\begin{lemma}\label{lem:lift-compositional}
$\lift^S_{J} \circ \lift^T_{J \cup S}= \lift^{S \cup T}_J$ for all coalitions $J,S,T$.
\end{lemma}

\begin{figure}[t]
\begin{center}
\hrule
%\renewcommand{\sb}[1]{\scalebox{0.75}[1]{#1}}

%{\scriptsize 
\begin{picture}(69,45)(0,0)
%\put(0,0){\framebox(69,45){}}

%\gasset{Nw=9,Nh=9,Nmr=4.5,rdist=1, loopdiam=6}
%\gasset{Nw=10,Nh=6,Nmr=3,rdist=1, loopdiam=5}
\gasset{Nw=6,Nh=5,Nmr=4, rdist=1, loopdiam=5, ELdist=0, AHnb=0}
%\gasset{Nw=5,Nh=5,Nmr=2.5,rdist=1, loopdiam=6, linewidth=0.12}

%\node[Nmarks=n, Nframe=y, Nmr=0, Nw=150](root)(76,40){}
\node[Nmarks=n, Nframe=y, Nmr=0, Nw=67](n1)(34.5,40){{\scriptsize $0$}}

\node[Nmarks=n, Nframe=n, Nmr=0, Nw=10](s1)(5,25){{\scriptsize $01$}}
\node[Nmarks=n, Nframe=y, Nmr=0, Nw=16](s2)(17,25){{\scriptsize $02$}}
\node[Nmarks=n, Nframe=n, Nmr=0, Nw=10](s3)(29,25){{\scriptsize $012$}}
\node[Nmarks=n, Nframe=y, Nmr=0, Nw=31](s)(17,25){}
%\drawline(20,22.5)(20,27.5)
%\drawline(30,22.5)(30,27.5)

\node[Nmarks=n, Nframe=y, Nmr=0](I1)(5,15){{\scriptsize $012$}}
\node[Nmarks=n, Nframe=y, Nmr=0](I2)(13,15){{\scriptsize $012$}}
\node[Nmarks=n, Nframe=y, Nmr=0](I3)(21,15){{\scriptsize $012$}}
\node[Nmarks=n, Nframe=n](l4)(29,15){$q_2$}

\node[Nmarks=n, Nframe=n](l1)(5,5){$q_2$}
\node[Nmarks=n, Nframe=n](l2)(13,5){$q_2$}
\node[Nmarks=n, Nframe=n](l3)(21,5){$q_3$}

\drawedge[ELpos=50, ELside=l, syo=-2, eyo=2.5, curvedepth=0](n1,s){}

\drawedge[ELpos=50, ELside=l, curvedepth=0](s1,I1){}
\drawedge[ELpos=50, ELside=l, curvedepth=0](s2,I2){}
\drawedge[ELpos=50, ELside=l, curvedepth=0](s2,I3){}
\drawedge[ELpos=50, ELside=l, curvedepth=0](s3,l4){}

\drawedge[ELpos=50, ELside=l, curvedepth=0](I1,l1){}
\drawedge[ELpos=50, ELside=l, curvedepth=0](I2,l2){}
\drawedge[ELpos=50, ELside=l, curvedepth=0](I3,l3){}

% % % % % % % % % % % % % % % % % % % % % % % % % % % % % % % % % % % % % % % % % % 
% % % 
% % % % % % % % % % % % % % % % % % % % % % % % % % % % % % % % % % % % % % % % % % 

\node[Nmarks=n, Nframe=n, Nmr=0, Nw=10](s1)(40,25){{\scriptsize $01$}}
\node[Nmarks=n, Nframe=y, Nmr=0, Nw=16](s2)(52,25){{\scriptsize $02$}}
\node[Nmarks=n, Nframe=n, Nmr=0, Nw=10](s3)(64,25){{\scriptsize $012$}}
\node[Nmarks=n, Nframe=y, Nmr=0, Nw=31](s)(52,25){}
%\drawline(20,22.5)(20,27.5)
%\drawline(30,22.5)(30,27.5)

\node[Nmarks=n, Nframe=y, Nmr=0](I1)(40,15){{\scriptsize $012$}}
\node[Nmarks=n, Nframe=y, Nmr=0](I2)(48,15){{\scriptsize $012$}}
\node[Nmarks=n, Nframe=y, Nmr=0](I3)(56,15){{\scriptsize $012$}}
\node[Nmarks=n, Nframe=n](l4)(64,15){$q_3$}

\node[Nmarks=n, Nframe=n](l1)(40,5){$q_3$}
\node[Nmarks=n, Nframe=n](l2)(48,5){$q_2$}
\node[Nmarks=n, Nframe=n](l3)(56,5){$q_3$}

\drawedge[ELpos=50, ELside=l, syo=-2, eyo=2.5, curvedepth=0](n1,s){}

\drawedge[ELpos=50, ELside=l, curvedepth=0](s1,I1){}
\drawedge[ELpos=50, ELside=l, curvedepth=0](s2,I2){}
\drawedge[ELpos=50, ELside=l, curvedepth=0](s2,I3){}
\drawedge[ELpos=50, ELside=l, curvedepth=0](s3,l4){}

\drawedge[ELpos=50, ELside=l, curvedepth=0](I1,l1){}
\drawedge[ELpos=50, ELside=l, curvedepth=0](I2,l2){}
\drawedge[ELpos=50, ELside=l, curvedepth=0](I3,l3){}

% % % % % % % % % % % % % % % % % % % % % % % % % % % % % % % % % % % % % % % % % % 
% % % 
% % % % % % % % % % % % % % % % % % % % % % % % % % % % % % % % % % % % % % % % % % 

%\node[Nmarks=n](n1)(30,60){$\ver_1$}
%\node[Nmarks=n, Nmr=0](n11)(15,45){$\ver_2$}
%\nodelabel[ExtNL=y, NLangle=270, NLdist=1](n111){{\small val=1}}

%\drawedge[ELpos=50, ELside=r, curvedepth=0, AHLength = 2, AHangle=25, AHlength = 1.81](n1,n11){$-1$}
%\drawedge[ELpos=50, ELside=r, curvedepth=0, AHnb=0, linewidth=.6](n1,n11dummy){}

%\drawedge[ELpos=50, ELside=l, curvedepth=-12](t3,n1){}
%\drawbpedge(t3,70,22,n1,110,22){}
%\drawline[AHnb=1,arcradius=1](113,17.5)(113,29)(5,29)(5,17.5)

%\drawloop[ELside=l,loopCW=y, loopdiam=6](n4){$1$}

%\drawloop[ELside=l,loopCW=y](nk){$0,1$}

%\drawedge[dash={1}0](n3bis,nkbis){$0,1$}

\end{picture}
%}
\hrule
\caption{Lifting of the configuration of \figurename~\ref{fig:configuration} after Player~$0$ synchronises with Player~$2$.\label{fig:lift-configuration2}}
\end{center}
\end{figure}
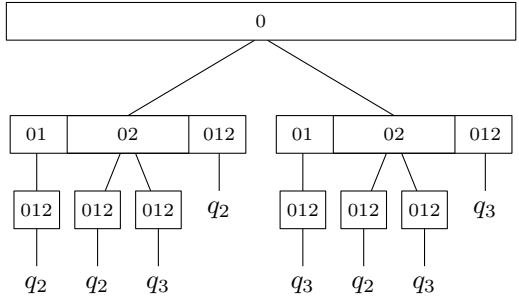

\begin{longversion}
\begin{proof}
The proof is by (descending) induction on $\abs{J}$. The base case $J = I$
is trivial since all three operators $\lift^S_{J}$, $\lift^T_{J \cup S}$, 
and $\lift^{S \cup T}_J$ are then the identity.

For the induction case, assume that the lemma holds for all coalitions 
of cardinality larger than $J$ (in particular for all $K\in \up{J}$) 
and show that it holds for coalition~$J$. % $J\cup S \neq J$ and $J\cup S \cup T \neq J \cup S$.

\begin{comment}
\begin{linenomath*}
\begin{align*}
& \lift_J^S \circ \lift_{J\cup S}^T(\psi) \\
& =\{\clift^S_{J}(\varphi,\lift^{T}_{J\cup S}(\psi))\,|\,\varphi\in\lift^{T}_{J\cup S}(\psi) \},\\
& =\{\clift^S_{J}(\varphi,\lift^{T}_{J\cup S}(\psi))\,|\,\varphi\in\{\clift^T_{J\cup S}(\varphi,\psi)\,|\, \varphi\in \psi \} \},\\
& =\{\clift^S_J(\clift^T_{J\cup S}(\varphi,\psi),\lift^{T}_{J\cup S}(\psi) )\,|\, \varphi\in \psi \},
\end{align*}		
\end{linenomath*}
\end{comment}

	\begin{align*}
	\lift_J^S \circ \lift_{J\cup S}^T(\psi)
	& =\{\clift^S_{J}(\varphi,\lift^{T}_{J\cup S}(\psi))\,|\,\varphi\in\lift^{T}_{J\cup S}(\psi) \},\\
	& =\{\clift^S_{J}(\varphi,\lift^{T}_{J\cup S}(\psi))\,|\,\varphi\in\{\clift^T_{J\cup S}(\varphi,\psi)\,|\, \varphi\in \psi \} \},\\
	& =\{\clift^S_J(\clift^T_{J\cup S}(\varphi,\psi),\lift^{T}_{J\cup S}(\psi) )\,|\, \varphi\in \psi \},
	\end{align*}

	and considering each $K\in \up{J}$:
	\begin{itemize}
	\item if $K\cup S = J\cup S$,
	\begin{align*}
		\clift^S_J(\clift^T_{J\cup S}(\varphi,\psi),\lift^{T}_{J\cup S}(\psi))[K] 
		 &=\lift^{S}_K(\lift^{T}_{J\cup S}(\psi)), & \text{by definition of }\clift^S_J,\\
		 &=\lift^{S}_K(\lift^{T}_{K\cup S}(\psi)), & \text{as }K\cup S= J\cup S,\\
		 &=\lift^{S \cup T}_K(\psi), & \text{by induction hypothesis}.
	\end{align*}

\begin{comment}
%\begin{linenomath*}
	$$\begin{array}{lr}
		\multicolumn{2}{l}{\clift^S_J(\clift^T_{J\cup S}(\varphi,\psi),\lift^{T}_{J\cup S}(\psi))[K]}  \\[+3pt] 
		=\lift^{S}_K(\lift^{T}_{J\cup S}(\psi)) & \text{by definition of }\clift^S_J, \\[+3pt]
		=\lift^{S}_K(\lift^{T}_{K\cup S}(\psi)) & \text{as }K\cup S= J\cup S,         \\[+3pt]
		=\lift^{S \cup T}_K(\psi), & \text{by induction hypothesis}.
	\end{array}$$
\end{comment}	
        \smallskip

	\item if $K\cup S\neq J\cup S$ and $K\cup S \cup T=J\cup S\cup T$,
\begin{comment}
	$$\begin{array}{lr}
		\multicolumn{2}{l}{\clift^S_J(\clift^T_{J\cup S}(\varphi,\psi),\lift^{T}_{J\cup S}(\psi))[K]} \\[+3pt] 
		\multicolumn{2}{l}{=\lift^{S}_K(\clift^T_{J\cup S}(\varphi,\psi)[K\cup S])}  \\[+3pt]
		& \text{by definition of }\clift^S_J,\\[+3pt]
		=\lift^{S}_K(\lift^{T}_{K\cup S}(\psi)), & \text{by definition of }\clift^T_{J\cup S},\\[+3pt]
		=\lift^{S \cup T}_K(\psi) & \text{by induction hypothesis.}\\[+3pt]
        \end{array}$$			
	\smallskip
\end{comment}			 
	\begin{align*}
		\clift^S_J(\clift^T_{J\cup S}(\varphi,\psi),\lift^{T}_{J\cup S}(\psi))[K] 
		 &=\lift^{S}_K(\clift^T_{J\cup S}(\varphi,\psi)[K\cup S]), & \text{by definition of }\clift^S_J,\\
		 &=\lift^{S}_K(\lift^{T}_{K\cup S}(\psi)), & \text{by definition of }\clift^T_{J\cup S},\\
		 &=\lift^{S \cup T}_K(\psi) & \text{by induction hypothesis.}
        \end{align*}

	\item if $K\cup S\neq J\cup S$ and $K\cup S \cup T\neq J\cup S\cup T$,
% 
\begin{comment}
	$$\begin{array}{lr}
		\multicolumn{2}{l}{\clift^S_J(\clift^T_{J\cup S}(\varphi,\psi),\lift^{T}_{J\cup S}(\psi))[K]  \phantom{spaaaaaaaaace}}\\ [+3pt]
		\multicolumn{2}{l}{=\lift^{S}_K(\clift^T_{J\cup S}(\varphi,\psi)[K \cup S])} \\[+3pt] 
		\multicolumn{2}{r}{\text{by definition of }\clift^S_J,}\\[+3pt]
		\multicolumn{2}{l}{=\lift^{S}_K(\lift^{T}_{K\cup S}(\varphi[K\cup S\cup T]))} \\[+3pt] 
		\multicolumn{2}{r}{\text{by definition of }\clift^T_{J\cup S}}\\[+3pt]
		\multicolumn{2}{l}{=\lift^{S \cup T}_K(\varphi[K\cup S\cup T])}\\[+3pt]
		\multicolumn{2}{r}{\text{by induction hypothesis.}} \\[+3pt]
	\end{array}$$
%\end{linenomath*}	
	\smallskip	
\end{comment}

	\begin{align*}
		\clift^S_J(\clift^T_{J\cup S}(\varphi,\psi),\lift^{T}_{J\cup S}(\psi))[K] 
		 &=\lift^{S}_K(\clift^T_{J\cup S}(\varphi,\psi)[K \cup S]), & \text{by definition of }\clift^S_J,\\
		 &=\lift^{S}_K(\lift^{T}_{K\cup S}(\varphi[K\cup S\cup T])), & \text{by definition of }\clift^T_{J\cup S},\\
		 &=\lift^{S \cup T}_K(\varphi[K\cup S\cup T]) & \text{by induction hypothesis.}
	\end{align*}

	\end{itemize}
	%\mynote{T: Must we keep the three case, that are very similar.}
	
	In summary we get:

	%\begin{comment}
	%\begin{linenomath*} 
	%\[\clift^S_{J}(\clift^T_{J\cup S}(\varphi,\psi),\lift^{T}_{J\cup S}(\psi))[K]=\phantom{spaaaaaaaaaaaaaace}\]
	%\[
	%\phantom{spaaaace}\begin{cases}
	%\lift^{S \cup T}_K(\psi), & \text{if }K\cup S \cup T=J\cup S \cup T,\\
	%\lift^{S \cup T}_K(\varphi[K\cup S\cup T]), & \text{if }K\cup S \cup T \neq J\cup S \cup T,\\
	%\end{cases}
	%\]
	%\end{linenomath*}
	%\end{comment}	

	\[\clift^S_{J}(\clift^T_{J\cup S}(\varphi,\psi),\lift^{T}_{J\cup S}(\psi))[K]=\begin{cases}
	\lift^{S \cup T}_K(\psi), & \text{if }K\cup S \cup T=J\cup S \cup T,\\
	\lift^{S \cup T}_K(\varphi[K\cup S\cup T]), & \text{if }K\cup S \cup T \neq J\cup S \cup T,\\
	\end{cases}
	\]

	and so $\clift^T_{J\cup S}(\clift^S_{J}(\varphi,\psi),\lift^{T}_{J\cup S}(\psi))=\clift^{T\cup S}_{J}(\varphi,\psi)$, 
which concludes the proof.
	\end{proof}
\end{longversion}

%\subsubsection{The update function $\Delta$}

The update function $\Delta: P \times \Gamma \to P$ is defined component-wise for each
coalition. The definition is recursive: given a coalition $J$, we first update
all coalitions $K\in \up{J}$ greater than $J$. Then, we update the coalition $J$ as
follows. Given the actual move $c$, let $S = \sync_J(c)$ be the coalition that
transfers their knowledge to $J$ (through communication). The update is calculated
as the (lifting of) the knowledge of coalition $S$ upon reading a move move $d$ 
that the coalition $S$ cannot distinguish from $c$, that is such that the 
observations for the players in $S$ are the same, $d[S] = c[S]$.
Note that for the grand coalition $S=I$, the condition $d[I] = c[I]$ is
equivalent to $d=c$.

Given a state $p \in P$ and a move $c \in \Gamma$, let
$\Delta(P,c) = (\delta^c_J(P[\sync_J(c)]))_{\{0\} \subseteq J \subseteq I}$
where $\delta^c_J$ is defined recursively as follows, for all 
$\psi \in \Psi_{S}$ where $S = \sync_J(c)$:

\begin{itemize} 
\item if $S = I$ is the grand coalition, then we update according to the last 
observations of $I$ followed by a lifting:
%\begin{linenomath*}
$$\delta^c_J(\psi) = \lift^I_J(\{\delta_{\M}(q,c) \mid q \in \psi\});$$
%\end{linenomath*}
%\mynote{T: must we keep the $d_I = c_I$, when with our definition, it means $c=d$}
\item otherwise, we lift the knowledge set of $S$, which is defined recursively:

$$
  \delta^c_J(\psi) = \lift^S_J(\{(\delta^d_K(\varphi[\sync_K(d)]))_{K \in \up{S}} \mid \varphi \in \psi, d[S] = c[S] \}).
$$

\begin{comment}
\begin{linenomath*}
$$
 \begin{array}{lr}
  \delta^c_J(\psi) = \lift^S_J(\{(\delta^d_K(\varphi[\sync_K(d)]))_{K \in \up{S}} & \phantom{space}\\[+4pt]
  \multicolumn{2}{r}{\mid \varphi \in \psi, d[S] = c[S] \}).} \\[+3pt]
 \end{array}
$$
\end{linenomath*}
\end{comment}

%$$\delta^c_J(\psi) = \lift^S_J(\{(\delta^d_K(\varphi[\sync_K(d)]))_{K \in \up{S}}  \mid \varphi \in \psi, d[S] = c[S] \}).$$
\mynote{L: alternative presentation: the argument $\sync_K(d)$
in $\delta^d_K(\varphi[\sync_K(d)])$ sounds redundant.}

\end{itemize}

\begin{longversion}
\begin{remark}\label{rmk:update}
We often use $\delta^c_J(\psi)$ with argument of the form 
$\psi = h(\tau)[S]$ with $\tau \in \Gamma^*$ and $S = \sync_J(c)$,
which by unfolding the definition of $h$ gives:

\begin{itemize} 
\item if $S = I$, then 

%\begin{linenomath*}
	$$\arraycolsep=1.5pt
	  \begin{array}{rlr}
	  \delta^c_J(h(\tau)[S]) & = \lift^S_J(\{\delta_{\M}(q,c) \mid q \in h(\tau)[S]\}) &  \\[+3pt]
	  & = \lift^S_J(\{\delta_{\M}(q,c) \mid q \in \{\delta_{\M}(q_\epsilon,\tau) \}\}) & \\[+3pt]
	  & \multicolumn{2}{r}{\text{since } h(\tau)[S] = h(\tau)[I]} \\[+3pt]
	  & = \lift^S_J(\{\delta_{\M}(\delta_{\M}(q_\epsilon,\tau),c)\}) & \\[+3pt]
	  & = \lift^S_J(\{\delta_{\M}(q_\epsilon,\tau c) \}). & \\[+3pt]
	  %& \multicolumn{2}{r}{\text{by Lemma~\ref{lem:fip-properties}(\ref{lem:fip-properties:P2}) }} \\[+3pt]
	\end{array}$$
	%\end{linenomath*}

\item if $S \neq I$, then 
\begin{comment}
\begin{linenomath*}
	$$\begin{array}{ll}
	  \delta^c_J(h(\tau)[S]) &  \\[+3pt]
	  \multicolumn{2}{l}{=\lift^S_J(\{(\delta^d_K(\varphi[\sync_K(d)]))_{K \in \up{S}}}   \\[+3pt] 
	  \multicolumn{2}{r}{\mid \varphi \in h(\tau)[S], d_S = c_S \})}   \\[+3pt]
	  \multicolumn{2}{l}{=\lift^S_J(\{(\delta^d_K(h(\tau')[\sync_K(d)]))_{K \in \up{S}}}   \\[+3pt] 
	  \multicolumn{2}{r}{\mid \tau' \sim_S \tau, d[S] = c[S] \})} \\[+3pt]
	  \multicolumn{2}{r}{\quad\text{since } h(\tau)[S] = \{ (h(\tau')[K])_{K \in \up{S}} \mid \tau' \sim_S \tau\}.} \\[+3pt]
	  \multicolumn{2}{l}{=\lift^S_J(\{(\delta^d_K(h(\tau')[\sync_K(d)]))_{K \in \up{S}}}  \\[+3pt] 
	  \multicolumn{2}{r}{\mid \tau' d \sim_S \tau c \})} \\[+3pt]
	  \multicolumn{2}{r}{\text{by Lemma~\ref{lem:fip-properties}(\ref{lem:fip-properties:P2}) }} \\[+3pt]
	\end{array}$$
\end{linenomath*}
\end{comment}

	$$\arraycolsep=1.5pt
	  \begin{array}{rlr}
	  \delta^c_J(h(\tau)[S]) & = \lift^S_J(\{(\delta^d_K(\varphi[\sync_K(d)]))_{K \in \up{S}}  \mid \varphi \in h(\tau)[S], d_S = c_S \}) &  \\[+3pt]
	  & = \lift^S_J(\{(\delta^d_K(h(\tau')[\sync_K(d)]))_{K \in \up{S}}  \mid \tau' \sim_S \tau, d[S] = c[S] \}) & \\[+3pt]
	  & \multicolumn{2}{r}{\text{since } h(\tau)[S] = \{ (h(\tau')[K])_{K \in \up{S}} \mid \tau' \sim_S \tau\}.} \\[+3pt]
	  & = \lift^S_J(\{(\delta^d_K(h(\tau')[\sync_K(d)]))_{K \in \up{S}}  \mid \tau' d \sim_S \tau c \}) & \\[+3pt]
	  & \multicolumn{2}{r}{\text{by Lemma~\ref{lem:fip-properties}(\ref{lem:fip-properties:P2}) }} \\[+3pt]
	\end{array}$$

\end{itemize}
\end{remark}
\end{longversion}

%\subsubsection{Correctness proof}

\begin{lemma}
For all histories $\tau \in \Gamma^*$ and moves $c \in \Gamma$, we have:
%\begin{linenomath*}
$$\Delta(h(\tau),c) = h(\tau c).$$
%\end{linenomath*}
\end{lemma}

\begin{corollary}\label{cor:morphic-h}
The function $h$ is a morphism.
\end{corollary}

\begin{longversion}
\begin{proof}
We show that $\Delta(h(\tau),c)[J] = h(\tau c)[J]$ for all coalitions $J$.
We proceed by (descending) induction on the cardinality of $J$.

The base case is for $\abs{J} = \abs{I}$, that is $J = I$:
\begin{comment}
%\begin{linenomath*}
$$
\arraycolsep=1.5pt
\begin{array}{lr}
\Delta(h(\tau),c)[I]  &  \\[+3pt]
          \ \  = \delta^c_I(h(\tau)[I]) &  \\[+3pt]
	  \ \  = \lift^S_I(\{\delta_{\M}(q_\epsilon,\tau c) \}) & \\[+3pt]
	  \multicolumn{2}{r}{\text{by  Remark~\ref{rmk:update} where } S = \sync^*_I(c) = I} \\[+3pt]
	  \ \  = \{\delta_{\M}(q_\epsilon,\tau c) \} & \\[+3pt]
	  \multicolumn{2}{r}{\qquad\text{since } I = I \cup S \text{ and thus } \lift^S_I \text{ is the identity}} \\[+3pt]
%	  & \multicolumn{2}{l}{= \{\delta_{\A}(q_\epsilon,\tau'd) \mid \tau'd \sim_I \tau c\} \hfill \text{by Lemma~\ref{lem:fip-properties}(\ref{lem:fip-properties:P2})}} \\[+3pt]
%	  & \multicolumn{2}{r}{\text{by Lemma~\ref{lem:fip-properties}(\ref{lem:fip-properties:P2})}} \\[+3pt]
	  \multicolumn{2}{l}{= h(\tau c)[I] \hfill \text{by definition of } h}  \\[+3pt]
%	  & =  & \\[+3pt]
%	  & = . & \\[+3pt]
\end{array}
$$
%\end{linenomath*}
\end{comment}
$$
\arraycolsep=1.5pt
\begin{array}{rlr}
\Delta(h(\tau),c)[I] 
          & = \delta^c_I(h(\tau)[I]) &  \\[+3pt]
	  & = \lift^S_I(\{\delta_{\M}(q_\epsilon,\tau c) \}) & \\[+3pt]
	  & \multicolumn{2}{r}{\text{by  Remark~\ref{rmk:update} where } S = \sync^*_I(c) = I} \\[+3pt]
	  & = \{\delta_{\M}(q_\epsilon,\tau c) \} & \\[+3pt]
	  & \multicolumn{2}{r}{\text{since } I = I \cup S \text{ and thus } \lift^S_I \text{ is the identity}} \\[+3pt]
%	  & \multicolumn{2}{l}{= \{\delta_{\A}(q_\epsilon,\tau'd) \mid \tau'd \sim_I \tau c\} \hfill \text{by Lemma~\ref{lem:fip-properties}(\ref{lem:fip-properties:P2})}} \\[+3pt]
%	  & \multicolumn{2}{r}{\text{by Lemma~\ref{lem:fip-properties}(\ref{lem:fip-properties:P2})}} \\[+3pt]
	  & \multicolumn{2}{l}{= h(\tau c)[I] \hfill \text{by definition of } h}  \\[+3pt]
%	  & =  & \\[+3pt]
%	  & = . & \\[+3pt]
\end{array}
$$

For the induction step, consider a coalition $J \subsetneq I$ and assume that the property holds
for all coalitions of cardinality larger than $\abs{J}$, in particular for all $K \in \up{J}$,
that is $\Delta(h(\tau),c)[K] = h(\tau c)[K]$ for all $\tau \in \Gamma^*$ and $c \in \Gamma$.
By definition of $\Delta$, the induction hypothesis boils down to 
$h(\tau c)[K] = \delta^c_K(h(\tau)[\sync_K(c)])$.

Given coalition $J$, history $\tau$, and move $c$, let $S = \sync_J(c)$.
We consider several cases:

\begin{enumerate}

\item if $S = I$, then  % (hence $J \cup S = J$), then

	$$
	\arraycolsep=1.5pt
	\begin{array}{rlr}
	\Delta(h(\tau),c)[J] 
	  & \multicolumn{2}{l}{= \delta^c_J(h(\tau)[I]) \hfill \text{ since } J \cup S = I}  \\[+3pt]
	  & \multicolumn{2}{l}{= \lift^S_J(\{\delta_{\M}(q_\epsilon,\tau c) \})\hfill \text{by Remark~\ref{rmk:update},}}  \\[+3pt]
	  & \multicolumn{2}{r}{\text{ since } S = I} \\[+3pt]
	  & \multicolumn{2}{l}{= \{ (\lift^S_K(\{\delta_{\M}(q_\epsilon,\tau c) \})_{K \in \up{J}} \}}  \\[+3pt]
 	  & \multicolumn{2}{r}{\text{by definition of $\lift$,}} \\[+3pt]
	  & \multicolumn{2}{r}{\text{since $K \cup S = I = J \cup S$ for all $K \in \up{J}$}} \\[+3pt]
	  & \multicolumn{2}{l}{= \{ (\delta^c_K(h(\tau)[I])  )_{K \in \up{J}} \} \hfill \text{by  Remark~\ref{rmk:update}}}  \\[+3pt]
	  & \multicolumn{2}{l}{= \{ (h(\tau c)[K])_{K \in \up{J}} \}\hfill \text{by induction hypothesis}}  \\[+3pt]
 	  & \multicolumn{2}{r}{\text{since } \sync_K(c) = I} \\[+3pt]
	  & \multicolumn{2}{l}{= \{ (h(\tau' d)[K])_{K \in \up{J}} \mid \tau' d \sim_I \tau c \}}  \\[+3pt]
	  & \multicolumn{2}{r}{\text{since} \sim_I \text{is the identity}} \\[+3pt]
	  & \multicolumn{2}{l}{= \{ (h(\tau' d)[K])_{K \in \up{J}} \mid \tau' d \sim_J \tau c \}}  \\[+3pt]
	  & \multicolumn{2}{r}{\text{by Lemma~\ref{lem:fip-properties}(\ref{lem:fip-properties:P1})}} \\[+3pt]
	  & = h(\tau c)[J] &  \\[+3pt]
	  %& = \{ (h(\tau' d)[K])_{K \in \up{J}} \mid \tau' d \sim_J \tau c \} &  \\[+3pt]
	  %& \multicolumn{2}{r}{\text{by Remark~\ref{rmk:h}}} \\[+3pt]
 	  %& \multicolumn{2}{r}{\text{and Lemma~\ref{lem:fip-properties}(\ref{lem:fip-properties:P1})}} \\[+3pt]
	  %& = h(\tau c)[J] &  \\[+3pt]
               
	%  & =  & \\[+3pt]
	\end{array}
	$$
	%Then by Lemma~\ref{lem:basic-property-h} $\tau' d \sim_J \tau c \Leftrightarrow \tau' d \sim_I \tau c$. Moreover, in our context, $\tau' d \sim_I \tau c \Leftrightarrow \tau' d = \tau c$, and so, we have:
	%$$
	%\arraycolsep=1.5pt
	%\begin{array}{rlr}
	%\Delta(h(\tau),c)[J] 
	%  &= \{ (h(\tau' d)[K])_{K \in \up{J}} \mid \tau' d \sim_J \tau c \}, &  \\[+3pt]
	%  &= h(\tau c)[J]. &  \\[+3pt]
	%\end{array}
	%$$
	
\item if $S \neq I$ and $J = S$, then:

\begin{comment}
\begin{linenomath*}
	$$
	\arraycolsep=1.5pt
	\begin{array}{lr}
	\Delta(h(\tau),c)[J]  & \\[+3pt]
	  \multicolumn{2}{l}{= \delta^c_J(h(\tau)[J])}  \\[+3pt]
	  = \lift^S_J(\{(\delta^d_K(h(\tau')[\sync_K(d)]))_{K \in \up{J}}  \mid \tau' d \sim_J \tau c \}) & \\[+3pt]
	  \multicolumn{2}{r}{\text{by Remark~\ref{rmk:update}}} \\[+3pt]
	  = \{(h(\tau' d)[K])_{K \in \up{J}}  \mid \tau' d \sim_J \tau c \} & \\[+3pt]
	  \multicolumn{2}{r}{\quad\text{since } J = J \cup S \text{ and thus } \lift^S_J \text{ is the identity,}} \\[+3pt]
	  \multicolumn{2}{r}{\text{and by induction hypothesis}} \\[+3pt]
%	  = \{(h(\tau' d)[K])_{K \in \up{J}}  \mid \tau' d \sim_J \tau c \} & \\[+3pt]
%	  \multicolumn{2}{r}{\text{by Lemma~\ref{lem:fip-properties}(\ref{lem:fip-properties:P2}) }} \\[+3pt]
	  = h(\tau c)[J] & \\[+3pt]
	\end{array}
	$$
\end{linenomath*}
\end{comment}

	$$
	\arraycolsep=1.5pt
	\begin{array}{rlr}
	\Delta(h(\tau),c)[J] 
	  & \multicolumn{2}{l}{= \delta^c_J(h(\tau)[J])}  \\[+3pt]
	  & = \lift^S_J(\{(\delta^d_K(h(\tau')[\sync_K(d)]))_{K \in \up{J}}  \mid \tau' d \sim_J \tau c \}) & \\[+3pt]
	  & \multicolumn{2}{r}{\text{by Remark~\ref{rmk:update}}} \\[+3pt]
	  & = \{(h(\tau' d)[K])_{K \in \up{J}}  \mid \tau' d \sim_J \tau c \} & \\[+3pt]
	  & \multicolumn{2}{r}{\text{since } J = J \cup S \text{ and thus } \lift^S_J \text{ is the identity,}} \\[+3pt]
	  & \multicolumn{2}{r}{\text{and by induction hypothesis}} \\[+3pt]
%	  & = \{(h(\tau' d)[K])_{K \in \up{J}}  \mid \tau' d \sim_J \tau c \} & \\[+3pt]
%	  & \multicolumn{2}{r}{\text{by Lemma~\ref{lem:fip-properties}(\ref{lem:fip-properties:P2}) }} \\[+3pt]
	  & = h(\tau c)[J] & \\[+3pt]
	\end{array}
	$$

\item if $S \neq I$ and $J \neq S$, then:

\begin{comment}
%\begin{linenomath*}
	$$
	\arraycolsep=1.5pt
	\begin{array}{lr}
	\Delta(h(\tau),c)[J]   & \\[+3pt]
	  \multicolumn{2}{l}{= \delta^c_J(h(\tau)[S])}  \\[+3pt]
	  = \lift^S_J(\psi) & \\[+3pt]
          \text{where } \psi = \{(\delta^d_K(h(\tau')[\sync_K(d)]))_{K \in \up{S}}  \mid \tau' d \sim_S \tau c \} & \\[+3pt]
	  \multicolumn{2}{r}{\text{by Remark~\ref{rmk:update}}} \\[+3pt]
          \phantom{where \psi\, } = \{(h(\tau' d)[K])_{K \in \up{S}} \mid \tau' d \sim_S \tau c \} & \\
	  \multicolumn{2}{r}{\text{by induction hypothesis}} \\[+3pt]
	  = \{\clift^S_J(\varphi, \psi) \mid \varphi \in \psi \} & \\[+3pt]
	  \multicolumn{2}{r}{\text{by definition of } \lift \text{ as } J \subsetneq S} \\[+3pt]
	  = \{\clift^S_J((h(\tau' d)[K])_{K \in \up{S}}, \psi) \mid \tau' d \sim_S \tau c \} & \\[+3pt]
	  = \{\clift^S_J((h(\tau' d)[K])_{K \in \up{S}}, \psi) \mid \tau' d \sim_J \tau c \} & \\[+3pt]
	  \multicolumn{2}{r}{\text{by Lemma~\ref{lem:fip-properties}(\ref{lem:fip-properties:P1}) }} \\[+3pt]
	\end{array}
	$$
%\end{linenomath*}
\end{comment}
	$$
	\arraycolsep=1.5pt
	\begin{array}{rlr}
	\Delta(h(\tau),c)[J] 
	  & \multicolumn{2}{l}{= \delta^c_J(h(\tau)[S])}  \\[+3pt]
	  & = \lift^S_J(\psi) & \\[+3pt]
          & \text{where } \psi = \{(\delta^d_K(h(\tau')[\sync_K(d)]))_{K \in \up{S}}  \mid \tau' d \sim_S \tau c \} & \\[+3pt]
	  & \multicolumn{2}{r}{\text{by Remark~\ref{rmk:update}}} \\[+3pt]
          & \phantom{where \psi\, } = \{(h(\tau' d)[K])_{K \in \up{S}} \mid \tau' d \sim_S \tau c \} & \\
	  & \multicolumn{2}{r}{\text{by induction hypothesis}} \\[+3pt]
	  & = \{\clift^S_J(\varphi, \psi) \mid \varphi \in \psi \} & \\[+3pt]
	  & \multicolumn{2}{r}{\text{by definition of } \lift \text{ as } J \subsetneq S} \\[+3pt]
	  & = \{\clift^S_J((h(\tau' d)[K])_{K \in \up{S}}, \psi) \mid \tau' d \sim_S \tau c \} & \\[+3pt]
	  & = \{\clift^S_J((h(\tau' d)[K])_{K \in \up{S}}, \psi) \mid \tau' d \sim_J \tau c \} & \\[+3pt]
	  & \multicolumn{2}{r}{\text{by Lemma~\ref{lem:fip-properties}(\ref{lem:fip-properties:P1}) }} \\[+3pt]
	\end{array}
	$$

        We now show that $\clift^S_J((h(\tau' d)[K])_{K \in \up{S}}, \psi)[L] = h(\tau' d)[L]$ for all $L \in \up{J}$, 
 	which concludes the proof as we get $\Delta(h(\tau),c)[J] = \{ (h(\tau' d)[l])_{L \in \up{J}} \mid \tau'd \sim_J \tau c \} = h(\tau c)[J]$.
        We consider several cases:

	\begin{enumerate}
	\item if $L \subseteq S$ (i.e., $L \cup S = S$), then 

\begin{comment}
\begin{linenomath*}
	$$
	%\arraycolsep=1.5pt
	\begin{array}{l}
	 \clift^S_J((h(\tau' d)[K])_{K \in \up{S}}, \psi)[L]  \\[+3pt]
	 \  = \lift^S_L(\psi) \\[+3pt]
%	 \quad  = \lift^S_L(\psi)  \\[+3pt]
	 \  = \lift^S_L(\{(\delta^e_K(h(\tau'')[\sync_K(e)]))_{K \in \up{S}}   \\[+3pt]  
	 \multicolumn{1}{r}{ \mid \tau'' e \sim_S \tau c \})} \\[+3pt]
	 \  = \delta^c_L(h(\tau)[\sync_L(c)]) \\[+3pt]
	 \multicolumn{1}{r}{\text{by Remark~\ref{rmk:update} since } S = \sync_L(c)}  \\[+3pt]
	 \  = h(\tau c)[L] \hfill \text{by induction hypothesis} \\[+3pt]
	 \  = h(\tau' d)[L]  \\[+3pt]
	 \multicolumn{1}{r}{\text{by Lemma~\ref{lem:basic-property-h}}} \\[+3pt]
	 \multicolumn{1}{r}{\text{since } \tau'd \sim_L \tau c \text{ by Lemma~\ref{lem:fip-properties}(\ref{lem:fip-properties:P1})}}  \\[+3pt]
	\end{array}
	$$
\end{linenomath*}
\end{comment}

	$$
	\arraycolsep=1.5pt
	\begin{array}{ll}
	 \clift^S_J((h(\tau' d)[K])_{K \in \up{S}}, \psi)[L] & = \lift^S_L(\psi) \\[+3pt]
%	 \quad  = \lift^S_L(\psi)  \\[+3pt]
	 &  = \lift^S_L(\{(\delta^e_K(h(\tau'')[\sync_K(e)]))_{K \in \up{S}}  \mid \tau'' e \sim_S \tau c \}) \\[+3pt]
	 &  = \delta^c_L(h(\tau)[\sync_L(c)])$ \hfill by Remark~\ref{rmk:update} since $S = \sync_L(c)  \\[+3pt]
	 &  = h(\tau c)[L] $\hfill by induction hypothesis$ \\[+3pt]
	 &  = h(\tau' d)[L] $\hfill by Lemma~\ref{lem:basic-property-h} since $\tau'd \sim_L \tau c $ by Lemma~\ref{lem:fip-properties}(\ref{lem:fip-properties:P1})$  \\[+3pt]
	\end{array}
	$$

        \item \label{item:correctness-proof} otherwise (i.e., $L \cup S \neq S$) let $T = \sync_{L \cup S}(d)$. It will be useful to remark that $S \subseteq L \cup S \subseteq T$ (as coalitions communicate with themselves), thus $S \cup T = T$; and that $S = \sync_J(c) = \sync_J(d) \subseteq \sync_{L}(d)$ (as $J \subseteq L$), thus $\sync_{L}(d) = \sync_{L \cup S}(d) = T$; and by transitivity, we get $S\cup T = \sync_{L}(d)$. We proceed as follows:

        If $T = I$ then:
\begin{comment}
\begin{linenomath*}
	$$
	%\arraycolsep=1.5pt
	\begin{array}{l}
	\clift^S_J((h(\tau' d)[K])_{K \in \up{S}}, \psi)[L] \\[+3pt]
	\  = \lift^S_L(h(\tau' d)[L \cup S]) \\[+3pt]
%	\quad   = \lift^S_L(h(\tau' d)[L \cup S])  \\[+3pt]
	\  = \lift^S_L(\delta^d_{L \cup S}(h(\tau')[T])) \\[+3pt]
        \multicolumn{1}{r}{\text{by induction hypothesis}}  \\[+3pt]	
	\  = \lift^S_L(\lift^T_{L \cup S}(\{\delta_{\M}(q_\epsilon,\tau''e)  \\[+3pt]
	\multicolumn{1}{r}{\mid \tau'' e \sim_I \tau' d\}))}  \\[+3pt]
	\multicolumn{1}{r}{\text{by Remark~\ref{rmk:update}}} \\[+3pt]
	\  = \lift^{S \cup T}_L(\{\delta_{\M}(q_\epsilon,\tau''e) \mid \tau'' e \sim_I \tau' d\})\\[+3pt]
	\multicolumn{1}{r}{\text{by Lemma~\ref{lem:lift-compositional}}} \\[+3pt]
	
	\  = \delta^d_L(h(\tau')[\sync_L(d)])  \\[+3pt]
	\multicolumn{1}{r}{\text{by Remark~\ref{rmk:update} as } S \cup T = \sync_{L}(d)} \\[+3pt]
	%\multicolumn{1}{r}{\text{\emph{(needs a precise justification\footnotemark)}}} \\[+3pt]
	\  = h(\tau' d)[L]$ \hfill by induction hypothesis$ \\[+3pt]
	\end{array}
	$$
\end{linenomath*}
\end{comment}
	$$
	\arraycolsep=1.5pt
	\begin{array}{ll}
	\clift^S_J((h(\tau' d)[K])_{K \in \up{S}}, \psi)[L] & = \lift^S_L(h(\tau' d)[L \cup S]) \\[+3pt]
%	\quad   = \lift^S_L(h(\tau' d)[L \cup S])  \\[+3pt]
	& = \lift^S_L(\delta^d_{L \cup S}(h(\tau')[T]))$ \hfill by induction hypothesis $\\[+3pt]
	& = \lift^S_L(\lift^T_{L \cup S}(\{\delta_{\M}(q_\epsilon,\tau''e) \mid \tau'' e \sim_I \tau' d\}))$ \hfill by Remark~\ref{rmk:update}$ \\[+3pt]
	& = \lift^{S \cup T}_L(\{\delta_{\M}(q_\epsilon,\tau''e) \mid \tau'' e \sim_I \tau' d\})$ \hfill by Lemma~\ref{lem:lift-compositional} $\\[+3pt]
	& = \delta^d_L(h(\tau')[\sync_L(d)])$ \hfill  by Remark~\ref{rmk:update} as $S \cup T = \sync_{L}(d) \\[+3pt]
	%\multicolumn{1}{r}{\text{\emph{(needs a precise justification\footnotemark)}}} \\[+3pt]
	& = h(\tau' d)[L]$ \hfill by induction hypothesis$ \\[+3pt]
	\end{array}
	$$
	
%\addtocounter{footnote}{-1}
%\footnotetext{}
%\stepcounter{footnote}
%\footnotetext{$S \cup T = T = \sync_{L \cup S}(d) = \sync_{L}(d)$.}
%\addtocounter{footnote}{-1}

        \item otherwise (i.e., $L \cup S \neq S$ and $T \neq I$), this case is similar to \ref{item:correctness-proof}:

\begin{comment}
\begin{linenomath*}        
	$$
	%\arraycolsep=1.5pt
	\begin{array}{l}
	\clift^S_J((h(\tau' d)[K])_{K \in \up{S}}, \psi)[L]\\[+3pt]
	 = \lift^S_L(h(\tau' d)[L \cup S]) \\[+3pt]
%	\quad  $ = \lift^S_L(h(\tau' d)[L \cup S])$  \\[+3pt]
	\  = \lift^S_L(\delta^d_{L \cup S}(h(\tau')[T]))  \\[+3pt] 
	\multicolumn{1}{r}{\text{by induction hypothesis}}  \\[+3pt]
	\  = \lift^S_L(\lift^T_{L \cup S}(  \\[+3pt]
	\qquad    \{(\delta^e_K(h(\tau'')[\sync_K(e)]))_{K \in \up{T}}    \\[+3pt] 
        \multicolumn{1}{r}{  \mid \tau'' e \sim_T \tau' d \}))} \\[+3pt]
	\multicolumn{1}{r}{\text{by Remark~\ref{rmk:update}}} \\[+3pt]
	\  = \lift^{S \cup T}_L(    \\[+3pt]
	\{(\delta^e_K(h(\tau'')[\sync_K(e)]))_{K \in \up{T}}  \\[+3pt] 
	\multicolumn{1}{r}{\mid \tau'' e \sim_T \tau' d \})}  \\[+3pt]
	\multicolumn{1}{r}{\text{by Lemma~\ref{lem:lift-compositional}}} \\[+3pt]
	\  = \delta^d_L(h(\tau')[\sync_L(d)])  \\[+3pt]
	\multicolumn{1}{r}{\text{ by Remark~\ref{rmk:update} as } S \cup T = \sync_{L}(d)} \\[+3pt]
	%\multicolumn{1}{r}{\text{\emph{(needs a precise justification\footnotemark)}}} \\[+3pt]
	\    = h(\tau' d)[L]$ \hfill by induction hypothesis $\\[+3pt]
	\end{array}
	$$
\end{linenomath*}	
\end{comment}

	$$
	\arraycolsep=1.5pt
	\begin{array}{ll}
	\clift^S_J((h(\tau' d)[K])_{K \in \up{S}}, \psi)[L] & = \lift^S_L(h(\tau' d)[L \cup S]) \\[+3pt]
%	\quad = \lift^S_L(h(\tau' d)[L \cup S])  \\[+3pt]
	& = \lift^S_L(\delta^d_{L \cup S}(h(\tau')[T])) \hfill \text{by induction hypothesis}  \\[+3pt]
	& = \lift^S_L(\lift^T_{L \cup S}(\{( \delta^e_K(h(\tau'')[\sync_K(e)]))_{K \in \up{T}} \mid \tau'' e \sim_T \tau' d \})) \\[+3pt]
	& \multicolumn{1}{r}{\text{by Remark~\ref{rmk:update}}} \\[+3pt]
	& = \lift^{S \cup T}_L(\{(\delta^e_K(h(\tau'')[\sync_K(e)]))_{K \in \up{T}} \mid \tau'' e \sim_T \tau' d \})  \\[+3pt]
	& \multicolumn{1}{r}{\text{by Lemma~\ref{lem:lift-compositional}}} \\[+3pt]
	& = \delta^d_L(h(\tau')[\sync_L(d)]) \hfill \text{ by Remark~\ref{rmk:update} as } S \cup T = \sync_{L}(d) \\[+3pt]
	%\multicolumn{1}{r}{\text{\emph{(needs a precise justification\footnotemark)}}} \\[+3pt]
	& = h(\tau' d)[L] \hfill \text{by induction hypothesis} \\[+3pt]
	\end{array}
	$$

	\end{enumerate}

\end{enumerate}

\end{proof}
\end{longversion}

By Lemma~\ref{lem:refinement-rectangular-h} and Corollary~\ref{cor:morphic-h},
the function $h$ is a rectangular morphism for FIP games. 
By Theorem~\ref{thm:bisimilar} (using Lemma~\ref{lem:info-quotient} and 
Lemma~\ref{lem:bisimulation-preserves-strat}) we can reduce FIP games
to a game of perfect information, and thus solving FIP games is decidable.

\begin{theorem}\label{theo:FIP-synthesis}
The synthesis problem for FIP games with a parity winning condition is decidable.
\end{theorem}

Theorem~\ref{theo:FIP-synthesis} extends to all (visible) winning conditions for which
perfect-information games are decidable, such as mean-payoff, discounted sum, etc.

Given a FIP game with $k$ observers, the pre-processing step (Section~\ref{sec:pre-processing})
adds one observer, and thus the size $\abs{P}$ of the perfect-information game (induced by
the rectangular morphism $h$) is $(k+1)$-fold exponential in the size of the FIP game.
Note that the parity objective is defined using the same number of priorities, 
and since perfect-information parity games can be solved in time
at most exponential in the number of priorities (even in quasi-polynomial time,
see e.g.~\cite[Section 1.2]{JMT22}) we derive a $(k+1)$-EXPTIME upper bound
for the synthesis problem.

We show a matching lower bound for reachability winning conditions.
We reduce the membership problem for alternating $k$-EXPSPACE Turing machines
(which is $(k+1)$-EXPTIME-complete)
to the synthesis problem for FIP with $k$ observers.

%\begin{longversion}

Given an alternating $k$-EXPSPACE Turing machine $M$ and an input word $w$ of length $n$,
we construct a FIP game with reachability objective in which the player
has a winning strategy if and only if~$M$ accepts~$w$. The contructed
game has size polynomial in~$n$ and~$\abs{M}$. 
 
Intuitively, the game %(in particular the winning condition via the coloring function) 
simulates an interactive execution of the alternating Turing machine
where the player chooses the transitions in existential states, and the environment
chooses the transitions in universal states. Moreover, 
in order to win, the player has to announce the configuration\footnote{From now on, configurations
are of the Turing machine, no longer from the rectangular morphism.}
of the Turing machine~$M$ after each transition, that is the content of the $k$-fold exponential-size
tape, the position of the tape head, and the control state (initially the tape contains
$w$ followed by blank symbols), and eventually to announce a configuration 
containing the accepting state (thus a reachability condition). 

The crux is to ensure that the configurations announced by the player are 
consistent with the transitions of the machine~$M$. This is simple
for the initial configuration, and we need to %use imperfect information to
verify that all other configurations are the successor of the previously announced
configuration upon executing the corresponding transition.

The difficulty is that the configurations, which are of $k$-fold exponential size,
cannot be stored by the winning condition (via the coloring function encoded by a Mealy machine
with polynomial number of states). 
Therefore we use imperfect information and observers to carry out the verification.
In fact it is sufficient to be able to check equality of two configurations,
because transition updates are local (within a window of three tape cells) and 
can be stored in the Mealy machine. 

The verification of configuration equality works according to to the following principle.
At any stage of the play,
the environment may challenge the equality of the last two configurations
announced by the player, by marking a position in the first configuration
and a position in the second configuration where the content of the two
configurations differ. To produce a mark, the environment sends a specific
observation signal to some designated observer, while the player 
announces the configurations. As the signal sent to the observer
is not visible to the player until a communication happens, the marking
can be done retrospectively by the environment, knowing the second configuration. 
After the marking is done, the player is allowed to communicate with the observer 
and to see the marks. 
If the bits at the marked positions are equal (which can be checked by 
the winning condition), then the player wins. Otherwise, the player
may challenge that the positions marked in the two configurations
are the same, which amounts to verifying \emph{inequality} of two numbers 
with $(k-1)$-fold exponentially many bits. We achieve this by requiring that 
the player announces, along with every bit of a configuration, 
its address (i.e., its position encoded in binary), and to claim
inequality of two marked positions $p_1$ and $p_2$, that is inequality
of their announced addresses, the player 
identifies a position $p^*$ in the two addresses where the bits are different.
As the player cannot do marking retrospectively, he gives the position $p^*$
by announcing the binary encoding of $p^*$ (thus over $(k-2)$-fold exponentially many bits),
along with the value $b_i$ of the bit at position $p^*$ in (the binary encoding of) $p_i$ ($i=1,2$),
thus either announcing $b_1=1$ and $b_2=0$, or $b_1=0$ and $b_2=1$. 
This allows the environment to challenge any of the two claims ``the bit at position 
$p^*$ in (the binary encoding of) $p_1$ is $b_1$'' or ``the bit at position 
$p^*$ in (the binary encoding of) $p_2$ is $b_2$''. 
The environment does so by marking a position in $p_i$ ($i=1$ or $i=2$), the marking 
being recorded by a fresh observer: 
either the bit marked is $b_i$ (which can be checked by the winning condition)
and the player wins, or the player may challenge that the positions marked 
in $p_i$ is $p^*$ after communicating with the new observer, 
which is again an inequality test, now over numbers with $(k-3)$-fold exponentially many bits.
The verification game proceeds in this way until the numbers to be compared
are encoded with a small number of bits, which can then be checked by a small
Mealy machine.  

In summary, every bit announced by the player is followed by its address (encoded in binary)
within the sequence of bits announced. As the address is itself a sequence of bits, 
the rule applies recursively. Given a size $n$ and depth $k$, define $\mathit{Tower}(n,0) = n$
and $\mathit{Tower}(n,k) = 2^{\mathit{Tower}(n,k-1)}$ if $k \geq 1$, thus $\mathit{Tower}(n,k)$ is the
$k$-fold exponential of $n$ (in base $2$). We define counters of depth $k$
as an encoding of numbers, ranging from $0$ to $\mathit{Tower}(n,k)-1$, as sequences
of bits where each bit is followed by its address, encoded as a counter of depth $k-1$.
This encoding is directly inspired by a similar definition in previous work~\cite[Section 4.2]{GGMW13}.

A counter of depth $0$ is a sequence $b \#_0$ where $b \in \{0,\dots,n-1\}$
and $\#_0$ is a separator (of level $0$). Define $\val_0(b \#_0) = b$. 
A counter of depth $k$  ($k \geq 1$) is a sequence 
$\gamma = b_0 c_0 \, b_1 c_1 \dots b_N c_N \#_k$ where $N = \mathit{Tower}(n,k-1)-1$,
consisting of bits $b_i \in \{0,1\}$, counters $c_i$ of depth $k-1$ and value $\val_{k-1}(c_i) = i$,
for all $0 \leq i \leq N$, and a separator $\#_k$ of level $k$. 
Define $\val_{k}(\gamma) = \sum_{i} b_i 2^{N-i}$. We refer to the bits $b_i$ 
in $\gamma$ as bits of level $k$.

A formal description the contructed FIP games would be tedious to read.
\begin{shortversion}
We give an informal description in the extended version.
\end{shortversion}

\begin{longversion}
We give an informal description.
We describe the action (of the player) and the moves (of the environment)
as if they were independent (the player may play a few actions, then the environment
may play some moves, rather than a strict interleaving of one action and one move).

The actions of the player are the bits $\{0,1\}$ along with their level, 
the separators $\{\#_0,\dots,\#k\}$,
the transitions of the Turing machine, and two actions $r^{01}$ and $r^{10}$
to announce that the value of certain two bits are different (either $0$ then $1$,
or $1$ then $0$); these actions are followed by the address where the bits 
can be found (it is the same address within two different counters). 

The moves of the environment are the marks that are given to the observers,
the synchronisation moves $\{\$_i \mid 1 \leq i \leq k\}$ that let the player
communicate with the corresponding observer, and an extra two moves that are not
shown to anyone and are used to drive the Mealy automaton for the winning
condition. The Mealy automaton consists of several components that check some
property of the history, to which the environment can branch using the extra two 
moves (possibly several times). As every component can be chosen, but the player
does not see which one is chosen, the player has to ensure that all are accepting.
We now show that the correct simulation of the Turing machine can be verified 
by a small number of small such components.

In the game, the player can only see his own actions, until a synchronisation occurs.
We call \emph{phase} the segment of a history between two (consecutive) synchronisations.
We construct game such that the player wins if each phase corresponds to a correct
encoding of a transition and a configuration of the Turing machine (encoded as a counter of depth $k$),
that is the format of the encoding is correct and the bit values correspond
to the execution of the transition from the configuration in the previous phase.

First we define some conditions that must hold for all $0 \leq i \leq k$ and 
that are easy to verify (individually) with a small automaton:
\begin{itemize}
\item the configuration starts with a $0$ of level $k$ and ends with a separator $\#_k$;
\item a bit of level $i \geq 1$ is followed by a bit of level $i-1$;
\item a bit of level $0$ is followed by $\#_0$; 
\item a separator of level $i < k$ is followed either by a separator of level $i+1$ or by a bit of level $i+1$;
\item all bits of level $i-1$ are $0$ until the first separator of level $i-1$;
\item all bits of level $i-1$ are $1$ after the second to last separator of level $i-1$.
\end{itemize}
The environment may branch to any of the components checking the above conditions.
Similarly, we ensure that the transitions of the Turing machines respect the
control states: we store the control state and update it according to the
the transitions, then blame either the player or the environment for a wrong
choice of transition. 

The content of the configuration is checked via the marking by the environment,
which is transmitted to the observers. First, we can detect whether the marking has the wrong format
(a separator is marked istead of a bit, more than two marks occur at some level, etc.)
and let the player win. If the format is correct, the environment has to 
let the player see the marking at the appropriate level~$i$ (by communicating 
with the corresponding observer, using $_i$), so the player may claim
an \emph{inequality}, namely  
that the two marked bits are not at the same position, by pointing, within the (two) addresses
of the marked bits (which are supposed to be equal if the marking is honest), 
at which position their bits differ, claiming either $r^{01}$ or $r^{10}$ 
and announcing their address as a counter of the appropriate (and lower) depth. 
Now, either the announced address is not a counter in the correct format,
and the environment can restart the game at the appropriate depth,
or we check the player's claim, letting the environment choose which of
the two bits to check, and if it does not match with the marked bit, 
allowing the player to further claim \emph{inequality} between the address
of the pointed bit and the announced address, at a lower depth. 
The whole process requires at most $k$ observers because the depth of the 
checked counters decreases until level $0$ where an automaton can directly check the claim.
We argue that a small (polynomial-size) automaton can store 
the value of the marked bits in order to verify the claims ($r^{01}$ or $r^{10}$) of the player.
Storing $2$ bits at each depth (thus $2k$ bits) would be sufficient if the format of the 
addresses announced by the player are never challenged by the environment.
Otherwise, the game restarts to check the format of a counter that is
now of lower depth, and requires to store $2(k-1)$ bits (if there is no further challenge 
on the format). Repeating this argument gives a number $O(k^2)$ of bits to store.
As $k$ is a constant (independent of $n$ and the size of the Turing machine),
the reduction is indeed polynomial.

We show that the player has a strategy to reach the 
accepting state of the Turing machine~$M$ if and only if $M$ accepts~$w$. 

If $M$ accepts~$w$, the strategy of the player is to produce a run of $M$ on $w$
where the configurations are in the format of counters of depth $k$,
thus satisfying the simple conditions.
If the environment ever challenges the content of the tape, he may either
mark bits at the same address within two successive configurations (thus that are correct with respect
to the the transition of the Turing machine), or choose different 
positions within the two configurations. In the first case, the pointed bits
will be compared by the automaton for the winning condition and the player would
win; in the second case, the two bits have different addresses and 
the strategy of the player is to further produce a counter that points
to a position of a bit that differs in the two addresses, which can be 
checked by an inductive argument, showing that the strategy of the player is winning.

For the converse direction, If $M$ does not accept~$w$, then a strategy of the 
player either does not reach an accepting state, or produces a spurious run of $M$ on $w$.
If the format of the counters is not respected, then the environment would
choose the corresponding small automaton and the player loses. If the content
of the tape cells is incorrect, the environment marks the position where a
fault occurs, and will be able to respond to any challenge in address inequality
from the player, by always marking bits at the same address in two counters.
Hence the strategy of the player is losing.
\end{longversion}

% to be improved later.

\begin{theorem}\label{theo:FIP-synthesis-complexity}
The synthesis problem for FIP games with $k$ observers is $(k+1)$-EXPTIME-complete,
both for parity and reachability winning conditions.
\end{theorem}

\section{Expressiveness}\label{sec:expressiveness}

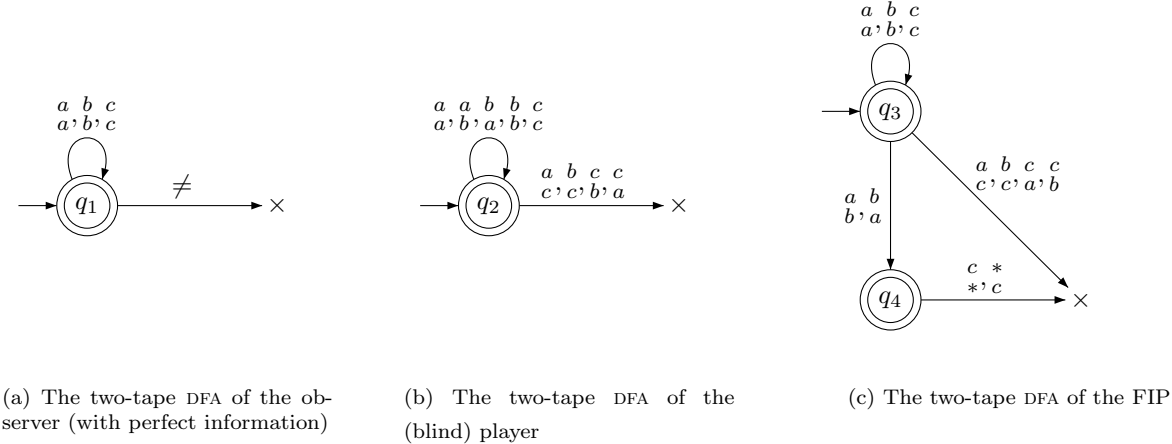
\begin{figure*}[th]%
\begin{center}
\hrule
%\hspace{10mm}
\subfloat[The two-tape \dfa of the observer (with perfect information)]{
   %\renewcommand{\sb}[1]{\scalebox{0.75}[1]{#1}}

%{\scriptsize 
\begin{picture}(40,50)(0,0)
%\put(0,0){\framebox(40,50){}}

%\gasset{Nw=9,Nh=9,Nmr=4.5,rdist=1, loopdiam=6}
%\gasset{Nw=10,Nh=6,Nmr=3,rdist=1, loopdiam=5}
\gasset{Nw=8,Nh=8,Nmr=4, rdist=1, loopdiam=5}      % , ELdist=0
%\gasset{Nw=5,Nh=5,Nmr=2.5,rdist=1, loopdiam=6, linewidth=0.12}

\node[Nmarks=ir, Nframe=y](q1)(10,22.5){$q_1$}
%\node[Nmarks=r, Nframe=y](q2)(35,45){$q_2$}
%\node[Nmarks=r, Nframe=y](q3)(10,20){$q_2$}
\node[Nmarks=n, Nframe=n, Nw=4, Nh=5](q4)(35,22.5){$\times$}
%\node[Nmarks=n, Nframe=n](label)(35,5){$\times = \reject$}
%\put(34,10){\makebox(0,0)[l]{$\times = \reject$}}

%\put(45,40){\makebox(0,0)[l]{$q_1\text{: equal histories, no } 0 \text{ occured}$}}
%\put(45,35){\makebox(0,0)[l]{$q_2\text{: equal histories, a } 0 \text{  did occur}$}}
%\put(45,30){\makebox(0,0)[l]{$q_3\text{: different histories, no } 0 \text{  occured}$}}
%\put(45,25){\makebox(0,0)[l]{$\times\text{: different histories, a } 0 \text{  occured}$}}

%\put(45,15){\makebox(0,0)[l]{$\begin{array}{c}\{a,b\} \\[-3pt] c \\[-3pt] \{a,b\}\end{array}$}}

%\drawedge[ELpos=50, ELside=l, curvedepth=0](q1,q2){$\fbi{0}{0}$}

%\drawedge[ELpos=50, ELside=r, ELdist=1, curvedepth=0](q1,q3){$\fbi{1}{2},\fbi{2}{1}$}

%\drawedge[ELpos=45, ELside=l, ELdist=-2, curvedepth=0](q1,q4){$\cfbi{\{1,2\}}{\ \ 0}\!,\!\cfbi{\ \ 0}{\{1,2\}}$}
\drawedge[ELpos=50, ELside=l, ELdist=1, curvedepth=0](q1,q4){$\neq$}
%\drawedge[ELpos=50, ELside=l, ELdist=1, curvedepth=0](q2,q4){$\neq$}
%\drawedge[ELpos=50, ELside=l, curvedepth=0](q3,q4){$\cfbi{0}{*},\cfbi{*}{0}$}

%\drawedge[ELpos=50, ELside=r, curvedepth=-4](v4,v0){$1$}

\drawloop[ELside=l,loopCW=y, loopdiam=5, loopangle=90](q1){$\cfbi{a}{a},\fbi{b}{b},\fbi{c}{c}$}
%\drawloop[ELside=l,loopCW=y, ELdist=1, loopdiam=5, loopangle=90](q2){$=$}   %$\cfbi{a}{a},\cfbi{b}{b},\cfbi{c}{c}$}
%\drawloop[ELside=l,loopCW=y, loopdiam=5, loopangle=270](q3){$\ufbi{\{1,2\}}{\{1,2\}}$}
%\drawloop[ELside=l,loopCW=y, loopdiam=5, loopangle=270](q3){$\cfbi{1}{1},\fbi{1}{2},\fbi{2}{1},\fbi{2}{2}$}

%\node[Nmarks=n](n1)(30,60){$\ver_1$}
%\node[Nmarks=n, Nmr=0](n11)(15,45){$\ver_2$}
%\nodelabel[ExtNL=y, NLangle=270, NLdist=1](n111){{\small val=1}}

%\drawedge[ELpos=50, ELside=r, curvedepth=0, AHLength = 2, AHangle=25, AHlength = 1.81](n1,n11){$-1$}
%\drawedge[ELpos=50, ELside=r, curvedepth=0, AHnb=0, linewidth=.6](n1,n11dummy){}

%\drawedge[ELpos=50, ELside=l, curvedepth=-12](t3,n1){}
%\drawbpedge(t3,70,22,n1,110,22){}
%\drawline[AHnb=1,arcradius=1](113,17.5)(113,29)(5,29)(5,17.5)

%\drawloop[ELside=l,loopCW=y, loopdiam=6](n4){$1$}

%\drawloop[ELside=l,loopCW=y](nk){$0,1$}

%\drawedge[dash={1}0](n3bis,nkbis){$0,1$}

\end{picture}
%}
   \label{fig:perfect}  
}
\hfill
\subfloat[The two-tape \dfa of the (blind) player {\large \strut}]{
   %\renewcommand{\sb}[1]{\scalebox{0.75}[1]{#1}}

%{\scriptsize 
\begin{picture}(40,50)(0,0)
%\put(0,0){\framebox(40,50){}}

%\gasset{Nw=9,Nh=9,Nmr=4.5,rdist=1, loopdiam=6}
%\gasset{Nw=10,Nh=6,Nmr=3,rdist=1, loopdiam=5}
\gasset{Nw=8,Nh=8,Nmr=4, rdist=1, loopdiam=5}      % , ELdist=0
%\gasset{Nw=5,Nh=5,Nmr=2.5,rdist=1, loopdiam=6, linewidth=0.12}

\node[Nmarks=ir, Nframe=y](q1)(10,22.5){$q_2$}
%\node[Nmarks=r, Nframe=y](q2)(35,45){$q_2$}
%\node[Nmarks=r, Nframe=y](q3)(10,20){$q_2$}
\node[Nmarks=n, Nframe=n, Nw=4, Nh=5](q4)(35,22.5){$\times$}
%\node[Nmarks=n, Nframe=n](label)(35,5){$\times = \reject$}
%\put(34,10){\makebox(0,0)[l]{$\times = \reject$}}

%\put(45,40){\makebox(0,0)[l]{$q_1\text{: equal histories, no } 0 \text{ occured}$}}
%\put(45,35){\makebox(0,0)[l]{$q_2\text{: equal histories, a } 0 \text{  did occur}$}}
%\put(45,30){\makebox(0,0)[l]{$q_3\text{: different histories, no } 0 \text{  occured}$}}
%\put(45,25){\makebox(0,0)[l]{$\times\text{: different histories, a } 0 \text{  occured}$}}

%\put(45,15){\makebox(0,0)[l]{$\begin{array}{c}\{a,b\} \\[-3pt] c \\[-3pt] \{a,b\}\end{array}$}}

%\drawedge[ELpos=50, ELside=l, curvedepth=0](q1,q2){$\fbi{0}{0}$}

%\drawedge[ELpos=50, ELside=r, ELdist=1, curvedepth=0](q1,q3){$\fbi{1}{2},\fbi{2}{1}$}

%\drawedge[ELpos=45, ELside=l, ELdist=-2, curvedepth=0](q1,q4){$\cfbi{\{1,2\}}{\ \ 0}\!,\!\cfbi{\ \ 0}{\{1,2\}}$}
\drawedge[ELpos=50, ELside=l, ELdist=1, curvedepth=0](q1,q4){$\cfbi{a}{c},\fbi{b}{c},\cfbi{c}{b},\fbi{c}{a}$}
%\drawedge[ELpos=50, ELside=l, ELdist=1, curvedepth=0](q2,q4){$\neq$}
%\drawedge[ELpos=50, ELside=l, curvedepth=0](q3,q4){$\cfbi{0}{*},\cfbi{*}{0}$}

%\drawedge[ELpos=50, ELside=r, curvedepth=-4](v4,v0){$1$}

\drawloop[ELside=l,loopCW=y, loopdiam=5, loopangle=90](q1){$\cfbi{a}{a},\fbi{a}{b},\fbi{b}{a},\fbi{b}{b},\cfbi{c}{c}$}
%\drawloop[ELside=l,loopCW=y, ELdist=1, loopdiam=5, loopangle=90](q2){$=$}   %$\cfbi{a}{a},\cfbi{b}{b},\cfbi{c}{c}$}
%\drawloop[ELside=l,loopCW=y, loopdiam=5, loopangle=270](q3){$\ufbi{\{1,2\}}{\{1,2\}}$}
%\drawloop[ELside=l,loopCW=y, loopdiam=5, loopangle=270](q3){$\cfbi{1}{1},\fbi{1}{2},\fbi{2}{1},\fbi{2}{2}$}

%\node[Nmarks=n](n1)(30,60){$\ver_1$}
%\node[Nmarks=n, Nmr=0](n11)(15,45){$\ver_2$}
%\nodelabel[ExtNL=y, NLangle=270, NLdist=1](n111){{\small val=1}}

%\drawedge[ELpos=50, ELside=r, curvedepth=0, AHLength = 2, AHangle=25, AHlength = 1.81](n1,n11){$-1$}
%\drawedge[ELpos=50, ELside=r, curvedepth=0, AHnb=0, linewidth=.6](n1,n11dummy){}

%\drawedge[ELpos=50, ELside=l, curvedepth=-12](t3,n1){}
%\drawbpedge(t3,70,22,n1,110,22){}
%\drawline[AHnb=1,arcradius=1](113,17.5)(113,29)(5,29)(5,17.5)

%\drawloop[ELside=l,loopCW=y, loopdiam=6](n4){$1$}

%\drawloop[ELside=l,loopCW=y](nk){$0,1$}

%\drawedge[dash={1}0](n3bis,nkbis){$0,1$}

\end{picture}
%}
    \label{fig:blind}
}
\hfill
\subfloat[The two-tape \dfa of the FIP]{
   %\renewcommand{\sb}[1]{\scalebox{0.75}[1]{#1}}

%{\scriptsize 
\begin{picture}(50,50)(0,0)
%\put(0,0){\framebox(50,50){}}

%\gasset{Nw=9,Nh=9,Nmr=4.5,rdist=1, loopdiam=6}
%\gasset{Nw=10,Nh=6,Nmr=3,rdist=1, loopdiam=5}
\gasset{Nw=8,Nh=8,Nmr=4, rdist=1, loopdiam=5}      % , ELdist=0
%\gasset{Nw=5,Nh=5,Nmr=2.5,rdist=1, loopdiam=6, linewidth=0.12}

\node[Nmarks=ir, Nframe=y](q1)(10,35){$q_3$}
%\node[Nmarks=r, Nframe=y](q2)(35,45){$q_2$}
\node[Nmarks=r, Nframe=y](q3)(10,10){$q_4$}
\node[Nmarks=n, Nframe=n, Nw=4, Nh=5](q4)(35,10){$\times$}
%\node[Nmarks=n, Nframe=n](label)(35,5){$\times = \reject$}
%\put(34,10){\makebox(0,0)[l]{$\times = \reject$}}

%\put(45,40){\makebox(0,0)[l]{$q_1\text{: equal histories, no } 0 \text{ occured}$}}
%\put(45,35){\makebox(0,0)[l]{$q_2\text{: equal histories, a } 0 \text{  did occur}$}}
%\put(45,30){\makebox(0,0)[l]{$q_3\text{: different histories, no } 0 \text{  occured}$}}
%\put(45,25){\makebox(0,0)[l]{$\times\text{: different histories, a } 0 \text{  occured}$}}

%\put(45,15){\makebox(0,0)[l]{$\begin{array}{c}\{a,b\} \\[-3pt] c \\[-3pt] \{a,b\}\end{array}$}}

%\drawedge[ELpos=50, ELside=l, curvedepth=0](q1,q2){$\fbi{0}{0}$}

\drawedge[ELpos=50, ELside=r, ELdist=1, curvedepth=0](q1,q3){$\fbi{a}{b},\fbi{b}{a}$}

%\drawedge[ELpos=45, ELside=l, ELdist=-2, curvedepth=0](q1,q4){$\cfbi{\{1,2\}}{\ \ 0}\!,\!\cfbi{\ \ 0}{\{1,2\}}$}
\drawedge[ELpos=50, ELside=l, ELdist=.5, curvedepth=0](q1,q4){$\cfbi{a}{c},\fbi{b}{c},\fbi{c}{a},\fbi{c}{b}$}
%\drawedge[ELpos=50, ELside=l, ELdist=1, curvedepth=0](q2,q4){$\neq$}
\drawedge[ELpos=50, ELside=l, curvedepth=0](q3,q4){$\cfbi{c}{*},\cfbi{*}{c}$}

%\drawedge[ELpos=50, ELside=r, curvedepth=-4](v4,v0){$1$}

\drawloop[ELside=l,loopCW=y, loopdiam=5, loopangle=90](q1){$\cfbi{a}{a},\fbi{b}{b},\fbi{c}{c}$}
%\drawloop[ELside=l,loopCW=y, ELdist=1, loopdiam=5, loopangle=90](q2){$=$}   %$\cfbi{a}{a},\cfbi{b}{b},\cfbi{c}{c}$}
%\drawloop[ELside=l,loopCW=y, loopdiam=5, loopangle=270](q3){$\ufbi{\{1,2\}}{\{1,2\}}$}
%\drawloop[ELside=l,loopCW=y, loopdiam=5, loopangle=270](q3){$\cfbi{1}{1},\fbi{1}{2},\fbi{2}{1},\fbi{2}{2}$}

%\node[Nmarks=n](n1)(30,60){$\ver_1$}
%\node[Nmarks=n, Nmr=0](n11)(15,45){$\ver_2$}
%\nodelabel[ExtNL=y, NLangle=270, NLdist=1](n111){{\small val=1}}

%\drawedge[ELpos=50, ELside=r, curvedepth=0, AHLength = 2, AHangle=25, AHlength = 1.81](n1,n11){$-1$}
%\drawedge[ELpos=50, ELside=r, curvedepth=0, AHnb=0, linewidth=.6](n1,n11dummy){}

%\drawedge[ELpos=50, ELside=l, curvedepth=-12](t3,n1){}
%\drawbpedge(t3,70,22,n1,110,22){}
%\drawline[AHnb=1,arcradius=1](113,17.5)(113,29)(5,29)(5,17.5)

%\drawloop[ELside=l,loopCW=y, loopdiam=6](n4){$1$}

%\drawloop[ELside=l,loopCW=y](nk){$0,1$}

%\drawedge[dash={1}0](n3bis,nkbis){$0,1$}

\end{picture}
%}
    \label{fig:fip-two-players}
}
%\hspace{10mm}
\hrule
%\smallskip
\caption{A FIP with one player and one observer. \label{fig:simple-fip}}%
\end{center}
\end{figure*}

We compare the expressive power of full-information protocols to
define indistinguishability relations. We show that FIP are 
strictly more expressive than the traditional partial-observation setting,
which corresponds to FIP with no observer.
We further generalise this result and show that the number of observers in a FIP 
induces a strict hierarchy  in terms of expressive power.
Finally, the general framework of two-tape automata~\cite{BD23} is 
strictly more expressive than FIP (with an arbitrary number of observers). 

Two-tape automata are \dfa over alphabet $\Gamma \times \Gamma$
that recognise synchronous relations over $\Gamma$, 
that is, relations between words of the same length. 
The relation recognised by such an automaton $\A$ consists of all pairs 
of words $c_1 c_2 \ldots c_\ell, c'_1 c'_2 \ldots c'_\ell \in \Gamma^*$ 
such that $(c_1, c_1') (c_2, c_2') \ldots (c_\ell, c_\ell') \in L(\A)$. 
With a slight abuse of notation, we also denote this relation by $L(\A)$. 
We say that a synchronous relation is regular if it is recognised by a \dfa.
It is decidable in polynomial time whether the relation recognised by a given 
two-tape automaton is an indistinguishability relation~\cite[Lemma 2.3]{BD23}.
The decidability of the synthesis problem is open when the indistinguishability 
relation is regular.

As a first example (inspired by ~\cite[Lemma 3.2]{BD23}), consider the following scenario over a set of moves $\Gamma = \{a,b,c\}$:
there is an observer with perfect information (their indistinguishability relation
is recognised by the two-tape \dfa of \figurename~\ref{fig:perfect}), and a player who does
not distinguish $a$ and $b$, but can observe $c$ (their indistinguishability relation 
is recognised by the two-tape \dfa of \figurename~\ref{fig:blind}). In a FIP where the player 
communicates with the observer on $c$, the indistinguishability relation is 
recognised by the two-tape \dfa of \figurename~\ref{fig:fip-two-players}.
Informally, two histories are indistinguishable for the FIP if 
they are equal up to the last $c$. The induced information tree has unbounded branching
as all histories of the same length $n$ that do not contain $c$ are indistinguishable, 
hence $u_n = \{a, b\}^n$ is an information set, and for every history $\tau \in u_n$ 
the history $\tau c$ forms a singleton information set. Therefore $u_n$ has at least 
$2^n$ successors, for every $n$.

A regular observation function (or equivalently, a FIP with no observer) induces 
an information tree with bounded branching~\cite[Theorem~4.1]{BD23}, which
implies that FIP are strictly more expressive than the traditional partial-observation games.

We show that two-tape automata are strictly more expressive than FIP.
First we show that two-tape automata can recognise the indistinguishibility 
relation of a FIP.

\begin{lemma}\label{lem:FIP-to-dfa}
Every indistinguishability relation defined by a FIP can be recognised by a two-tape \dfa.
\end{lemma}

Given a FIP $F = \tuple{I, (\M_i)_{i\in I},(R_{\sigma})_{\sigma \in \Sigma}}$
with $n$ observers, where the Mealy machines $\M_i$ define observation functions $\beta_i$
for each player $i \in I$, and the relations $R_{\sigma}$ define the communication
links on observation $\sigma$, we construct a two-tape \dfa $\A_{{\textsf{FIP}}}$ that defines the 
indistinguishability relation $\sim$ of the FIP $F$ as follows.
For each player $i \in I$, consider the two-tape \dfa $\A_i$ that accepts 
a pair $(\tau,\tau')$ of histories if $\hat{\beta_i}(\tau) = \hat{\beta_i}(\tau')$. 
The automaton $\A_i$ also stores the last observation produced by $\beta_i$ (if 
it is the same in the two input histories).

Now consider the synchronised product of the automata $\A_i$ for $i \in I$
and its transition relation $\delta$. Construct the transition relation $\delta'$
over the same state space as $\delta$, where given $p = \delta(q,c)$,
we define $r = \delta'(q,c)$ as follows.  
Consider the least set $J \subseteq I$ containing all $i \in I$
such that either the entry $p[i]$ is the rejecting state, or
$(i,j) \in R_{\sigma}$ and $j \in J$ where $\sigma$ is the 
observation of Player~$i$ stored in $p$.
The state $r$ is defined by $r[i] = q_{\rej}$ if $i \in J$,
and $r[i] = p[i]$ otherwise, where $q_{\rej}$ is an absorbing rejecting state.
A state $q$ is accepting in $\A_{{\textsf{FIP}}}$ if the entry $q[0]$ corresponding
to Player~$0$ is accepting. Intuitively, two histories are indistinguishable
if Player~$0$ receives the same observations for both of them, and the two
histories are indistinguishable for all players with whom Player~$0$ 
communicates (possibly indirectly).
%\mynote{check definition of $\rej$.}

\begin{comment}
For each Mealy machine $\M_i$, construct the two-tape \dfa $\P_i$ recognising 
the relation $\{(\tau, \tau') \mid \hat{\beta_i}(\tau) = \hat{\beta_i}(\tau')\}$
induced by equality of observation sequences for Player $i \in I$.

Consider the synchronised product $\M\P$ of, for all $i \in I$, the two-tape automaton $\P_i$
and two copies of $\M_i$ as two-tape automata, one reading the first tape,
and the other one reading the second tape. In $\M\P$, we can determine 
which players distinguish the two histories based on their own observation,
and for the other players what was the last observation (which must the same
on both histories). The last step is to modify the transition relation of $\M\P$
to take into account the communication. Given the last observation for every player,
we can construct the (transitive closure of the) communication links between the players.
If such a link connects some Player $i$ to some player $j$ and the state of 
player $j$ (i.e., the state of $\P_j$ in the product) is rejecting, 
then we replace the state of Player $i$ by the rejecting sink of $\P_i$.
The accepting states are those where the component of $\P_0$ (the automaton
for Player~$0$) is accepting (i.e., all states except the rejecting sink).
\end{comment}

The construction is illustrated in \figurename~\ref{fig:simple-fip} where
the state $q_3$ corresponds to $(q_1,q_2)$ (from the automata $\P_0$ of 
\figurename~\ref{fig:perfect} and $\P_1$ of \figurename~\ref{fig:blind}),
and $q_4$ corresponds to $(q_1,q_{\rej})$ where Player~$1$ % the automaton $\P_1$ 
has distinguished the histories, but Player~$0$ %the automaton $\P_0$ 
did not as there was no communication with $\P_1$ yet.
In $q_4$, whenever a communication occurs (via $c$), the histories get
distinguished by Player~$0$.
 
\begin{figure}[t]
\begin{center}
\hrule
%\renewcommand{\sb}[1]{\scalebox{0.75}[1]{#1}}

%{\scriptsize 
\begin{picture}(55,72)(0,0)
%\put(0,0){\framebox(50,65){}}

      \gasset{Nw=8,Nh=8,Nmr=4, rdist=1, loopdiam=5}      % , ELdist=0

      \node[Nmarks=ir, Nframe=y](q1)(10,55){$q_1$}
      \node[Nmarks=r, Nframe=y](q2)(45,55){$q_2$}
      \node[Nmarks=r, Nframe=y](q3)(10,20){$q_3$}
      \node[Nmarks=n, Nframe=y](q4)(45,20){$q_{\rej}$}   %{$\times$}
      \put(34,10){\makebox(0,0)[l]{$\times = \reject$}}
      
      %\put(45,40){\makebox(0,0)[l]{$q_1\text{: equal histories, no } c \text{ occurred}$}}
      %\put(45,35){\makebox(0,0)[l]{$q_2\text{: equal histories, } c \text{ did occur}$}}
      %\put(45,30){\makebox(0,0)[l]{$q_3\text{: different histories, no } c \text{ occurred}$}}
      %\put(45,25){\makebox(0,0)[l]{$\times\text{: different histories, } c \text{ occurred}$}}

      \drawedge[ELpos=50, ELside=l, curvedepth=-3](q1,q2){$\fbi{c}{c}$}
      \drawedge[ELpos=50, ELside=r, ELdist=1, curvedepth=-6](q2,q1){$\fbi{\#}{\#}$}

      \drawedge[ELpos=52, ELside=l, ELdist=.5, curvedepth=0](q1,q3){\rotatebox{-90}{$\fbi{a}{b}, \fbi{b}{a}, \fbi{a}{\#}, \fbi{\#}{a}, \fbi{b}{\#}, \fbi{\#}{b}$}}
      %\drawedge[ELpos=50, ELside=r, ELdist=-2.5, curvedepth=0](q1,q3){$\fbi{a}{b}\ \fbi{b}{a}$}
      %\drawedge[ELpos=48, ELside=l, ELdist=.5, curvedepth=0](q1,q3){$\fbi{b}{a}$}
      \drawedge[ELpos=50, ELside=l, ELdist=1, curvedepth=6](q3,q1){$\fbi{\#}{\#}$}

      \drawedge[ELpos=50, ELside=l, ELdist=-7, curvedepth=0](q1,q4){\rotatebox{-45}{$\cfbi{a}{c},\fbi{b}{c},\fbi{\#}{c},\fbi{c}{a},\fbi{c}{b},\fbi{c}{\#}$}}
      \drawedge[ELpos=50, ELside=l, ELdist=1, curvedepth=0](q2,q4){$\neq$}
      \drawedge[ELpos=50, ELside=l, curvedepth=0](q3,q4){$\cfbi{c}{*},\cfbi{*}{c}$}
      %\drawedge[ELpos=50, ELside=r, curvedepth=0](q3,q4){$\fbi{a}{d}, \fbi{d}{a}, \fbi{b}{d}, \fbi{d}{b}$}

      \drawloop[ELside=l,loopCW=y, loopdiam=5, loopangle=90](q1){$\cfbi{a}{a},\fbi{b}{b},\fbi{\#}{\#}$}
      \drawloop[ELside=l,loopCW=y, ELdist=1, loopdiam=5, loopangle=90](q2){$\cfbi{a}{a},\fbi{b}{b},\fbi{c}{c}$}   %$\cfbi{a}{a},\cfbi{b}{b},\cfbi{c}{c}$}
      \drawloop[ELpos=40, ELdist=-4.5, ELside=l,loopCW=y, loopdiam=5, loopangle=270](q3){$\cfbi{a}{a},\fbi{b}{b},\fbi{a}{b}, \fbi{b}{a}, \fbi{a}{\#}, \fbi{\#}{a}, \fbi{b}{\#}, \fbi{\#}{b}$}
      \drawloop[ELside=l,loopCW=y, loopdiam=5, loopangle=0](q4){$\fbi{*}{*}$}

%\node[Nmarks=n](n1)(30,60){$\ver_1$}
%\node[Nmarks=n, Nmr=0](n11)(15,45){$\ver_2$}
%\nodelabel[ExtNL=y, NLangle=270, NLdist=1](n111){{\small val=1}}

%\drawedge[ELpos=50, ELside=r, curvedepth=0, AHLength = 2, AHangle=25, AHlength = 1.81](n1,n11){$-1$}
%\drawedge[ELpos=50, ELside=r, curvedepth=0, AHnb=0, linewidth=.6](n1,n11dummy){}

%\drawedge[ELpos=50, ELside=l, curvedepth=-12](t3,n1){}
%\drawbpedge(t3,70,22,n1,110,22){}
%\drawline[AHnb=1,arcradius=1](113,17.5)(113,29)(5,29)(5,17.5)

%\drawloop[ELside=l,loopCW=y, loopdiam=6](n4){$1$}

%\drawloop[ELside=l,loopCW=y](nk){$0,1$}

%\drawedge[dash={1}0](n3bis,nkbis){$0,1$}

\end{picture}
%}
\hrule
\caption{A two-tape \dfa defining an indistinguishability relation that cannot be defined by any FIP (the symbol $\neq$ stands for $\{ \tbi{x}{y} \in \Gamma \times \Gamma \mid x \neq y\}$ and the symbol $*$ stands for $\{a,b,c,\#\}$).\label{fig:infinite}}
\end{center}
\end{figure}
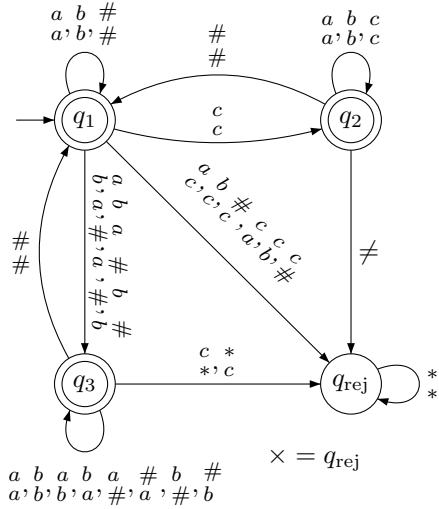

An important consequence of this construction is that the synthesis problem
for reachability games with imperfect information defined by a two-tape automaton
is nonelementary-hard. This follows from the hardness result of Theorem~\ref{theo:FIP-synthesis-complexity},
given the size of the automaton $\A_{{\textsf{FIP}}}$ is only exponential in the
number of observers of the FIP. 
Theorem~\ref{theo:synthesis-2dfa} suggests that in order to show a decidability 
result, a fairly complex construction will be necessary.

\begin{theorem}\label{theo:synthesis-2dfa}
The synthesis problem for reachability games with imperfect information 
defined by a two-tape automaton is nonelementary-hard.
\end{theorem}

%\begin{longversion}
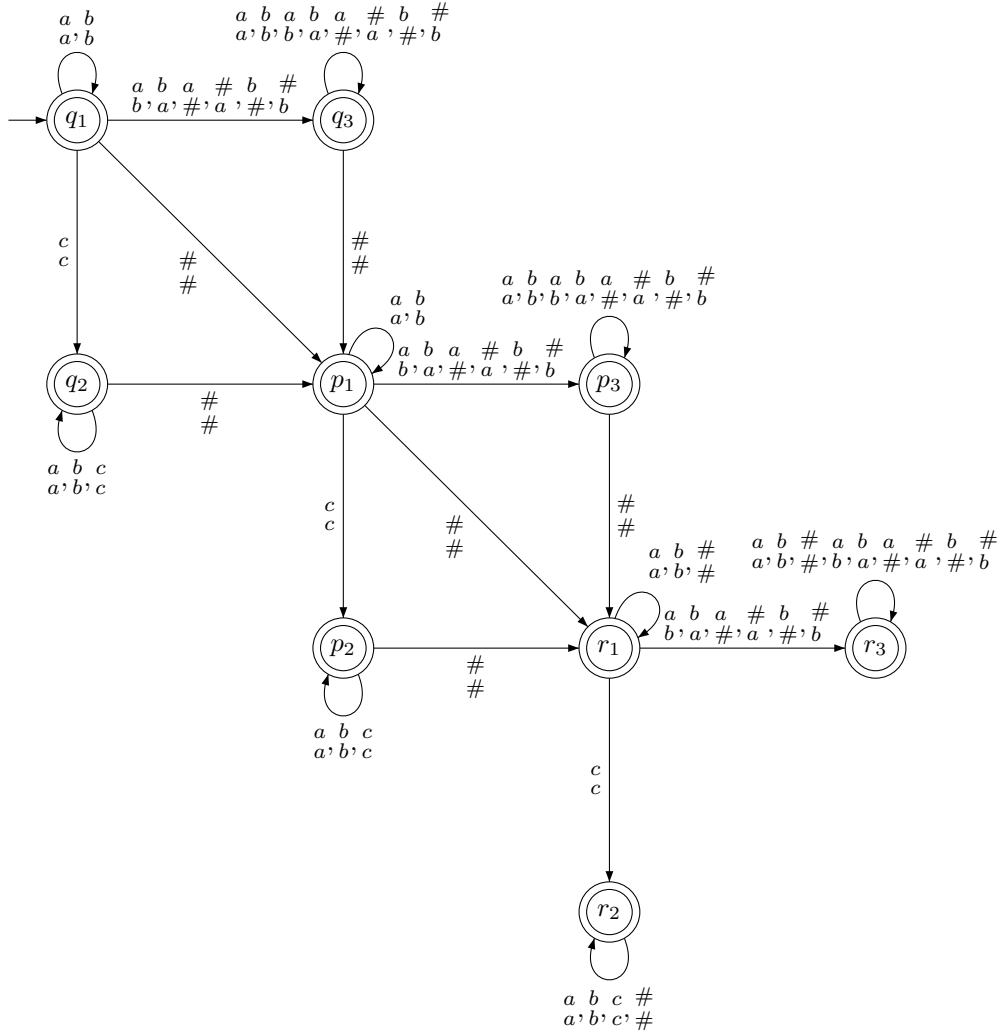
\begin{figure*}[t]
\begin{center}
\hrule
%\renewcommand{\sb}[1]{\scalebox{0.75}[1]{#1}}

%{\scriptsize 
\begin{picture}(125,142)(0,0)
%\put(0,0){\framebox(50,65){}}

      \gasset{Nw=8,Nh=8,Nmr=4, rdist=1, loopdiam=5}      % , ELdist=0

      \node[Nmarks=ir, Nframe=y](q1)(10,125){$q_1$}
      \node[Nmarks=r, Nframe=y](q2)(10,90){$q_2$}
      \node[Nmarks=r, Nframe=y](q3)(45,125){$q_3$}
      %\node[Nmarks=n, Nframe=n, Nw=4, Nh=5](q4)(45,20){$\times$}
      %\put(34,10){\makebox(0,0)[l]{$\times = \reject$}}

      \node[Nmarks=r, Nframe=y](p1)(45,90){$p_1$}
      \node[Nmarks=r, Nframe=y](p2)(45,55){$p_2$}
      \node[Nmarks=r, Nframe=y](p3)(80,90){$p_3$}

      \node[Nmarks=r, Nframe=y](r1)(80,55){$r_1$}
      \node[Nmarks=r, Nframe=y](r2)(80,20){$r_2$}
      \node[Nmarks=r, Nframe=y](r3)(115,55){$r_3$}

      %\put(45,40){\makebox(0,0)[l]{$q_1\text{: equal histories, no } c \text{ occurred}$}}
      %\put(45,35){\makebox(0,0)[l]{$q_2\text{: equal histories, } c \text{ did occur}$}}
      %\put(45,30){\makebox(0,0)[l]{$q_3\text{: different histories, no } c \text{ occurred}$}}
      %\put(45,25){\makebox(0,0)[l]{$\times\text{: different histories, } c \text{ occurred}$}}

      \drawedge[ELpos=50, ELside=r, curvedepth=0](q1,p1){$\fbi{\#}{\#}$}
      \drawedge[ELpos=50, ELside=r, curvedepth=0](q1,q2){$\fbi{c}{c}$}
      \drawedge[ELpos=50, ELside=r, ELdist=1, curvedepth=0](q2,p1){$\fbi{\#}{\#}$}

      \drawedge[ELpos=52, ELside=l, ELdist=.5, curvedepth=0](q1,q3){$\fbi{a}{b}, \fbi{b}{a}, \fbi{a}{\#}, \fbi{\#}{a}, \fbi{b}{\#}, \fbi{\#}{b}$}
      %\drawedge[ELpos=50, ELside=r, ELdist=-2.5, curvedepth=0](q1,q3){$\fbi{a}{b}\ \fbi{b}{a}$}
      %\drawedge[ELpos=48, ELside=l, ELdist=.5, curvedepth=0](q1,q3){$\fbi{b}{a}$}
      \drawedge[ELpos=50, ELside=l, ELdist=1, curvedepth=0](q3,p1){$\fbi{\#}{\#}$}

      \drawloop[ELside=l,loopCW=y, loopdiam=5, loopangle=90](q1){$\cfbi{a}{a},\fbi{b}{b}$}
      \drawloop[ELside=l,loopCW=y, ELdist=1, loopdiam=5, loopangle=270](q2){$\cfbi{a}{a},\fbi{b}{b},\fbi{c}{c}$}   %$\cfbi{a}{a},\cfbi{b}{b},\cfbi{c}{c}$}
      \drawloop[ELpos=50, ELdist=1, ELside=l,loopCW=y, loopdiam=5, loopangle=90](q3){$\cfbi{a}{a},\fbi{b}{b},\fbi{a}{b}, \fbi{b}{a}, \fbi{a}{\#}, \fbi{\#}{a}, \fbi{b}{\#}, \fbi{\#}{b}$}

      \drawedge[ELpos=50, ELside=r, curvedepth=0](p1,r1){$\fbi{\#}{\#}$}
      \drawedge[ELpos=50, ELside=r, curvedepth=0](p1,p2){$\fbi{c}{c}$}
      \drawedge[ELpos=50, ELside=r, ELdist=1, curvedepth=0](p2,r1){$\fbi{\#}{\#}$}

      \drawedge[ELpos=52, ELside=l, ELdist=.5, curvedepth=0](p1,p3){$\fbi{a}{b}, \fbi{b}{a}, \fbi{a}{\#}, \fbi{\#}{a}, \fbi{b}{\#}, \fbi{\#}{b}$}
      %\drawedge[ELpos=50, ELside=r, ELdist=-2.5, curvedepth=0](q1,q3){$\fbi{a}{b}\ \fbi{b}{a}$}
      %\drawedge[ELpos=48, ELside=l, ELdist=.5, curvedepth=0](q1,q3){$\fbi{b}{a}$}
      \drawedge[ELpos=50, ELside=l, ELdist=1, curvedepth=0](p3,r1){$\fbi{\#}{\#}$}

      \drawloop[ELside=l,loopCW=y, loopdiam=5, loopangle=50](p1){$\cfbi{a}{a},\fbi{b}{b}$}
      \drawloop[ELside=l,loopCW=y, ELdist=1, loopdiam=5, loopangle=270](p2){$\cfbi{a}{a},\fbi{b}{b},\fbi{c}{c}$}   %$\cfbi{a}{a},\cfbi{b}{b},\cfbi{c}{c}$}
      \drawloop[ELpos=50, ELdist=1, ELside=l,loopCW=y, loopdiam=5, loopangle=90](p3){$\cfbi{a}{a},\fbi{b}{b},\fbi{a}{b}, \fbi{b}{a}, \fbi{a}{\#}, \fbi{\#}{a}, \fbi{b}{\#}, \fbi{\#}{b}$}

      \drawedge[ELpos=50, ELside=r, curvedepth=0](r1,r2){$\fbi{c}{c}$}

      \drawedge[ELpos=52, ELside=l, ELdist=.5, curvedepth=0](r1,r3){$\fbi{a}{b}, \fbi{b}{a}, \fbi{a}{\#}, \fbi{\#}{a}, \fbi{b}{\#}, \fbi{\#}{b}$}
      %\drawedge[ELpos=50, ELside=r, ELdist=-2.5, curvedepth=0](q1,q3){$\fbi{a}{b}\ \fbi{b}{a}$}
      %\drawedge[ELpos=48, ELside=l, ELdist=.5, curvedepth=0](q1,q3){$\fbi{b}{a}$}

      \drawloop[ELside=l,loopCW=y, loopdiam=5, loopangle=50](r1){$\cfbi{a}{a},\fbi{b}{b},\fbi{\#}{\#}$}
      \drawloop[ELside=l,loopCW=y, ELdist=1, loopdiam=5, loopangle=270](r2){$\cfbi{a}{a},\fbi{b}{b},\fbi{c}{c},\fbi{\#}{\#}$}   %$\cfbi{a}{a},\cfbi{b}{b},\cfbi{c}{c}$}
      \drawloop[ELpos=50, ELdist=1, ELside=l,loopCW=y, loopdiam=5, loopangle=90](r3){$\cfbi{a}{a},\fbi{b}{b},\fbi{\#}{\#}, \fbi{a}{b}, \fbi{b}{a}, \fbi{a}{\#}, \fbi{\#}{a}, \fbi{b}{\#}, \fbi{\#}{b}$}

%\node[Nmarks=n](n1)(30,60){$\ver_1$}
%\node[Nmarks=n, Nmr=0](n11)(15,45){$\ver_2$}
%\nodelabel[ExtNL=y, NLangle=270, NLdist=1](n111){{\small val=1}}

%\drawedge[ELpos=50, ELside=r, curvedepth=0, AHLength = 2, AHangle=25, AHlength = 1.81](n1,n11){$-1$}
%\drawedge[ELpos=50, ELside=r, curvedepth=0, AHnb=0, linewidth=.6](n1,n11dummy){}

%\drawedge[ELpos=50, ELside=l, curvedepth=-12](t3,n1){}
%\drawbpedge(t3,70,22,n1,110,22){}
%\drawline[AHnb=1,arcradius=1](113,17.5)(113,29)(5,29)(5,17.5)

%\drawloop[ELside=l,loopCW=y, loopdiam=6](n4){$1$}

%\drawloop[ELside=l,loopCW=y](nk){$0,1$}

%\drawedge[dash={1}0](n3bis,nkbis){$0,1$}

\end{picture}
%}
\hrule
\caption{A two-tape \dfa defining an indistinguishability relation that can be defined by a FIP with four players, but not by FIP with three players. Missing transitions are directed to a sink rejecting state. \label{fig:strict-hierarchy}}
\end{center}
\end{figure*}
%\end{longversion}

We now show that two-tape automata are strictly more expressive than FIP.
Consider the move alphabet $\Gamma = \{a,b,c,\#\}$ where $\#$ is used as a separator,
and let two histories $\tau,\tau' \in \Gamma^*$ be indistinguishable if their suffix 
after the last position where they both contain a separator~$\#$ (or from the initial position 
if no such position exists) are equal, or none of them contains the letter~$c$.
Intuitively, along a history the symbol $\#$ separates block of letters over $\{a,b,c\}$.
Within a block the letters $a$ and $b$ are indistinguishable until a letter $c$
occurs, which reveals the current block (similar to the example of \figurename~\ref{fig:simple-fip}).
Note that the letters $a$ and $b$ in all previous blocks remain indistinguishable
forever.

This indistinguishability relation is defined by the two-tape \dfa in
\figurename~\ref{fig:infinite}. Intuitively, it cannot be defined by a FIP because
whenever a letter $c$ occurs in a history, Player~$0$ would need
to communicate with some observer who can see the sequence of $a$'s and $b$'s
in the current block (Player~$0$'s own observations are not 
sufficient to define the indistinguishability relation, as it has unbounded branching).
However, we can never reuse the same observer for the next block because 
communicating with such an observer would reveal information to which  
Player~$0$ does not have access. 
We need a fresh observer for each block, and since an history may contain an 
arbitrarily large number of blocks, a finite number of observers would not be
sufficient.

\begin{lemma}\label{lem:dfa-to-FIP}
There exists an indistinguishability relation defined by a two-tape \dfa that 
cannot be defined by any FIP (no matter the number of players).
\end{lemma}

\begin{longversion}
\begin{proof}
Consider the two-tape \dfa in \figurename~\ref{fig:infinite}.
We note that, for all $n\in \nat$, the set of words $\{a,b\}^n$ is such that 
for all histories $\tau \in \Gamma^*$,
if $(\tau,\tau)$ leads to the initial state $q_1$, 
then the words in the set $u_{\tau} = \tau\{a,b\}^n$ are pairwise indistinguishable.
Moreover, if a letter $c$ occurs, the words $\tau\{a,b\}^nc$ become all pairwise distinguishable.
%We call such a set of histories an ambiguous clique.
%We can also observed that every pair of different histories in this set leads to the state $q_3$.
%Thus we can speak about a $q_3$-ambiguous clique.

Towards contradiction, assume that there exists a FIP $F$ that defines $\sim$ 
(i.e., such that $\sim_0\, = \,\sim$). Let $N$ be the number of players in $F$.

For every history $\tau \in \Gamma^*$, consider the communication set 
$\Com(\tau) \subseteq I$ containing all players with which Player~$0$ may 
communicate (directly or indirectly) along continuations $\tau w$ for all $w \in \Gamma^*$ 
(which is tedious to define formally).

We construct a sequence $\tau_0,\tau_1, \dots, \tau_N$ of histories $\tau_n \in \Gamma^*$ 
such that for all $n \geq 1$:
\begin{itemize}
\item the communication sets are strictly decreasing, $\Com(\tau_{n}) \subsetneq \Com(\tau_{n-1})$, and
\item the communication sets are nonempty, $\Com(\tau_{n}) \neq \emptyset$.
\end{itemize}

Since the size of the communication sets is bounded by the total number $N$ of players, 
this implies a contradiction for $\tau_N$. 

We now show how to construct the histories $\tau_n$. 
First, as an intermediate proposition, we show that if $(\tau_n,\tau_n)$ 
leads to the initial state $q_1$ in the two-tape automaton,
then $\Com(\tau_{n}) \neq \emptyset$. 
For $k > \abs{\Sigma}^N$, consider the histories of the form $\tau_n \{a,b\}^k c$.
Since there are $\abs{\Sigma}^N$ tuples of observations for $N$ players, 
by the pigeonhole principle there exist two sequences $w,w'\in \{a,b\}^k$ such 
that the two histories $\tau_n w c$ and $\tau_n w' c$ have the 
same (last) observation for all players, $\beta_i(\tau_n w c) = \beta_i(\tau_n w' c)$ 
for all players $i\in I$.
In particular, as $\tau_n w \sim_0 \tau_n w'$ are indistinguishable histories for 
Player~$0$, we have $\hat{\beta}_0(\tau_n w c) = \hat{\beta}_0(\tau_n w' c)$.
Hence % by Lemma~\ref{lem:BranchingLemma}, 
upon reading the last $c$, there 
must be a (direct or indirect) communication between Player~$0$ and some 
Player $i \in I$ who distinguishes the two histories, $\tau_n wc \nsim_i \tau_n w'c$.
It follows by definition of $\Com(\tau_n)$ that $i \in Com(\tau_n)$, thus $\Com(\tau_n) \neq \emptyset$.
Since all players have the same (last) observation on $\tau_n wc$ and $\tau_n w'c$, 
we can assume w.l.o.g. that the distinction occurred earlier, $\tau_n w \nsim_i \tau_n w'$.

We now present the construction. 
Let $\tau_0 = \epsilon$, and construct $\tau_{n+1}$ from $\tau_n$ such that $
\Com(\tau_{n+1}) \subsetneq \Com(\tau_n)$ and $(\tau_{n+1},\tau_{n+1})$ leads 
to the initial state, inductively as follows.

Consider the history $\tau_{n+1} = \tau_n w d$ and note that for all continuations $z \in \Gamma^*$, 
the histories $\tau_{n+1} z \sim \tau_n w' d z$ are indistinguishable 
because the pair $(\tau_n w d,\tau_n w' d)$ leads to the state $q_1$.

As we know that the histories $\tau_n w \nsim_i \tau_n w'$ are distinguishable 
for Player $i$, so are the histories $\tau_n w d z \nsim_i \tau_n w' d z$
for all continuations $z \in \Gamma^*$.
Since $\tau_n w d z \sim \tau_n w' d z$ for all $z \in \Gamma^*$, 
Player~$0$ cannot communicate with Player $i$ after the history $\tau_{n+1} = \tau_n w d$
(as otherwise it would let Player~$0$ distinguish indistinguishable histories).
Hence $i \notin \Com(\tau_{n+1})$. % and thus $i \notin \Com(\tau_n w d z)$ for all $z \in \Gamma^*$. 
Since $\tau_n$ is a prefix of $\tau_{n+1}$, we have 
$\Com(\tau_{n+1} \subseteq \Com(\tau_n)$ and since $i \in \Com(\tau_n)$
we conclude that $\Com(\tau_{n+1} \subsetneq \Com(\tau_n)$ and it is easy to 
check that $(\tau_{n+1},\tau_{n+1})$ leads to the initial state $q_1$ as required.
\end{proof}
\end{longversion}

The same idea can be used to show that increasing the number of players in FIP
increases the expressive power, that is for all $n\geq 2$, there exists 
an indistinguishability relation that can be defined by a FIP with $n$ players
but not by any FIP with $n-1$ players. \figurename~\ref{fig:strict-hierarchy} 
shows a two-tape \dfa that defines such an indistinguishability relation for $n=4$.
Intuitively it is obtained by ``unfolding'' the automaton of \figurename~\ref{fig:infinite}
into $n-1$ copies, redirecting the transitions on $(\#,\#)$ to the next copy 
of the automaton, except in the last copy where the transitions on  $(\#,\#)$
are self-loops. The reader can verify that Player~$0$ needs three observers
to track the moves in the first (at most) three blocks separated by $\#$ along
a history.

Given a move alphabet $\Gamma$, denote by $\F_n$ the class of indistinguishability relations
definable by a FIP with $n$ observers. 

\begin{theorem}\label{theo:FIP-expressive-power}
  The hierarchy of indistinguishability relation classes~$\F_n$ induced by a FIP with~$n$ observers is strict and does not exhaust the class of $\F_{{\textsf{2DFA}}}$ of regular indistinguishability relations:
%\begin{linenomath*}
$$\F_1 \subsetneq \F_2 \subsetneq \ldots \subsetneq \F_n \subsetneq \ldots \subsetneq\F_{\textsf{2DFA}}.$$.
%\end{linenomath*}
\end{theorem}

\begin{longversion}
\begin{proof}
\figurename~\ref{fig:strict-hierarchy} illustrates the construction of a 
witness relation $\sim^n$ such that $\sim^n \in \F_{n+1}$ and $\sim^n \not\in \F_n$.

The proof that $\sim^n \not\in \F_n$ follows the same line as the proof of Lemma~\ref{lem:dfa-to-FIP}, and the proof that $\sim^n \in \F_{n+1}$
is straightforward.
\end{proof}
\end{longversion}

\bibliography{biblio} 

\newcommand{\etalchar}[1]{$^{#1}$}
\begin{thebibliography}{GGMW13}

\bibitem[APR01]{AzharPetRei01}
Salman Azhar, Gary Peterson, and John Reif.
\newblock Lower bounds for multiplayer non-cooperative games of incomplete
  information.
\newblock {\em Journal of Computers and Mathematics with Applications},
  41:957--992, 2001.

\bibitem[AVW03]{ArnoldWal03}
Andr{\'e} Arnold, Aymeric Vincent, and Igor Walukiewicz.
\newblock Games for synthesis of controllers with partial observation.
\newblock {\em Theoretical computer science}, 303(1):7--34, 2003.

\bibitem[BD23]{BD23}
D.~Berwanger and L.~Doyen.
\newblock Observation and distinction. representing information in infinite
  games.
\newblock {\em Theory of Computing Systems}, 67(1):4--27, 2023.

\bibitem[BK08]{BK08}
C.~Baier and J.-P. Katoen.
\newblock {\em Principles of Model Checking}.
\newblock MIT, 2008.

\bibitem[BKP11]{BerwangerKP11}
Dietmar Berwanger, Lukasz Kaiser, and Bernd Puchala.
\newblock A perfect-information construction for coordination in games.
\newblock In {\em IARCS Annual Conference on Foundations of Software Technology
  and Theoretical Computer Science, FSTTCS 2011, December 12-14, 2011, Mumbai,
  India}, volume~13 of {\em LIPIcs}, pages 387--398. Schloss Dagstuhl -
  Leibniz-Zentrum fuer Informatik, 2011.

\bibitem[BL69]{BL69}
J.~R. B\"uchi and L.~H. Landweber.
\newblock Solving sequential conditions by finite-state strategies.
\newblock {\em Transactions of the American Mathematical Society},
  138:295--311, 1969.

\bibitem[BMvdB18]{BMvdB18}
D.~Berwanger, A.~B. Mathew, and M.~van~den Bogaard.
\newblock Hierarchical information and the synthesis of distributed strategies.
\newblock {\em Acta Informatica}, 55(8):669--701, 2018.

\bibitem[CDHR07]{CDHR07}
K.~Chatterjee, L.~Doyen, T.~A. Henzinger, and J.-F. Raskin.
\newblock Algorithms for omega-regular games of incomplete information.
\newblock {\em Logical Methods in Computer Science}, 3(3:4), 2007.

\bibitem[Chu62]{Chu62}
A.~Church.
\newblock Logic, arithmetics, and automata.
\newblock {\em Proc. Int. Congr. Math.}, pages 23--35, 1962.

\bibitem[CJK{\etalchar{+}}22]{CaludeJKLS22}
Cristian~S. Calude, Sanjay Jain, Bakhadyr Khoussainov, Wei Li, and Frank
  Stephan.
\newblock Deciding parity games in quasi-polynomial time.
\newblock {\em {SIAM} J. Comput.}, 51(2):17--152, 2022.

\bibitem[DM90]{DworkMos90}
Cynthia Dwork and Yoram Moses.
\newblock Knowledge and common knowledge in a byzantine environment: Crash
  failures.
\newblock {\em Inf. Comput.}, 88(2):156--186, 1990.

\bibitem[DR11]{DR11}
L.~Doyen and J.-F. Raskin.
\newblock Games with imperfect information: Theory and algorithms.
\newblock In {\em Lectures in Game Theory for Computer Scientists}, pages
  185--212. Cambridge University Press, 2011.

\bibitem[FO17]{FO17}
B.~Finkbeiner and E.{-}R. Olderog.
\newblock Petri games: Synthesis of distributed systems with causal memory.
\newblock {\em Inf. Comput.}, 253:181--203, 2017.

\bibitem[FS05]{FinkbeinerS05}
Bernd Finkbeiner and Sven Schewe.
\newblock Uniform distributed synthesis.
\newblock In {\em 20th {IEEE} Symposium on Logic in Computer Science {(LICS}
  2005), 26-29 June 2005, Chicago, IL, USA, Proceedings}, pages 321--330.
  {IEEE} Computer Society, 2005.

\bibitem[GGMW13]{GGMW13}
B.~Genest, H.~Gimbert, A.~Muscholl, and I.~Walukiewicz.
\newblock Asynchronous games over tree architectures.
\newblock In {\em Proc. of ICALP: Automata, Languages, and Programming}, LNCS
  7966, pages 275--286. Springer, 2013.

\bibitem[GH82]{GurevichHar82}
Yuri Gurevich and Leo Harrington.
\newblock {T}rees, automata, and games.
\newblock In {\em Proceedings of the fourteenth annual ACM symposium on theory
  of computing}, pages 60--65, 1982.

\bibitem[GLZ04]{GLZ04}
P.~Gastin, B.~Lerman, and M.~Zeitoun.
\newblock Distributed games and distributed control for asynchronous systems.
\newblock In {\em Proc. of {LATIN}: Theoretical Informatics, 6th Latin American
  Symposium}, LNCS 2976, pages 455--465. Springer, 2004.

\bibitem[GTW02]{automata}
E.~Gr{\"a}del, W.~Thomas, and T.~Wilke, editors.
\newblock {\em Automata, Logics, and Infinite Games: A Guide to Current
  Research}, LNCS 2500. Springer, 2002.

\bibitem[HP85]{HarelPnu85}
D.~Harel and A.~Pnueli.
\newblock On the development of reactive systems.
\newblock In Krzysztof~R. Apt, editor, {\em Logics and Models of Concurrent
  Systems}, pages 477--498, {Berlin, Heidelberg}, 1985. {Springer Berlin
  Heidelberg}.

\bibitem[JMT22]{JMT22}
M.~Jurdzinski, R.~Morvan, and K.~S. Thejaswini.
\newblock Universal algorithms for parity games and nested fixpoints.
\newblock In {\em Principles of Systems Design - Essays Dedicated to Thomas A.
  Henzinger on the Occasion of His 60th Birthday}, LNCS 13660, pages 252--271.
  Springer, 2022.

\bibitem[K{\"o}n36]{Konig36}
D.~K{\"o}nig.
\newblock {\em Theorie der endlichen und unendlichen Graphen}.
\newblock Akademische Verlagsgesellschaft, Leipzig, 1936.

\bibitem[KV01]{KupfermanVar01}
Orna Kupferman and Moshe~Y. Vardi.
\newblock Synthesizing distributed systems.
\newblock In {\em Proc. of LICS~'01}, pages 389--398. IEEE Computer Society
  Press, June 2001.

\bibitem[MTY05]{MTY05}
P.~Madhusudan, P.~S. Thiagarajan, and S.~Yang.
\newblock The {MSO} theory of connectedly communicating processes.
\newblock In {\em Proc. of {FSTTCS}: Foundations of Software Technology and
  Theoretical Computer Science}, LNCS 3821, pages 201--212. Springer, 2005.

\bibitem[MW03]{MohalikWal03}
Swarup Mohalik and Igor Walukiewicz.
\newblock {D}istributed {G}ames.
\newblock In {\em FSTTCS'03}, volume 2914 of {\em LNCS}, pages 338--351, 2003.

\bibitem[PR89]{PnueliRos89}
A.~Pnueli and E.~Rosner.
\newblock {O}n the synthesis of a reactive module.
\newblock In {\em Proceedings of the 16th ACM SIGPLAN-SIGACT symposium on
  Principles of programming languages}, pages 179 -- 190. ACM Press, 1989.

\bibitem[PR90]{PR90}
A.~Pnueli and R.~Rosner.
\newblock Distributed reactive systems are hard to synthesize.
\newblock In {\em Proc. of FOCS: Foundations of Computer Science}, pages
  746--757. {IEEE}, 1990.

\bibitem[PSL80]{PeaseSL80}
M.~Pease, R.~Shostak, and L.~Lamport.
\newblock Reaching agreements in the presence of faults.
\newblock {\em Journal of the ACM}, 27(2):228--234, April 1980.

\bibitem[Rab69]{Rabin69}
M.~O. Rabin.
\newblock Decidability of second-order theories and automata on infinite trees.
\newblock {\em Transactions of the AMS}, 141:1--35, 1969.

\bibitem[Rab72]{Rabin72}
M.~O. Rabin.
\newblock {\em Automata on Infinite Objects and {C}hurch's Problem}.
\newblock American Mathematical Society, Boston, MA, USA, 1972.

\bibitem[Rei84]{Rei84}
John~H. Reif.
\newblock The complexity of two-player games of incomplete information.
\newblock {\em Journal of Computer and System Sciences}, 29(2):274--301, 1984.

\bibitem[RW87]{RamadgeWon87}
P.~J. Ramadge and W.~M. Wonham.
\newblock {S}upervisory control of a class of discrete event processes.
\newblock {\em SIAM J. Control Optim.}, 25(1):206--230, 1987.

\bibitem[San11]{San11}
D.~Sangiorgi.
\newblock {\em Introduction to Bisimulation and Coinduction}.
\newblock Cambridge University Press, 2011.

\bibitem[Sch14]{Schewe2014}
Sven Schewe.
\newblock Distributed synthesis is simply undecidable.
\newblock {\em Inf. Process. Lett.}, 114(4):203--207, April 2014.

\bibitem[Tho95]{Thomas95}
W.~Thomas.
\newblock On the synthesis of strategies in infinite games.
\newblock In {\em Proc. of STACS: Symposium on Theoretical Aspects of Computer
  Science}, LNCS 900. Springer, 1995.

\bibitem[Tho97]{Thomas97}
W.~Thomas.
\newblock Languages, automata, and logic.
\newblock In {\em Handbook of Formal Languages}, volume 3, Beyond Words,
  chapter~7, pages 389--455. Springer, 1997.

\bibitem[WL94]{WooLam94}
Thomas~YC Woo and Simon~S Lam.
\newblock A lesson on authentication protocol design.
\newblock {\em ACM SIGOPS Operating Systems Review}, 28(3):24--37, 1994.

\bibitem[Zie87]{Zie87}
W.~Zielonka.
\newblock Notes on finite asynchronous automata.
\newblock {\em {RAIRO} Theor. Informatics Appl.}, 21(2):99--135, 1987.

\end{thebibliography}

\bibliographystyle{alpha}

\end{document}